\documentclass[final]{IEEEtran}

\usepackage{cite}
\usepackage{color}
\usepackage{threeparttable}
\usepackage{graphicx}
\usepackage{picinpar}
\usepackage[cmex10]{amsmath}
\usepackage{amsmath,amsfonts,amssymb}
\usepackage{subfigure}
\usepackage{changepage}
\usepackage{algorithm}
\usepackage{algpseudocode}
\usepackage{stfloats}
\usepackage{bm}
\usepackage{amsthm}
\usepackage{enumerate}
\usepackage{multirow}
\usepackage{booktabs}
\usepackage{cuted}%%\stripsep-3pt

\allowdisplaybreaks
\begin{document}
\newtheorem{lemma}{Lemma}
\newtheorem{corollary}{Corollary}
\newtheorem{theorem}{Theorem}
\newtheorem{proposition}{Proposition}
\newtheorem{definition}{Definition}
\newcommand{\e}{\begin{equation}}
\newcommand{\ee}{\end{equation}}
\newcommand{\eqn}{\begin{eqnarray}}
\newcommand{\eeqn}{\end{eqnarray}}
%  调整公式间距
\newenvironment{shrinkeq}[1]
{ \bgroup
\addtolength\abovedisplayshortskip{#1}
\addtolength\abovedisplayskip{#1}
\addtolength\belowdisplayshortskip{#1}
\addtolength\belowdisplayskip{#1}}
{\egroup\ignorespacesafterend}
\title{Terahertz Massive MIMO with Holographic Reconfigurable Intelligent Surfaces}

\author{Ziwei Wan, Zhen Gao, \IEEEmembership{Member,~IEEE},
Feifei Gao, \IEEEmembership{Fellow,~IEEE}, \\
Marco Di Renzo, \IEEEmembership{Fellow,~IEEE}, and Mohamed-Slim Alouini, \IEEEmembership{Fellow,~IEEE}
\thanks{The codes and some other associated materials of this work may be available at https://gaozhen16.github.io.}
\thanks{Z. Wan, and Z. Gao are with the School of Information and Electronics, Beijing Institute of Technology, Beijing 100081, China, and also with the Advanced Research Institute of Multidisciplinary Science, Beijing Institute of Technology, Beijing 100081, China (e-mail: \{ziweiwan, gaozhen16\}@bit.edu.cn).}
\thanks{F. Gao is with the Institute for Artificial Intelligence, Tsinghua University (THUAI), Beijing 100084, China, also with the State Key Laboratory of Intelligent Technologies and Systems, Tsinghua University, Beijing 100084, China, and also with the Beijing National Research Center for Information Science and Technology (BNRist), Department of Automation, Tsinghua University, Beijing 100084, China (e-mail: feifeigao@ieee.org).}
\thanks{M. Di Renzo is with Universit\'e Paris-Saclay, CNRS, CentraleSup\'elec, Laboratoire des Signaux et Syst\`emes, 3 Rue Joliot-Curie, 91192 Gif-sur-Yvette, France(e-mail: marco.di-renzo@universite-paris-saclay.fr).}
\thanks{M.-S. Alouini is with King Abdullah University of Science and Technology (KAUST), Thuwal, Makkah Province, Saudi Arabia (e-mail: slim.alouini@kaust.edu.sa).}
}
\maketitle

\begin{abstract}
We propose a holographic version of a reconfigurable intelligent surface 
(RIS) and investigate its application to terahertz (THz) massive multiple-input multiple-output systems.
Capitalizing on the miniaturization of THz electronic components, RISs can be implemented by densely packing sub-wavelength unit cells, so as to realize continuous or quasi-continuous apertures and to enable {\it holographic communications}. 
In this paper, in particular, we derive the beam pattern of a holographic RIS.
Our analysis reveals that the beam pattern of an ideal holographic RIS can be well approximated by that of an ultra-dense RIS, which has a more practical hardware architecture.
In addition, we propose a closed-loop channel estimation (CE) scheme to effectively estimate the broadband channels that characterize THz massive MIMO systems aided by holographic RISs. The proposed CE scheme includes a downlink coarse CE stage and an uplink finer-grained CE stage. The uplink pilot signals are judiciously designed for obtaining good CE performance.
Moreover, to reduce the pilot overhead, we introduce a compressive sensing-based CE algorithm, which exploits the dual sparsity of THz MIMO channels in both the angular domain and delay domain. Simulation results demonstrate the superiority of holographic RISs over the non-holographic ones, and the effectiveness of the proposed CE scheme.
\end{abstract}


\IEEEpeerreviewmaketitle

\section{Introduction}
Over the past few years, the demand for wireless data traffic has increased significantly due to the explosive growth of mobile devices and multimedia applications \cite{THz1,THz2}. To accommodate these demands, the possible use of the terahertz (THz) band has attracted great interest from both industry and academia \cite{THz1,THz2,DifSca,MM,UMM2,THzModel1,THzModel2}. The THz band can provide more abundant bandwidth (from 0.1 THz to 10 THz), higher data rates (from tens of Gbps to several Tbps), and lower latency (of the order of micro-seconds \cite{THz1}), as compared with the millimeter-wave (mmWave) band. The THz band is considered to be a promising candidate to enable beyond 5G and 6G communications.

In spite of the appealing advantages of the THz band, establishing a reliable transmission link at THz frequencies is a non-trivial task. This is because {(i)} there exist strong atmospheric attenuation and extremely high free-space losses in the THz band; and {(ii)} the line-of-sight (LoS) link is very sensitive to blockage effects in the THz band and thus the links are usually intermittent. These disadvantages may negatively affect the communication range and may severely degrade the service coverage of THz communication systems. The deployment of massive \cite{MM} or even ultra-massive \cite{UMM2} multiple-input multiple-output (MIMO) systems in the context of THz communications may provide considerable beamforming gain in order to overcome the mentioned limitations, but it may result in an unaffordable power consumption and may put an overweight burden on the overall communication system design.

{Recently, the emerging technology of reconfigurable intelligent surface (RIS) has been proposed and applied to wireless communications in order to enhance the communication performance in various scenarios, such as MIMO communications, physical layer security, unmanned aerial vehicle communications, simultaneous wireless information and power transfer, cognitive radio systems \cite{CR,Model2,Y,Reviewer,THzIRS2,
UAV,UAV2,CHuang,Model,SWIPT,B,C,D,MDR,
THzIRS,Europe,Fellow3,RZhang1,RZhang2,MSlim,FFFCE,FSFCE,AA,XYuan,My1,SLiu}.} Made of passive and metamaterial-based reconfigurable elements, an RIS can manipulate both the phase and amplitude of the incident electromagnetic (EM) signals so as to reflect them towards the desired directions. More importantly, unlike other transmission technologies such as active relay \cite{B}, an RIS does not need power-hungry radio frequency chains (RFCs) and power amplifiers, which may be beneficial for developing green and cost-efficient communications. Although the application of RISs to the mmWave band has been investigated recently \cite{AA,XYuan,My1,SLiu}, the utilization of RIS for THz communications is still at its infancy. 

\vspace{-3mm}
\subsection{Prior Work}
RIS-aided MIMO communications have attracted lots of research interest lately. The authors of \cite{RZhang2} propose a joint active and passive beamforming scheme based on convex optimization to maximize the signal-to-interference-plus-noise ratio (SINR) at the receivers.
A similar scenario is considered in \cite{MSlim}, where the phases of the RIS that maximize the SINR are computed via the projected gradient ascent algorithm.
%In \cite{KYing}, broadband beamforming for RIS-aided hybrid mmWave MIMO systems is investigated.
%A geometric mean decomposition (GMD)-based beamforming scheme is applied to achieve better bit-error-rate (BER) performance than conventional singular value decomposition (SVD)-based beamforming.
In \cite{CHuang}, the advanced deep reinforcement learning (DRL) method is considered for the beamforming design in RIS-aided multiuser systems.
{The authors of \cite{Reviewer} develop a scalable optimization framework for large RISs, based on which a two-stage method, including an offline design and an online optimization, is proposed to optimize the RIS.}
{A practical amplitude and phase shift model for the reflecting element of an RIS is investigated in \cite{Model}, where the beamforming optimization is conducted by considering the phase-dependent amplitude variation of each reflecting element.}
{In \cite{Model2}, the authors introduce a new communication model for RISs that accounts for the mutual coupling and the interplay between the amplitude and phase response of the scattering elements of the RIS.}

The designs in \cite{Reviewer,Model,Model2,RZhang2,MSlim,CHuang} rely on the knowledge of global channel state information (CSI).
Since no active elements are used in RISs, channel estimation (CE) is a challenging task and an essential prerequisite in RIS-aided MIMO systems.
In \cite{FFFCE} and \cite{FSFCE}, the authors propose a CE algorithm for RIS-aided systems, based on the least square (LS) estimator, for application to frequency-flat fading channels and frequency-selective fading channels, respectively.
To further facilitate the CE task when large-size arrays are deployed, compressive sensing (CS)-based CE schemes are investigated in \cite{XYuan,My1}, where the inherent sparsity of mmWave channels in the angular domain is exploited in order to reduce the pilot overhead.
In \cite{AA} and \cite{SLiu}, a new architecture of RIS for application to mmWave band is proposed, where a few active RFCs are available at the RIS in order to learn the channel in real time. Based on this architecture, deep learning (DL) techniques are adopted for CE and beamforming.

As far as the application of RISs to THz communications is concerned, the authors of \cite{THzIRS} present an RIS-aided THz MIMO communication system for indoor applications. Joint CE and data rate maximization schemes are proposed in \cite{THzIRS} based on both CS and DL techniques.
{In \cite{THzIRS2}, the problem of RIS-assisted secure transmission in the THz band is investigated. The phase shifts at the RIS that maximize the secrecy rate are obtained based on convex optimization.}
Based on the current state of research, we evince that the design of RIS-aided THz communication systems is still at an early development stage.

Recently, the concept of {\it holographic communication} has been proposed as a new paradigm shift in MIMO \cite{Holo1,Holo2,Holo3} and RIS-aided \cite{Holo4} communications. One of the main features of holographic communications is the integration of very large numbers of tiny and inexpensive antennas or reconfigurable elements into a compact space in order to realize a holographic array with a spatially continuous aperture \cite{Holo1,Holo2,Holo3,Holo4}.
{This holographic architecture is easier to realize in the THz band thanks to the miniaturization of THz electronic components.
In \cite{THzIRS}, for example, the size of each graphene-based reflecting element is $200$ $\mu$m $\times$ $190$ $\mu$m at a carrier frequency of $0.22$ THz that corresponds to a wavelength $\lambda \approx 1360$ $\mu$m.
Therefore, the reflecting elements can be spaced more densely than $\lambda/2$
%, as done in previous works \cite{Model,THzIRS,RZhang2,MSlim,CHuang,FFFCE,FSFCE,XYuan,My1,SLiu,AA}, 
so as to form a spatially continuous surface \cite{Holo4}, since the resulting surface is homogenizable \cite{MDR}. This densely spaced or continuous implementation of RISs is referred to as {\it holographic RIS}.}
Channel modeling and data transmission schemes based on active holographic surfaces are investigated in \cite{Holo1} and \cite{Holo4}, respectively. However, to the best of our knowledge, no current research works have tackled the physical layer transmission design of passive holographic communications, where nearly-passive holographic RISs with spatially continuous apertures are deployed.
%A possible realization of RISs for application to THz frequencies is based on using novel metamaterials, such as graphene, which offers the possibility of realizing sub-wavelength reflecting elements \cite{UMM1,THzIRS}. 

\vspace{-2mm}
\subsection{Paper Contributions}
In this work, we focus our attention on the analysis of holographic RISs for application to massive MIMO systems in the THz band. In particular, the main contributions of this paper can be summarized as follows:
\begin{itemize}
{\item {\bf We derive the beam pattern of an RIS made of discrete elements, and propose an angular-domain beamforming framework.} By applying Fourier analysis to the reflection coefficients of the elements at the RIS, we prove that the beam pattern of an RIS with discrete elements can be represented as a weighted integral of Dirichlet kernel functions. On this basis, we propose an angular-domain beamforming framework. The weighting factors in the beam pattern are designed and the corresponding reflection coefficients of the RIS are reconstructed via the Fourier transform of the obtained weighting factors.
}
{\item {\bf We generalize the analysis and design to (continuous) holographic RISs.} Based on the proposed beamforming framework, we derive and obtain closed-form solutions for the beamforming design in two important cases, i.e., {\it narrow beam steering} (NBS) and {\it spatial bandpass filtering} (SBF), which play an important role in RIS-aided communication systems. We further extend these solutions to holographic RISs, in which the elements are closely spaced so as to yield a virtually continuous spatial aperture. The results reveal that the beam pattern of an ideal holographic RIS can be well approximated by an ultra-dense RIS, which has a practical hardware architecture.}
{\item {\bf We propose a closed-loop CE scheme to effectively estimate the broadband channels of THz massive MIMO systems based on holographic RISs.}
The proposed approach consists of downlink and uplink transmissions. In the downlink transmission, the holographic RIS uses SBF beamforming so that the users can coarsely estimate the range of LoS angles.
In the subsequent uplink transmission, the users with similar LoS angles are scheduled into the same {\it group}, and the coarsely-estimated LoS angles are exploited to design the uplink pilot signals for the finer-grained uplink CE.
To further reduce the uplink pilot overhead, a CS-based CE scheme is introduced, where the dual sparsity of THz MIMO channels in both the angular domain and delay domain is leveraged.}
\end{itemize}
\begin{table*}[h]
\caption{Some Important System Parameters and Channel Variables}
\centering
\begin{tabular}{|l|l|l|l|l|}
\hline
\multicolumn{1}{|c|}{\textbf{Notation}} & \multicolumn{1}{|c|}{\textbf{Definition}} & \multicolumn{1}{|c|}{\textbf{Notation}} & \multicolumn{1}{|c|}{\textbf{Definition}}  \\ \hline
 $M^{\rm B}$& Number of antennas at the BS & $M^{\rm U}$ & Number of antennas at the UE \\ \hline 
 $N_{\rm RF}$ & Number of RFCs at the BS & $A_x \times A_y$ & Physical size of the RIS \\ \hline
 $d$ & Element spacing of the DPA-based RIS & $\Phi(m,n)$ & Reflection coefficient of the ($m$,$n$)-th RIS element \\ \hline
\end{tabular}

\vspace{1mm}

\begin{tabular}{|l|l|}
\hline
\multicolumn{1}{|c|}{\textbf{Notation}$^*$} & \multicolumn{1}{c|}{\textbf{Definition}} \\ \hline
$L$, $K_{\rm f}$ & Number of NLoS paths of the RIS-UE channel, and the corresponding Rican factor \\ \hline
${\bm \psi}^{\rm B}$ (${\bm \psi}^{\rm R}$) & LoS angle of the BS-RIS channel seen by the BS (RIS) \\ \hline
${\bm \mu}^{\rm LoS}$ (${\bm \nu}^{\rm LoS}$) & LoS angle of the RIS-UE channel seen by the RIS (UE) \\ \hline
${\bm \mu}_{l}$ (${\bm \nu}_{l}$) & The $l$-th NLoS angle of the RIS-UE channel seen by the RIS (UE) \\ \hline
$\tau^{\rm DL,LoS}$ ($\tau^{\rm DL}_l$) & Delay offsets of the LoS ($l$-th NLoS) path in the downlink effective baseband channel \\ \hline
$\alpha^{\rm DL}$, $\beta^{\rm DL,LoS}$, $\beta^{\rm DL}_l$ & Channel coefficients of the BS-RIS-UE channel accounting for large-scale fading \\ \hline
${g_r^{{\rm{B}}}}\left( {{\bm \psi}_{\rm out}} \right)$ / ${g^{{\rm{U}}}}\left( {{\bm \psi}_{\rm out}} \right)$ & General beam pattern of the MIMO systems at the $r$-th RFC of the BS / at the UE \\ \hline
${g}\left( {{\bm \psi}_{\rm out}},{{\bm \psi}_{\rm in}} \right)$& General beam pattern of the RIS corresponding to an incident angle ${{\bm \psi}_{\rm in}}$ \\ \hline
\end{tabular}
\begin{tablenotes}
	\footnotesize
	\item *: The uplink version of these symbols is obtained by replacing the superscript ``DL'' with ``UL''. The user index $u$ may be added as a subscript if necessary.
\end{tablenotes}
\end{table*}
\vspace{-2mm}
\subsection{Notation}
Column vectors and matrices are denoted by lower- and upper-case boldface letters, respectively. ${\left(  \cdot  \right)\!^*}$, ${\left(  \cdot  \right)\!^T}$, ${\left( \cdot \right)\!^H}$ and ${\left(  \cdot  \right)\!^{\dag}}$ denote the conjugate, transpose, conjugate transpose and the pseudo-inverse, respectively. 
$j = \sqrt{-1}$ is the imaginary unit.
$\mathbb{C}$ and $\mathbb{Z}$ are the sets of complex-valued numbers and integers, respectively.
{$a \propto b$ denotes $a = Cb$ with $C$ being a non-zero constant.}
${[ \cdot ]_{m}}$ and ${[ \cdot ]_{m,n}}$ represent the ${m}$-th element of a vector and the ${m}$-th row, ${n}$-th column element of a matrix, respectively. 
{$[\bf A]_{\cal I}$ denotes the submatrix consisting of the columns of $\bf{A}$ indexed by the ordered set $\cal I$.}
${\rm diag}(\cdot)$, ${\left\| \cdot \right\|_F}$, and $ \otimes $ represent the diagonalization, Frobenius norm, and Kronecker product, respectively.
${\rm{vec}}\left(  \cdot  \right)$ is the vectorization operation according to the columns of the matrix, and ${\rm{vec}}^{-1}\left(  \cdot  \right)$ is the corresponding inverse operation.
{${\rm card} ({\cal I})$ is the cardinality of the set ${\cal I}$.}
${\Xi _N}\left(  x  \right)$ is the $N$-order Dirichlet kernel function given by ${\Xi _N}\left( x \right) = {\frac{{\sin \left( {Nx/2} \right)}}{{N\sin (x/2)}}}$ for ${x \ne 2k\pi }$, and ${\Xi _N}\left( 2k\pi \right) = {{\left( { - 1} \right)}^{k\left( {N - 1} \right)}}$, $k \in \mathbb{Z}$.
%\begin{equation*}
%{\Xi _N}\left( x \right) = \left\{ {\begin{array}{*{20}{l}}
%{\frac{{\sin \left( {Nx/2} \right)}}{{N\sin (x/2)}},}&{x \ne 2k\pi }\\
%{{{\left( { - 1} \right)}^{k\left( {N - 1} \right)}},}&{x = 2k\pi }
%\end{array}} \right.,k \in \mathbb{Z} \text{.}
%\end{equation*}
The ``sinc'' function is defined as ${\rm{sinc}}\left( x \right) = \frac{{\sin {x} }}{{x}}$ for $x \ne 0$, and ${\rm{sinc}}\left( 0 \right) = 1$.
$\delta(x)$ is the Dirac function.
${{\bf{F}}_N}$ is an $N \times N$ normalized discrete Fourier transformation (DFT) matrix whose elements are ${\left[ {{{\bf{F}}_N}} \right]_{m,n}} = \frac{1}{{\sqrt N }}{e^{ - j\frac{{2\pi(m-1)(n-1)}}{N}}}$. $\mathcal{U}(a,b)$ denotes the uniform distribution within $(a,b)$. 
{Important system parameters and channel variables used in this paper are collected in Table I.}
\begin{figure}[t]
\vspace{-3mm}
\centering
%\captionsetup{font={footnotesize,color={red}}, name = {Fig.}, labelsep = period}
\subfigure[]{\includegraphics[width=3in]{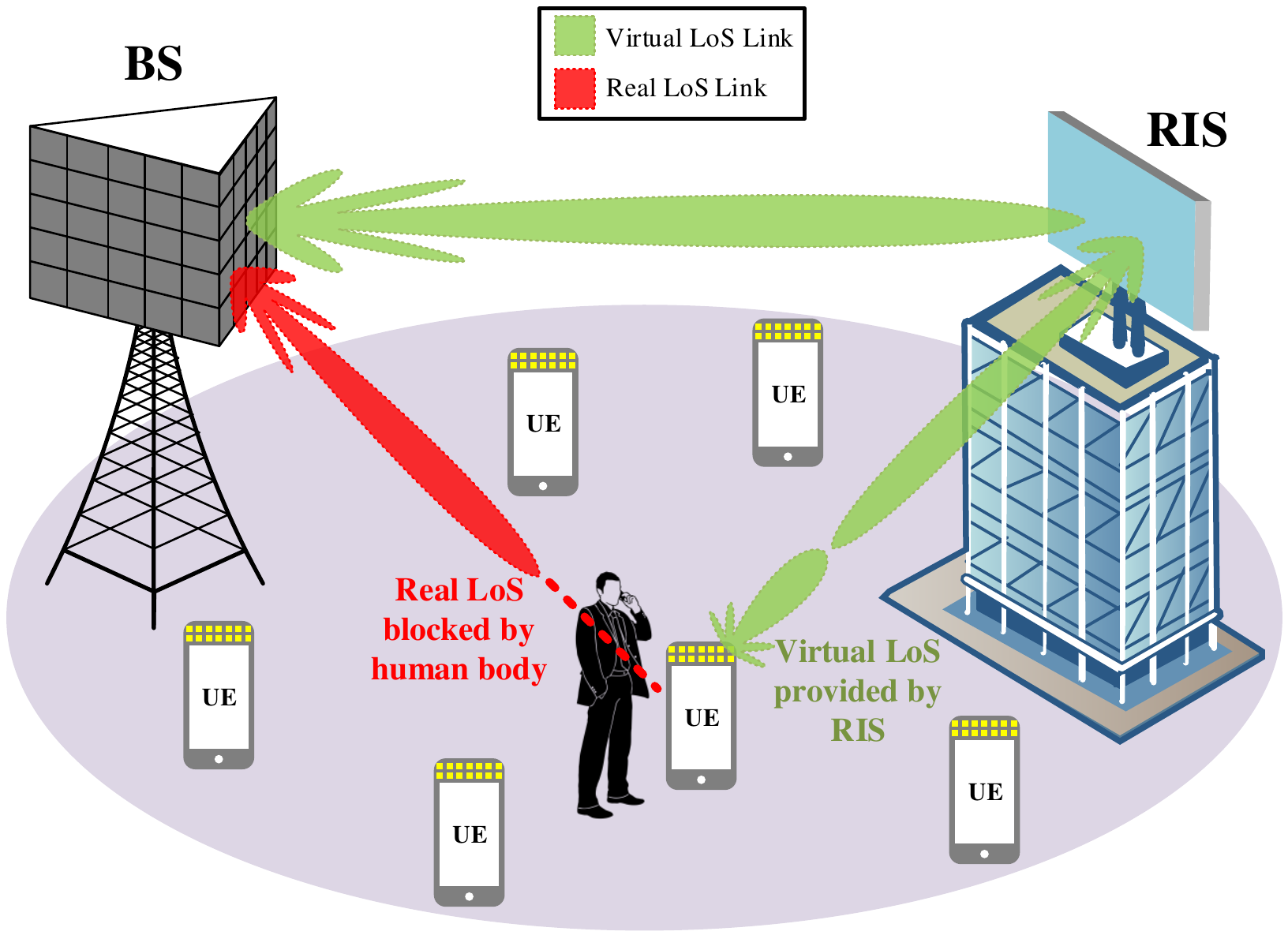}}
\\
\subfigure[]{\includegraphics[width=3in]{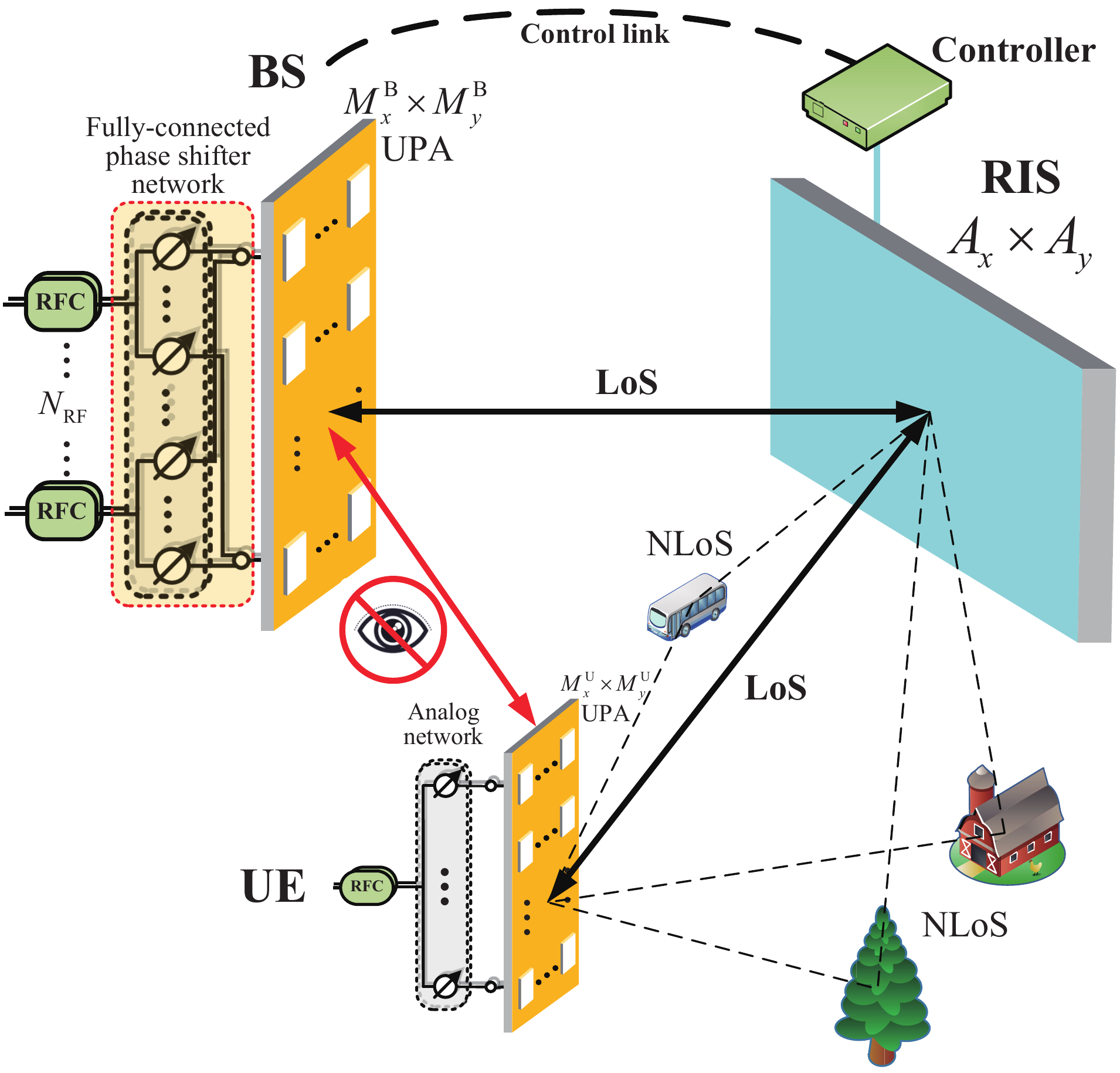}}
\caption{The model of an RIS-aided THz massive MIMO system. (a) A multi-antenna UE is served by the BS with the help of an RIS when the LoS path is blocked by possible obstacles. (b) The hardware architectures at the BS and the UE.}
\label{fig:SysMod}
\vspace{-5mm}
\end{figure}

\vspace{-3mm}
\section{System Description and Channel Model}
In this section, we present the system model and the effective baseband channel model of the considered RIS-aided THz massive MIMO system over frequency-selective fading channels.
We consider an RIS-aided THz massive MIMO system that operates in time division duplex (TDD) mode, as illustrated in Fig. \ref{fig:SysMod}. The base station (BS) and user-equipments (UEs) are equipped with half-wavelength spaced uniform planar arrays (UPAs) that consist of $M^{\rm B} = M^{\rm B}_x \times M^{\rm B}_y$ and $M^{\rm U} = M^{\rm U}_x \times M^{\rm U}_y$ antennas, respectively, {where $M^{\rm B}_x$ ($M^{\rm U}_x$) and $M^{\rm B}_y$ ($M^{\rm U}_y$) are the numbers of antennas along the azimuth and elevation directions, respectively}. An RIS of physical size $A_x \times A_y$ is deployed to enhance the effective coverage of the BS.
{A low-power RIS controller is connected to the RIS and is operated by the BS through a control link. The BS carries out beamforming optimization and feeds the results to the RIS controller through the control link so that the reflection coefficients of each reflecting element of the RIS can be adjusted accordingly.}
As shown in Fig. \ref{fig:SysMod}(a), we assume that the LoS link between the BS and the UE is blocked by an obstacle or by a human body. Thus, the UE communicates with the BS only via the RIS, which is regarded as a {\it virtual LoS} transmission in RIS-aided systems \cite{My1}.
To reduce the power consumption and hardware cost, a hybrid analog-digital architecture is considered at the BS, i.e., there are only $N_{\rm RF} \ll M^{\rm B}$ RFCs at the BS, and each of them is connected to $M^{\rm B}$ antennas through $M^{\rm B}$ phase shifters. In addition, each UE employs analog beamforming, in which only one RFC is connected to $M^{\rm U}$ antennas through $M^{\rm U}$ phase shifters.
An orthogonal frequency division multiplexing (OFDM) transmission scheme with $K$ subcarriers and sampling period $T_{\rm s}$ is adopted. The cyclic prefix (CP) of length $N_{\rm CP}T_{\rm s}$ is added before each OFDM symbol to avoid inter-symbol interference. The center-carrier frequency is $f_{\rm c}$ corresponding to a wavelength $\lambda$.

Through an appropriate deployment of the RIS, we assume that there exists a LoS path between the BS and the RIS. The LoS angles between the BS and the RIS are assumed to be known in advance based on the location of the RIS \cite{My1}. In THz channels, the path loss of the non-LoS (NLoS) paths is known to be much larger than that of the LoS paths. Therefore, we neglect the NLoS paths in the channel between the BS and the RIS.

In the following text, we first introduce the physical channel model that is widely considered in previous works on RIS-aided MIMO systems \cite{XYuan,My1,AA,SLiu,THzIRS}. Then, we introduce an effective baseband channel model by taking into account the beamforming design at the BS, the RIS and the UE.

{\it a) Physical Channel Model with Discrete RISs:} Under the assumptions that the RIS has $N$ reflecting elements, the downlink spatial channel ${\bf G} \in \mathbb{C}^{N \times M^{\rm B}}$ from the BS to the RIS can be modeled as
\begin{equation}
\label{equ:PhyCh1}
{\bf{G}} = \alpha^{\rm DL} {\bf{a}}_{{\rm{R}}}\left( {{\bm \psi^{{\rm{R}}}}} \right){\bf{a}}_{{\rm{B}}}^H\left( {{\bm \psi^{{\rm{B}}}}} \right) \text{,}
\end{equation}
where $\alpha^{\rm DL}$ is the channel coefficient, ${\bm \psi^{\rm B}} = [\psi_{\rm BS}^{\rm azi},\psi_{\rm BS}^{\rm ele}]^T$ and ${\bm \psi^{\rm R}} = [\psi_{\rm RIS}^{\rm azi},\psi_{\rm RIS}^{\rm ele}]^T$ are the LoS angle of departure (AoD) and LoS angle of arrival (AoA) of the BS-RIS channel, respectively.
%The specific definition of ``virtual'' angle will be given in the next section. 
${\bm \psi^{\rm B}}$ (${\bm \psi^{\rm R}}$) includes both the azimuth part $\psi_{\rm BS}^{\rm azi}$ ($\psi_{\rm RIS}^{\rm azi}$) and elevation part $\psi_{\rm BS}^{\rm ele}$ ($\psi_{\rm RIS}^{\rm ele}$), which are assumed to be fixed and known as detailed in previous text. ${\bf a}_{\rm B}\left( {{\bm \psi^{{\rm{B}}}}} \right) \in \mathbb{C}^{M^{\rm B} \times 1}$ and ${\bf a}_{\rm R}\left( {{\bm \psi^{{\rm{R}}}}} \right) \in \mathbb{C}^{N \times 1}$ denote the steering vectors at the BS and the RIS, respectively. ${\bf a}_{\rm B}\left( {{\bm \psi^{{\rm{B}}}}} \right)$ is given by
\begin{align}
\label{equ:StrVec}
{\bf a}_{\rm B}\left( {{\bm \psi^{{\rm{B}}}}} \right) = & \left[ 1,...,{e^{ - jd_{\rm UPA}\left( {{m_x}\psi _{{\rm{BS}}}^{{\rm{azi}}} + {m_y}\psi _{{\rm{BS}}}^{{\rm{ele}}}} \right)}},..., \right. \nonumber \\
& \left. {e^{ - jd_{\rm UPA}\left[ {\left( {M_x^{\rm{B}} - 1} \right)\psi _{{\rm{BS}}}^{{\rm{azi}}} + \left( {M_y^{\rm{B}} - 1} \right)\psi _{{\rm{BS}}}^{{\rm{ele}}}} \right]}} \right]^T \text{,}
\vspace{-1mm}
\end{align}
where $d_{\rm UPA} = \lambda/2$ is the element spacing of the UPA at the BS, and $0 \le m_x \le (M_{x}^{\rm B}-1)$, $0 \le m_y \le (M_{y}^{\rm B}-1)$. ${\bf a}_{\rm R}\left( {{\bm \psi^{{\rm{R}}}}} \right)$ can be written by using a similar notation and assumptions.

{As for the RIS-UE spatial channel ${\bf H} \in \mathbb{C}^{M^{\rm U} \times N}$, we consider a Rician fading channel model that consists of one LoS path and $L$ NLoS paths, as shown in Fig. \ref{fig:SysMod}(b). In particular, we have
\begin{align}
\label{equ:PhyCh2}
{\bf{H}} = & \beta^{\rm DL,LoS} {{\bf{a}}_{{\rm{U}}}}\left( {{\bm \nu^{{\rm{LoS}}}}} \right){\bf{a}}_{{\rm{R}}}^H\left( {{\bm \mu^{{\rm{LoS}}}}} \right) \nonumber \\
& + \frac{1}{{\sqrt {L{K_{\rm{f}}}} }}\sum\limits_{l = 1}^L {{\beta}^{\rm DL}_{l}{{\bf{a}}_{{\rm{U}}}}\left( {{\bm \nu _l}} \right){\bf{a}}_{{\rm{R}}}^H\left( {{\bm \mu _l}} \right)}  \text{,}
\end{align}
where $\beta^{\rm DL,LoS}$ and ${\beta}^{\rm DL}_{l}$ are the channel coefficients of the LoS component and the $l$-th NLoS component, respectively,}
$K_{\rm f}$ is the Rician factor that denotes the ratio of the energy between the LoS and NLoS channels, $\bm \mu^{{\rm{LoS}}}$ and $\bm \nu^{{\rm{LoS}}}$ are the LoS AoD and the LoS AoA, respectively, $\bm \mu_{l}$ and $\bm \nu_{l}$ are the NLoS AoD and the NLoS AoA of the $l$-th NLoS path. The steering vector at the UE ${\bf a}_{\rm U}\left( \cdot \right) \in \mathbb{C}^{M^{\rm U} \times 1}$ can be formulated similar to \eqref{equ:StrVec}.

{\it b) Effective Baseband Channel Model:} The effective baseband channel is defined as the inner product between the analog beamforming vector at the transceiver and the steering vector of the physical channels. Assuming that the $r$-th RFC ($1 \le r \le N_{\rm RF}$) of the BS and the UE use the analog beamforming vectors ${\bf f}_r \in \mathbb{C}^{M^{\rm B}\times 1}$ and ${\bf w} \in \mathbb{C}^{M^{\rm U}\times 1}$, respectively, and that the RIS uses $\bm \varphi \in \mathbb{C}^{N \times 1}$ as the reflection coefficients, the effective baseband channel $h_r^{\rm DL,eff}$ can be written as
{\begin{align}
\label{equ:EffCh}
h_r^{\rm DL,eff} & = {{\bf{w}}^H}{\bf{H}}{\bm \Phi}{\bf{G}}{{\bf{f}}_r} \nonumber \\
& = \alpha^{\rm DL} {{g_r^{{\rm{B}}}}\left( {{\bm \psi^{{\rm{B}}}}} \right)} \bigg[ \beta^{\rm DL,LoS} g\left( {{\bm \mu^{{\rm{LoS}}}},{\bm \psi^{{\rm{R}}}}} \right){g^{{\rm{U}}}}\left( {{\bm \nu^{{\rm{LoS}}}}} \right) \nonumber \\
& + \frac{1}{{\sqrt {L{K_{\rm{f}}}} }}\sum\limits_{l = 1}^L {\beta^{\rm DL}_{l} g\left( {{\bm \mu _l},{\bm \psi^{{\rm{R}}}}} \right){g^{{\rm{U}}}}\left( {{\bm \nu _l}} \right)}  \bigg] \text{,}
\vspace{-2mm}
\end{align}
where ${\bf \Phi} = {\rm diag} (\bm \varphi) \in \mathbb{C}^{N \times N}$, and ${g_r^{{\rm{B}}}} \left( {{\bm \psi^{{\rm{B}}}}} \right) = {\bf{a}}_{{\rm{B}}}^H\left( {{\bm \psi^{{\rm{B}}}}} \right){\bf f}_r$, $g\left( {{\bm \mu},{\bm \psi^{{\rm{R}}}}} \right) = {\bf{a}}^H_{{\rm{R}}}\left( {{\bm \mu}} \right) {\bm \Phi} {\bf{a}}_{{\rm{R}}}\left( {{\bm \psi^{{\rm{R}}}}} \right)$, and ${g^{{\rm{U}}}}\left( {{\bm \nu}} \right) = {\bf w}^H{{\bf{a}}_{{\rm{U}}}}\left( {{\bm \nu}} \right)$ are the {\it beam patterns} of the BS, the RIS, and the UE, respectively.
}
%We observe that the dimension of the effective baseband channel is agnostic to the dimensions of the arrays deployed in the system, and thus this model simplifies the channel representation.

The effective baseband channel model in \eqref{equ:EffCh} can be generalized to the case of frequency-selective channel. In a frequency-selective channel, in particular, the BS-RIS-UE effective baseband channel in the delay domain can be formulated as
\begin{equation}
\label{equ:DLmodel1}
h^{\rm DL}_{r}\left( \tau  \right) = \underbrace {\alpha^{\rm DL}  g_r^{{\rm{B}}}\left( {{\bm \psi^{{\rm{B}}}}} \right)}_{\text{BS to RIS}} \underbrace {\left[ {h^{{\rm{DL,LoS}}}\left( \tau  \right) + h^{{\rm{DL,NLoS}}}\left( \tau  \right)} \right]}_{\text{RIS to UE}} \text{,}
\end{equation}
\vspace{-1mm}
where
\vspace{-1mm}
\begin{equation}
\label{equ:DLmodel2}
{h^{{\rm{DL,LoS}}}\left( \tau  \right)} = \beta^{\rm DL,LoS} g\left( {\bm \mu^{{\rm{LoS}}},{\bm \psi^{{\rm{R}}}}} \right) g^{{\rm{U}}}\left( {\bm \nu^{{\rm{LoS}}}} \right) p\left( {\tau  - {\tau^{\rm DL,LoS}}} \right) \text{,}
\end{equation}
\vspace{-1mm}
and
\vspace{-1mm}
\begin{equation}
\label{equ:DLmodel3}
{h^{{\rm{DL,NLoS}}}\left( \tau  \right)}=\frac{1}{{\sqrt {L{K_{\rm{f}}}} }}\sum\limits_{l = 1}^L {\beta^{\rm DL}_{l} g\left( {{\bm \mu_{l}},{\bm \psi^{{\rm{R}}}}} \right) g^{{\rm{U}}}\left( {{\bm \nu_{l}}} \right)} p\left( {\tau  - {\tau^{\rm DL}_{l}}} \right) \text{.}
\end{equation}
In (\ref{equ:DLmodel2}) and (\ref{equ:DLmodel3}), $\tau^{\rm DL,LoS}$ and $\tau^{\rm DL}_{l}$ denote the delay offsets corresponding to the LoS path and the $l$-th NLoS path, respectively, and $p(\tau)$ is the pulse shaping filter function. 

It is worth noting that, in contrast with RISs made of discrete elements \cite{CHuang,Model,Reviewer,THzIRS,RZhang2,MSlim,FFFCE,FSFCE,XYuan,My1,SLiu,AA}, a holographic RIS is modeled as an array with a spatially continuous aperture \cite{Holo4} (i.e., $N \to \infty$). The physical channels associated with a holographic RIS cannot be represented in terms of the finite-dimensional matrices in \eqref{equ:PhyCh1} and \eqref{equ:PhyCh2}. Therefore, in this paper, we utilize the effective baseband channel model in (\ref{equ:DLmodel1})-(\ref{equ:DLmodel3}) to describe the channels associated with a holographic RIS.
The beam patterns in (\ref{equ:EffCh}) are elaborated in detail in the next section.

\color{black}
Since we consider that the system operates in TDD mode, we assume that the channel reciprocity between the uplink and downlink transmissions holds. In particular, the uplink channels can be modeled based on the downlink channels reported in previous text. We omit the details of the uplink channel model for brevity.

\begin{figure*}[t]
\vspace{-2mm}
%\captionsetup{font={footnotesize}, name = {Fig.}, labelsep = period}
%\captionsetup[subfigure]{singlelinecheck = on, justification = raggedright, font={footnotesize}}
\centering
\begin{minipage}[t]{1\linewidth}
\centering
\includegraphics[width=6.5in]{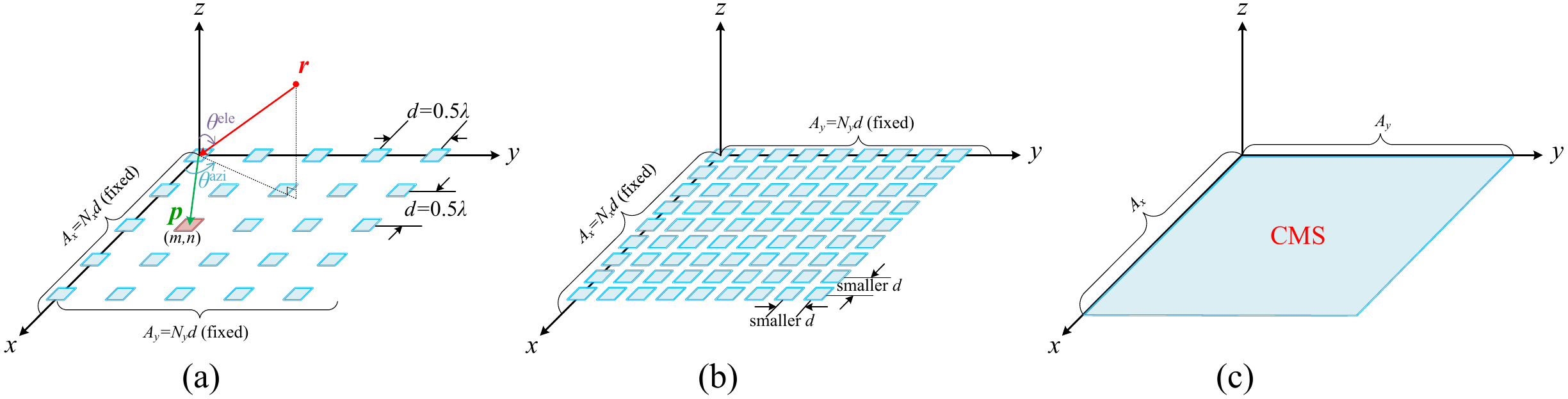}
%\centerline{(a)}
\end{minipage}
\caption{Illustration of different types of RIS. (a) Critically-spaced RIS with $d = 0.5\lambda$; (b) ultra-dense RIS with $d < 0.5\lambda$, which has a spatial quasi-continuous aperture; (c) CMS with $d \to 0$ (ignoring the physical size of a single reflecting element), which has an ideal reconfigurable continuous aperture.}
\label{fig:model}
\end{figure*}

{{\it Remark}: Although it has been reported that the overall MIMO channel model in the THz band is consistent with that in the mmWave band (see, e.g., \cite{THzModel1,THzModel2,THzIRS}), some distinct differences between the THz and the mmWave bands, as detailed in the following text, may cause different setting of channel parameters and thus deserve more attention. First, the atmospheric attenuation caused by molecular absorption \cite{THz1,THz2} becomes non-negligible in the THz band, thus introducing an extra path loss for link budget. Second, THz signals with an extremely short wavelength are likely to undergo {\it diffuse scattering} on the scattering surface \cite{DifSca}. This results in a limited number of effective NLoS paths (e.g., $L=1$) and a large Rician factor (e.g. $K_{\rm f} = 30$ dB). Third, the THz are more vulnerable to the presence of blockages, which may result in weak LoS directive paths. These factors motivate the deployment of RISs to provide strong virtual LoS links for THz communications in the absence of reliable LoS/NLoS links. In addition, due to the large number of reflecting elements, which yield a better controllability of the radio signals and a high beamforming gain, holographic-based RISs provide further advantages compared with their non-holographic counterpart. This is further discussed next.
}
\vspace{-3mm}
\section{Beamforming Design for Holographic RISs}
To complete the formulation of the effective baseband channel model in (\ref{equ:DLmodel1})-(\ref{equ:DLmodel3}), we first derive and analyze the beam pattern of RISs based on discrete planar arrays (DPAs).
Then, we extend the results to RISs with spatially continuous apertures, which are referred to as continuous metasurfaces (CMSs).
In addition, we derive closed-form beamforming solutions in two important cases, {i.e., narrow beam steering and spatial bandpass filtering}, which both play an essential role in the proposed CE scheme for RIS-aided THz massive MIMO systems.
\vspace{-2mm}
\subsection{RISs Based on Discrete Planar Arrays}
As shown in Fig. \ref{fig:model}(a) and (b), a DPA-based RIS placed on the $x$-$y$ plan consists of numerous evenly-spaced reflecting elements. The number of elements along the $x$- and $y$-directions is $N_x \in \mathbb{Z}$ and $N_y \in \mathbb{Z}$, respectively. The distance between two adjacent elements is $d$, and we assume $d \le \lambda/2$ due to the Nyquist sampling theorem. We define the total physical size of a DPA-based RIS as $A_x \times A_y$ with $A_x =N_xd$ and $A_y = N_yd$, and assume that $A_x$ and $A_y$ remain unchanged unless stated otherwise. Let $(x_m,y_n)$ be the coordinate of the $(m,n)$-th reflecting element, then we have
\begin{equation}
{x_m} = (m-1)d \text{, } {y_n} = (n-1)d \text{,}
\end{equation}
where $1 \le m \le N_x$, $1 \le n \le N_y$. Assume that a narrowband reference signal ``1'' impinges on the RIS with an azimuth AoA ${\theta_{\rm{in}}^{{\rm{azi}}}} \in [0,2\pi)$ and an elevation AoA ${\theta_{\rm{in}}^{{\rm{ele}}}} \in [0,\pi/2]$ (defined in Fig. 2(a)). The phase-difference of the incident signal at the $(m,n)$-th element {(compared to the reference point $(0,0)$)} can be written as
\begin{align}
\label{equ:a}
a\left( {x_m,y_n;{\bm \psi _{{\rm{in}}}}} \right) & = {e^{j\frac{{2\pi }}{\lambda} {\bf p}^T {\bf r} }} \nonumber \\
& = {e^{j\frac{{2\pi }}{\lambda}(x_m\cos{\theta_{\rm{in}}^{{\rm{azi}}}}\sin{\theta_{\rm{in}}^{{\rm{ele}}}} + y_n\sin{\theta_{\rm{in}}^{{\rm{azi}}}}\sin {\theta_{\rm{in}}^{{\rm{ele}}}} )}} \nonumber \\
& = {e^{j({x_m}\psi _{{\rm{in}}}^{{\rm{azi}}} + {y_n}\psi _{{\rm{in}}}^{{\rm{ele}}})}} \text{,}
\end{align}
where ${\bf{p}} = [x_n,y_m,0]^T$ is the position vector of the $(m,n)$-th element, ${\bf{r}} = [\cos{\theta_{\rm{in}}^{{\rm{azi}}}}\sin {\theta_{\rm{in}}^{{\rm{ele}}}},$ $ \sin{\theta_{\rm{in}} ^{{\rm{azi}}}}\sin {\theta_{\rm{in}}^{{\rm{ele}}}},\cos {\theta_{\rm{in}}^{{\rm{ele}}}}]^T$ is the vector of the incident direction, and ${\bm \psi _{{\rm{in}}}} \buildrel \Delta \over = \left[ {\psi _{{\rm{in}}}^{{\rm{azi}}},\psi _{{\rm{in}}}^{{\rm{ele}}}} \right]^T$ with $\psi _{{\rm{in}}}^{{\rm{azi}}} \buildrel \Delta \over = \frac{{2\pi }}{\lambda }\cos{\theta_{\rm{in}}^{{\rm{azi}}}}\sin{\theta_{\rm{in}}^{{\rm{ele}}}}$, $\psi _{{\rm{in}}}^{{\rm{ele}}} \buildrel \Delta \over = \frac{{2\pi }}{\lambda }\sin{\theta_{\rm{in}}^{{\rm{azi}}}}\sin {\theta_{\rm{in}}^{{\rm{ele}}}}$ is a 2-tuple variable accounting for the AoA of the incident signal{\footnote{We assume that $\bm \psi_{\rm in}=\left[ {\psi _{{\rm{in}}}^{{\rm{azi}}},\psi _{{\rm{in}}}^{{\rm{ele}}}} \right]^T$ is known and is treated as a constant in the beamforming design. This follows from the assumption that the LoS angles of the BS-RIS channel are fixed and known.}}. Similarly, for the AoD denoted by ${\theta _{\rm{out}}^{{\rm{azi}}}}$ and ${\theta _{\rm{out}}^{{\rm{ele}}}}$, we define ${\bm \psi _{{\rm{out}}}} \buildrel \Delta \over = \left[ {\psi _{{\rm{out}}}^{{\rm{azi}}},\psi _{{\rm{out}}}^{{\rm{ele}}}} \right]^T$ with $\psi _{{\rm{out}}}^{{\rm{azi}}} \buildrel \Delta \over = \frac{{2\pi }}{\lambda }\cos{\theta_{\rm{out}}^{{\rm{azi}}}}\sin {\theta_{\rm{out}}^{{\rm{ele}}}}$, $\psi _{{\rm{out}}}^{{\rm{ele}}} \buildrel \Delta \over = \frac{{2\pi }}{\lambda }\sin{\theta_{\rm{out}}^{{\rm{azi}}}}\sin {\theta_{\rm{out}}^{{\rm{ele}}}}$. The definitions of ${\bm \psi _{{\rm{in}}}}$ and ${\bm \psi _{{\rm{out}}}}$ simplify the notation. The received signal along the direction of observation (i.e., reflection) ${\bm \psi}_{\rm{out}}$ after the incident signal is reflected by the RIS is denoted by $g({\bm \psi}_{\rm{out}},{\bm \psi}_{\rm{in}})$, which is given by the superposition of the signals reflected by all the individual elements of the RIS
\begin{align}
\label{equ:g1}
g\left( {{{\bm \psi} _{{\rm{out}}}},{{\bm \psi} _{{\rm{in}}}}} \right) & = \sum\limits_{n = 1}^{{N_x}} {\sum\limits_{m = 1}^{{N_y}} {{a^*}\left( {{x_m},{y_n};{{\bm \psi} _{{\rm{out}}}}} \right)} } \Phi \left( {m,n} \right)a\left( {{x_m},{y_n};{{\bm \psi} _{{\rm{in}}}}} \right) \nonumber \\
&= \sum\limits_{n = 1}^{{N_x}} {\sum\limits_{m = 1}^{{N_y}} {\Phi \left( {m,n} \right){{e^{ - j\left( {{x_m}{\Delta _x} + {y_n}{\Delta _y}} \right)}}}} } \text{,}
\end{align}
where $\Delta_x = {\psi}^{\rm{azi}} _{{\rm{out}}} - {\psi}^{\rm{azi}} _{{\rm{in}}}$, $\Delta_y = {\psi}^{\rm{ele}} _{{\rm{out}}} - {\psi}^{\rm{ele}} _{{\rm{in}}}$, and $\Phi \left( {m,n} \right) \in \mathbb{C}$ is the reflection coefficient of the $(m,n)$-th element of the RIS, whose amplitude and phase are software-programmable via the RIS controller.

We refer to $g\left( {{{\bm \psi} _{{\rm{out}}}},{{\bm \psi} _{{\rm{in}}}}} \right)$ as the {beam pattern} of the DPA-based RIS, which is consistent with the definition in \eqref{equ:EffCh}. The amplitude of $g\left( {{{\bm \psi} _{{\rm{out}}}},{{\bm \psi} _{{\rm{in}}}}} \right)$ can be used to evaluate the intensity of the signal along the direction ${\bm \psi}_{\rm{out}}$ after reflection from the RIS.
In general terms, $\Phi \left( {m,n} \right)$ can be treated as a two-dimensional discrete signal in the spatial domain with spatial sampling period equal to $d$ in both the azimuth and elevation directions. By using the discrete-time Fourier transform (DTFT), $\Phi \left( {m,n} \right)$ can be represented as

\begin{equation}
\label{equ:2dDTFT}
\Phi \left( {m,n} \right) = {\left( {\frac{d}{{2\pi }}} \right)^2}\int_{\frac{{2\pi }}{d}} {\int_{\frac{{2\pi }}{d}} {\omega \left( {k,l} \right){e^{j\left( {dmk + dnl} \right)}}} } {\rm{d}}k{\rm{d}}l \text{,}
\end{equation}
where 
\begin{equation}
\label{equ:2dIDTFT}
\omega \left( {k,l} \right) = \sum\limits_{m = 1}^{{N_x}} {\sum\limits_{n = 1}^{{N_y}} {\Phi \left( {m,n} \right)} } {e^{ - j\left( {dmk + dnl} \right)}}
\end{equation}
is a 2-dimensional periodic function whose period is $\frac{2\pi}{d}$ with respect to both $k$ and $l$, and $\int_{\frac{{2\pi }}{d}}(\cdot){\rm d}(\cdot)$ denotes the integral in an arbitrary interval of length $\frac{2\pi}{d}$. Substituting (\ref{equ:2dDTFT}) into (\ref{equ:g1}), we obtain \eqref{equ:g2}, as shown at the top of the next page,
\begin{figure*}
\begin{align}
\label{equ:g2}
g\left( {{{\bm \psi} _{{\rm{out}}}},{{\bm \psi} _{{\rm{in}}}}} \right) & = {\left( {\frac{d}{{2\pi }}} \right)^2} {e^{ - j\left( {{x_m}{\Delta _x} + {y_n}{\Delta _y}} \right)}} \sum\limits_{m = 1}^{{N_x}} \sum\limits_{n = 1}^{{N_y}} \int_{\frac{{2\pi }}{d}} {\int_{\frac{{2\pi }}{d}} {\omega \left( {k,l} \right){e^{j\left( {dmk + dnl} \right)}}} } {\rm{d}}k{\rm{d}}l \nonumber \\
& = {\left( {\frac{d}{{2\pi }}} \right)^2}\int_{\frac{{2\pi }}{d}} \int_{\frac{{2\pi }}{d}} \omega \left( {k,l} \right) \sum\limits_{m = 1}^{{N_x}} {{e^{j\left( {dmk - {x_m}{\Delta_x}} \right)}}}  \sum\limits_{n = 1}^{{N_y}} {{e^{j\left( {dnl - {y_n}{\Delta _y}} \right)}}}  {\rm{d}}k{\rm{d}}l \nonumber \\
& = \frac{{{A_x}{A_y}}}{{4{\pi ^2}}}{e^{j\frac{{d - {A_x}}}{2}{\Delta _x}}}{e^{j\frac{{d - {A_y}}}{2}{\Delta _y}}} \int_{\frac{{2\pi }}{d}} {\int_{\frac{{2\pi }}{d}} {\omega '\left( {k,l} \right)}} {\Xi _{{N_x}}}\left[ {d\left( {k - \Delta_x} \right)} \right] {\Xi _{{N_y}}}\left[ {d\left( {l - \Delta_y} \right)} \right]  {\rm{d}}k{\rm{d}}l \text{,}
\end{align}
\hrule
\end{figure*}
where 
\begin{equation}
\label{equ:EffOmega}
\omega '\left( {k,l} \right) = \omega \left( {k,l} \right){e^{j\frac{{{N_x} + 1}}{2}dk}}{e^{j\frac{{{N_y} + 1}}{2}dl}} \text{.}
\end{equation}
Equation (\ref{equ:g2}) reveals that the beam pattern $g\left( {{{\bm \psi} _{{\rm{out}}}},{{\bm \psi} _{{\rm{in}}}}} \right)$ can be formulated as a weighted integral of Dirichlet kernel functions $\Xi(\cdot)$ \cite{korean} whose weighting factors are $\omega '\left( {k,l} \right)$.
This is a generalization of the far-field results in \cite{SJin}.
Equations (\ref{equ:2dDTFT})-(\ref{equ:g2}) shed light on the design of angular-domain beamforming for DPA-based RIS. In particular, the proposed design is based on two steps: {(i)} first,  $\omega' \left( {k,l} \right)$ in \eqref{equ:g2} is optimized in order to design $g\left( {{{\bm \psi} _{{\rm{out}}}},{{\bm \psi} _{{\rm{in}}}}} \right)$ that corresponds to the desired direction of observation; and {(ii)} then, the corresponding reflection coefficients $\Phi \left( {m,n} \right)$ are reconstructed via (\ref{equ:2dDTFT}). In the next sub-section, we further explain this procedure by providing the optimal beamforming for two specific cases.

\vspace{-2mm}
\subsection{Beamforming Design for DPA-Based RIS}
In this sub-section, we describe the following two use cases for beamforming design based on the proposed approach.

{\it a) Narrow beam steering.} Given the desired beamforming direction ${{\bm \psi}_{{\rm{opt}}}} = \left[ {\psi _{{\rm{opt}}}^{{\rm{azi}}},\psi _{{\rm{opt}}}^{{\rm{ele}}}} \right]^T$, the target of NBS is to maximize $|g\left( {{{\bm \psi} _{{\text{out}}}},{{\bm \psi} _{{\text{in}}}}} \right)|$ for ${{\bm \psi}_{{\rm{out}}}}={{\bm \psi}_{{\rm{opt}}}}$, and to minimize (null) $|g\left( {{{\bm \psi} _{{\text{out}}}},{{\bm \psi} _{{\text{in}}}}} \right)|$ for ${{\bm \psi}_{{\rm{out}}}}\ne{{\bm \psi}_{{\rm{opt}}}}$. Since $\left|{\Xi _{{N_x}}}\left[ {d\left( {k -{\Delta _x}} \right)} \right]\right|$ and $\left|{\Xi _{{N_y}}}\left[ {d\left( {l - {\Delta _y}} \right)} \right]\right|$ attain their maximum if and only if $k=\frac{2\pi u}{d} + \Delta_x$, $u \in \mathbb{Z}$ and $l=\frac{2\pi v}{d} + \Delta_y$, $v \in \mathbb{Z}$, respectively, we can impose the following design for the angular-domain coefficients of NBS
\vspace{-1mm}
\begin{equation}
\label{equ:OmegaStr}
{{\omega}_{{\rm{NBS}}}}\left( {k,l}; {\bm \psi}_{\rm opt} \right) \propto \sum\limits_{u \in \mathbb{Z}} {\delta (k - {k_{{\rm{opt}},u}})} \sum\limits_{v \in \mathbb{Z}} {\delta (l - {l_{{\rm{opt}},v}})} \text{,}
\vspace{-1mm}
\end{equation}
where $k_{{\rm opt},u} = \frac{2\pi u}{d}+\left(\psi_{\rm{opt}}^{\rm{azi}} - \psi_{\rm{in}}^{\rm{azi}}\right)$ and
$l_{{\rm{opt}},v} = \frac{2\pi v}{d}+\left(\psi_{\rm{opt}}^{\rm{ele}} - \psi_{\rm{in}}^{\rm{ele}}\right)$.
It can be readily verified that (\ref{equ:OmegaStr}) ensures the required periodicity of $\frac{2\pi}{d}$ with respect to both $k$ and $l$. 
The corresponding beam pattern for NBS can be obtained by substituting (\ref{equ:OmegaStr}) into \eqref{equ:g2} and \eqref{equ:EffOmega}, which yields
\begin{align}
\label{equ:gNBS}
g_{\rm{NBS}}& \left( {{{\bm \psi} _{{\rm{out}}}},{{\bm \psi} _{{\rm{in}}}}};{\bm \psi}_{\rm opt} \right) \propto {e^{j\frac{{d - {A_x}}}{2}{\Delta _x}}}{e^{j\frac{{d - {A_y}}}{2}{\Delta _y}}} \nonumber \\
& \times {\Xi _{{N_x}}}\left[ {d} \left( {k_{\rm{opt}} -\Delta_x} \right) \right] {\Xi _{{N_y}}}\left[ {d} \left( {l_{\rm{opt}} - \Delta_y} \right) \right] \text{,}
\end{align}
where $k_{\rm opt} \buildrel \Delta \over =  k_{\rm{opt},0}$ and $l_{\rm opt} \buildrel \Delta \over =  l_{\rm{opt},0}$, since in (\ref{equ:g2}) we choose the integral intervals containing $k_{{\rm opt},0}$ and $l_{{\rm opt},0}$ for $k$ and $l$, respectively.
From (\ref{equ:2dDTFT}), the reflection coefficients for NBS can be formulated as follows
\vspace{-3mm}
\begin{equation}
\label{equ:PhiStr}
\Phi_{\rm NBS} \left( {m,n};{\bm \psi}_{\rm opt} \right) \propto {e^{j\left( {dm{k_{{\rm{opt}}}} + dn{l_{{\rm{opt}}}}} \right)}}
\text{.}
\vspace{-2mm}
\end{equation}
{This NBS design results in a narrowest beam pattern towards a certain transmission direction ${\bm \psi}_{\rm opt}$, and thus provides the optimal beamforming gain in the point-to-point communications.}

{\it b) Spatial bandpass filtering.} As far as SBF is concerned, we aim to design $|g\left( {{{\bm \psi} _{{\text{out}}}},{{\bm \psi} _{{\text{in}}}}} \right)|$ to be quasi-constant for $\psi_{\min }^{{\rm{azi}}} \le \psi_{\rm{out}}^{{\rm{azi}}} \le \theta _{\max }^{{\rm{azi}}}$ and $\psi_{\min }^{{\rm{ele}}} \le \psi_{\rm{out}}^{{\rm{ele}}} \le \theta _{\max }^{{\rm{ele}}}$, and to be almost zero, i.e., $|g\left( {{{\bm \psi} _{{\text{out}}}},{{\bm \psi} _{{\text{in}}}}} \right)| \approx 0$ otherwise, where ${{\bm \psi}_{{\rm{min}}}} = \left[ {\psi_{\min }^{{\rm{azi}}},\psi_{\min }^{{\rm{ele}}}} \right]^T$ and ${{\bm \psi}_{{\rm{max}}}}=\left[ {\psi _{\max }^{{\rm{azi}}},\psi _{\max }^{{\rm{ele}}}} \right]^T$ are referred to as the cut-off angles. To clearly explain the design of SBF, we first decompose the values of $k$ in one period as follows
\begin{equation}
\label{equ:k_pass}
k_{\min} \le k \le k_{\max} \text{,}
\end{equation}
\begin{equation}
\label{equ:k_stop}
a < k < k_{\min} {\text{ or }}k_{\max} < k \le a+\frac{{2\pi}}{d}\text{,}
\end{equation}
where $k_{\min} = \psi _{{\rm{min}}}^{{\rm{azi}}} - \psi _{{\rm{in}}}^{{\rm{azi}}}$, $k_{\max} = \psi _{{\rm{max}}}^{{\rm{azi}}} - \psi _{{\rm{in}}}^{{\rm{azi}}}$, and $a \in \mathbb{R}$ is an arbitrary number that satisfies the constraints $a < k_{\min}$ and $a+\frac{2\pi}{d} > k_{\max}$. Since we assume $d\le\lambda/2$, the existence of $a$ is guaranteed. Similarly, the values of $l$ in one period can be decomposed as follows
\begin{equation}
\label{equ:l_pass}
l_{\min} \le l \le l_{\max} \text{,}
\end{equation}
\begin{equation}
\label{equ:l_stop}
b < l < l_{\min} {\text{ or }} l_{\max} < l \le b + \frac{{2\pi}}{d}\text{,}
\end{equation}
where $l_{\min} = \psi _{{\rm{min}}}^{{\rm{ele}}} - \psi _{{\rm{in}}}^{{\rm{ele}}}$, $l_{\max} = \psi _{{\rm{max}}}^{{\rm{ele}}} - \psi _{{\rm{in}}}^{{\rm{ele}}}$, and $b \in \mathbb{R}$ is an arbitrary number that satisfies the constraints $b < l_{\min}$ and $b+\frac{2\pi}{d} > l_{\max}$. 

The angular-domain coefficients for SBF (in one period of $k$ and $l$) can be designed as follows
\begin{align}
\label{equ:OmegaSBF}
&{\omega _{{\rm{SBF}}}}\left( {k,l;{\bm \psi _{\min }},{\bm \psi _{\max }}}  \right) \nonumber \\
& \propto \left\{ {\begin{array}{*{20}{c}}
{\begin{array}{*{20}{c}}
{{e^{ - j\frac{{{N_x} + 1}}{2}dk}}{e^{ - j\frac{{{N_y} + 1}}{2}dl}},}\\
{0,}
\end{array}}&{\begin{array}{*{20}{l}}
{\text{for (\ref{equ:k_pass}) and (\ref{equ:l_pass})}}\\
{\text{for (\ref{equ:k_stop}) and (\ref{equ:l_stop})}}
\end{array}}
\end{array}} \right. \text{.}
\end{align}
Substituting (\ref{equ:OmegaSBF}) into (\ref{equ:g2}) and \eqref{equ:EffOmega}, and choosing the intervals of integration equal to $\left( {a,a+\frac{{2\pi}}{d}} \right]$ for $k$ and equal to $\left( {b,b+\frac{{2\pi}}{d}} \right]$ for $l$, we obtain the following beam pattern
\begin{align}
\label{equ:gSBF}
&{g_{{\rm{SBF}}}}\left( {{\bm \psi _{{\rm{out}}}},{\bm \psi _{{\rm{in}}}};{\bm \psi _{\min }},{\bm \psi _{\max }}} \right) \propto {e^{j\frac{{d - {A_x}}}{2}{\Delta _x}}}{e^{j\frac{{d - {A_y}}}{2}{\Delta _y}}} \nonumber \\
& \times \int_{l_{\min}}^{l_{\max}} {\int_{k_{\min}}^{k_{\max}} {{\Xi _{{N_x}}}\left[ {d\left( {k - {\Delta _x}} \right)} \right]} } {\Xi _{{N_y}}}\left[ {d\left( {l - {\Delta _y}} \right)} \right]{\rm{d}}k{\rm{d}}l \text{.}
\end{align}
The corresponding reflection coefficients can be calculated from (\ref{equ:2dDTFT}), which yields
\begin{align}
\label{equ:PhiSBF}
&{\Phi _{{\rm{SBF}}}}\left( {m,n} ; \bm \psi_{\min},\bm \psi_{\max}\right) \nonumber \\
& \propto \int_{l_{\min}}^{l_{\max}} {\int_{k_{\min}}^{k_{\max}} {e^{ - j\frac{{{N_x} + 1}}{2}dk}}{e^{ - j\frac{{{N_y} + 1}}{2}dl}}{{e^{j\left( {dmk + dnl} \right)}}{\rm{d}}k{\rm{d}}l} } \nonumber \\
& = {e^{ - jd\left( {{\bar m}\psi _{{\rm{in}}}^{{\rm{azi}}} + {\bar n}}\psi _{{\rm{in}}}^{{\rm{ele}}} \right)}}\frac{{{e^{jd{\bar m}\psi _{{\rm{max}}}^{{\rm{azi}}}}} - {e^{jd{\bar m}\psi _{\min }^{{\rm{azi}}}}}}}{{d{\bar m}}}\frac{{{e^{jd{\bar n}\psi _{{\rm{max}}}^{{\rm{ele}}}}} - {e^{jd{\bar n}\psi _{\min }^{{\rm{ele}}}}}}}{{d{\bar n}}}
 \text{,}
\end{align}
where ${\bar m}=m-\frac{N_x+1}{2}$ and ${\bar n}=n-\frac{N_y+1}{2}$.

To validate the proposed designs, we show two realizations of $|g_{\rm{NBS}}\left( {{{\bm \psi} _{{\text{out}}}},{{\bm \psi} _{{\text{in}}}}};{\bm \psi}_{\rm opt} \right)|$ and $|g_{\rm{SBF}}\left( {{{\bm \psi} _{{\text{out}}}},{{\bm \psi} _{{\text{in}}}}};{\bm \psi _{\min }},{\bm \psi _{\max }} \right)|$ with normalized amplitudes in Fig. \ref{fig:NBS3d} and Fig. \ref{fig:SBF3d}, respectively. It can be observed that the obtained beam patterns based on the proposed beamforming framework well fulfill the desired design. It is noteworthy that the reflection coefficients in (\ref{equ:PhiStr}) and (\ref{equ:PhiSBF}) are given in closed-form and are physically realizable. By direct inspection of Fig. \ref{fig:NBS3d} and Fig. \ref{fig:SBF3d}, we observe that (i) the NBS design criterion allows one to obtain a small focused region, which can be useful for beamforming applications; and (ii) the SBF design criterion allows one to obtain a wide focused region, which can be useful for broadcasting applications.
\begin{figure}[t]
%\captionsetup{font={footnotesize}, name = {Fig.}, labelsep = period}
\centering
\includegraphics[width=3in]{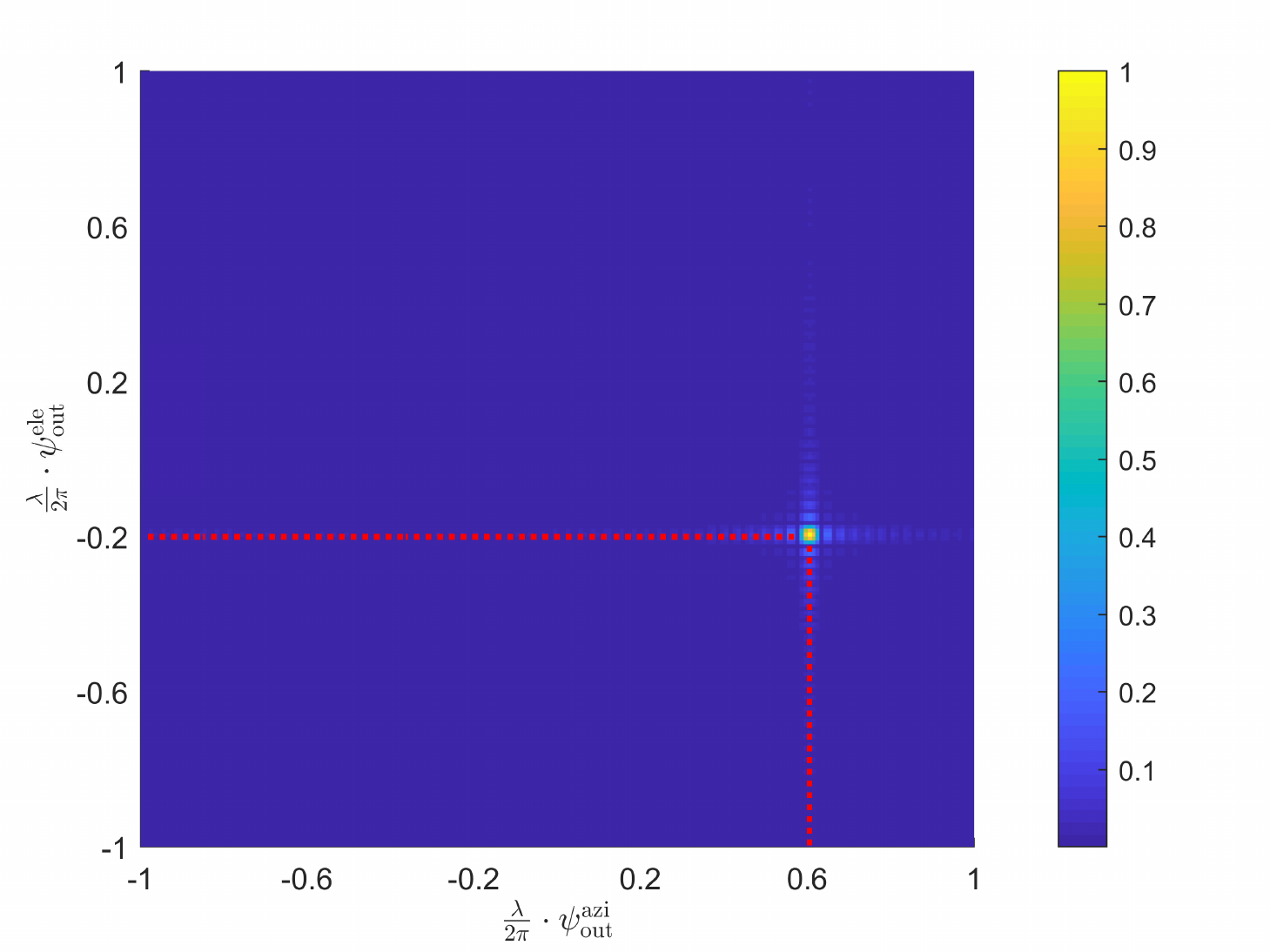}
\vspace{-3mm}
\caption{An example of beam pattern of NBS based on (\ref{equ:gNBS}). $N_x = N_y = 64$, $d = \lambda/2$. The desired angle is $\frac{\lambda}{2\pi}{\bm \psi}_{\rm{opt}}=\left[ 0.6,-0.2\right]^T$.}
\label{fig:NBS3d}
\end{figure}
\begin{figure}[t]
\centering
\includegraphics[width=3in]{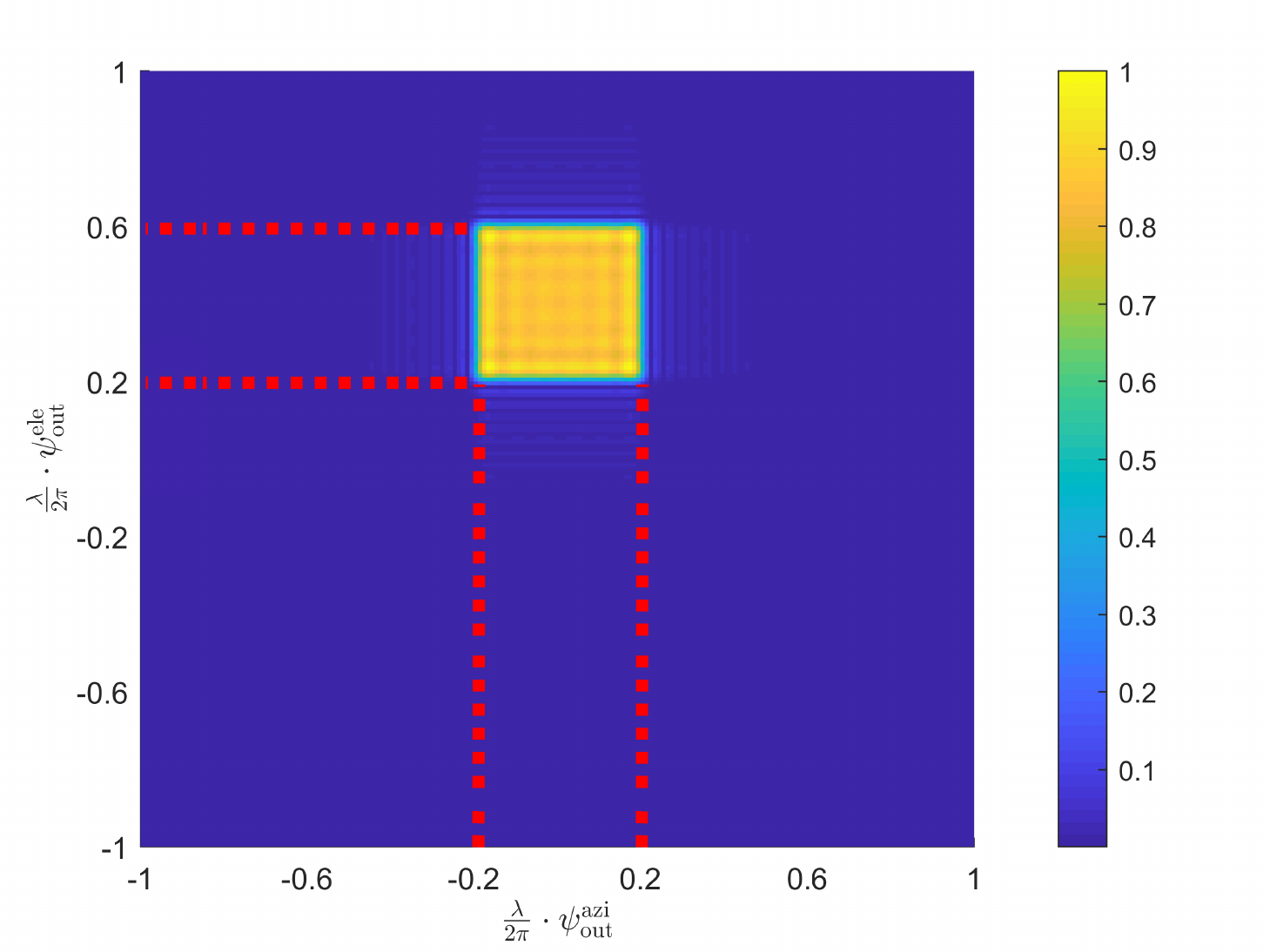}
\vspace{-3mm}
\caption{An example of beam pattern of SBF based on (\ref{equ:gSBF}). $N_x = N_y = 64$, $d = \lambda/2$. The cut-off angles are $\frac{\lambda}{2\pi}{\bm \psi}_{\min}=\left[ -0.2,0.2\right]^T$ and $\frac{\lambda}{2\pi}{\bm \psi}_{\max}=\left[ 0.2,0.6\right]^T$.}
\label{fig:SBF3d}
%\vspace{-8mm}
\end{figure}

\vspace{-3mm}
\subsection{Extension to Holographic RISs}
As far as the design of DPA-based RISs is concerned, we have considered an element spacing equal to $d=\lambda/2$, which is a {\it critical spacing} based on the Nyquist sampling theorem. This assumption has been, implicitly or explicitly, adopted in previous research works \cite{Model,Reviewer,THzIRS,RZhang2,MSlim,CHuang,FFFCE,FSFCE,XYuan,My1,SLiu,AA}. However, the design and optimization of RISs based on elements spaced at the critical distance has some inherent drawbacks. In particular, (i) due to the periodicity and the non-negligible sidelobes of the Dirichlet kernel functions, a {\it power leakage} phenomenon \cite{My2} is usually observed, which may result in inter-beam interference; (ii) the energy received and reflected by the RIS highly depends on its {\it effective reflection area} \cite{SJin}. More specifically, the use of critically-spaced RISs usually degrades the effective reflection area, which may result in a reduced energy efficiency.

Since the physical size of an RIS is limited by several practical factors, an effective solution to overcome the just mentioned drawbacks is to increase the number of reflecting elements and to reduce their spacing ($d<\lambda/2$) while keeping $A_x$ and $A_y$ unchanged, as illustrated in Fig. \ref{fig:model}(b). 
A DPA-based RIS whose elements have an element spacing $d<\lambda/2$ is referred to as {\it ultra-dense} RIS. 
It is worth mentioning that the analysis reported in the previous sub-section can be applied to critically-spaced and ultra-dense RISs by appropriately choosing the value of $d$. 
In Figs. \ref{fig:NBS} and \ref{fig:SBF}, to elucidate these aspects, we report the beam patterns of DPA-based RISs for different values of $d$ (only the beam patterns along the azimuth direction, i.e., with respect to $\bm \psi_{\rm{out}}^{\rm{azi}}$, are reported for ease of illustration).
It can be observed that the beam patterns of an ultra-dense RIS with $d=\{\lambda/4, \lambda/8\}$ are quite similar to those of a critically-spaced RIS ($d=\lambda/2$), but the power leakage is suppressed.
By increasing the number of reflecting elements, more importantly, the effective reflection area of an ultra-dense RIS is expected to increase as compared with critically-spaced RIS, and thus a larger fraction of the energy of the incident EM signal can be steered towards the desired direction.
This point is elaborated in detail in Section V, where the numerical results are presented.

The performance improvement from critically-spaced RISs to ultra-dense RISs naturally motivates us to ask: what is the performance of DPA-based RISs in the asymptotic regime $d \to 0$ and $N_x,N_y \to \infty$ (while keeping $N_xd = A_x$ and $N_yd = A_y$ fixed)? An RIS that is obtained by letting $d \to 0$ is referred to as CMS. In the following text, we show that the beam patterns and the design of the reflection coefficients of a CMS can be obtained by extending the analysis of DPA-based RISs.
Consider the CMS illustrated in Fig. \ref{fig:model}(c), where each point $(x,y)$, $x\in \left[ 0, A_x\right]$, $y\in \left[ 0, A_y\right]$ is capable of manipulating the phase and amplitude of the incident signal. The reflection coefficients of a CMS, which are denoted by $\tilde \Phi \left( {x,y} \right)$, are continuously distributed within $\left[ 0, A_x\right] \times \left[ 0, A_y\right]$. Based on the properties of the Fourier transform, the DTFT in (\ref{equ:2dDTFT}) and (\ref{equ:2dIDTFT}) tends to the continuous-time Fourier transformation (CTFT). In particular, we have
\begin{figure}[t]
%\captionsetup{font={footnotesize}, name = {Fig.}, labelsep = period}
\centering
\includegraphics[width=3in]{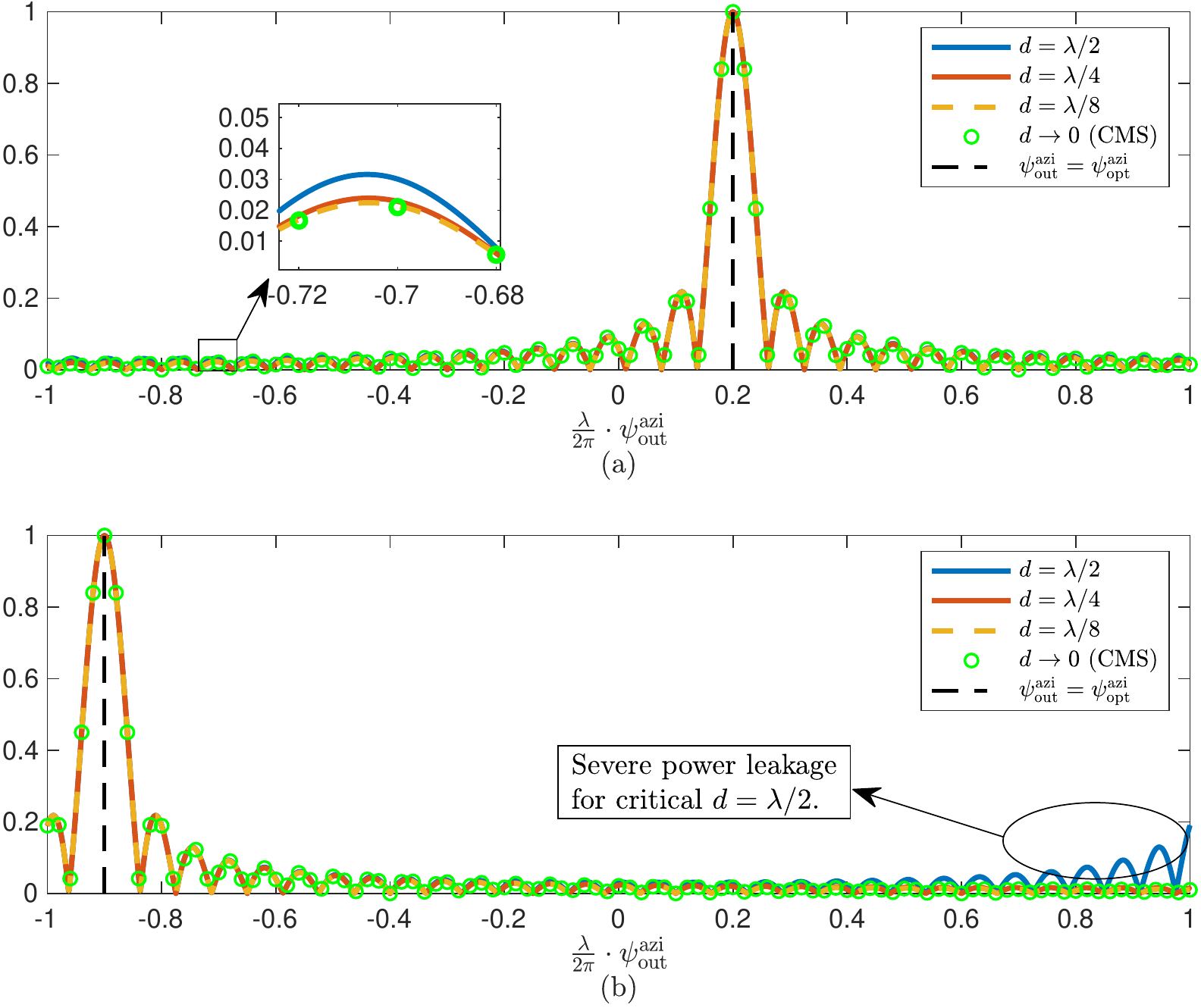}
\vspace{-3mm}
\caption{The beam patterns of NBS with normalized amplitudes when $A_x = 16\lambda$. (a) $\frac{\lambda}{2\pi}{\psi}^{\rm{azi}}_{\rm{opt}}=0.2$; (b) $\frac{\lambda}{2\pi}{\psi}^{\rm{azi}}_{\rm{opt}}=0.9$. Power leakage is observed in (b) for $d=\lambda/2$.}
\label{fig:NBS}
\end{figure}
\begin{figure}[t]
\centering
\includegraphics[width=3in]{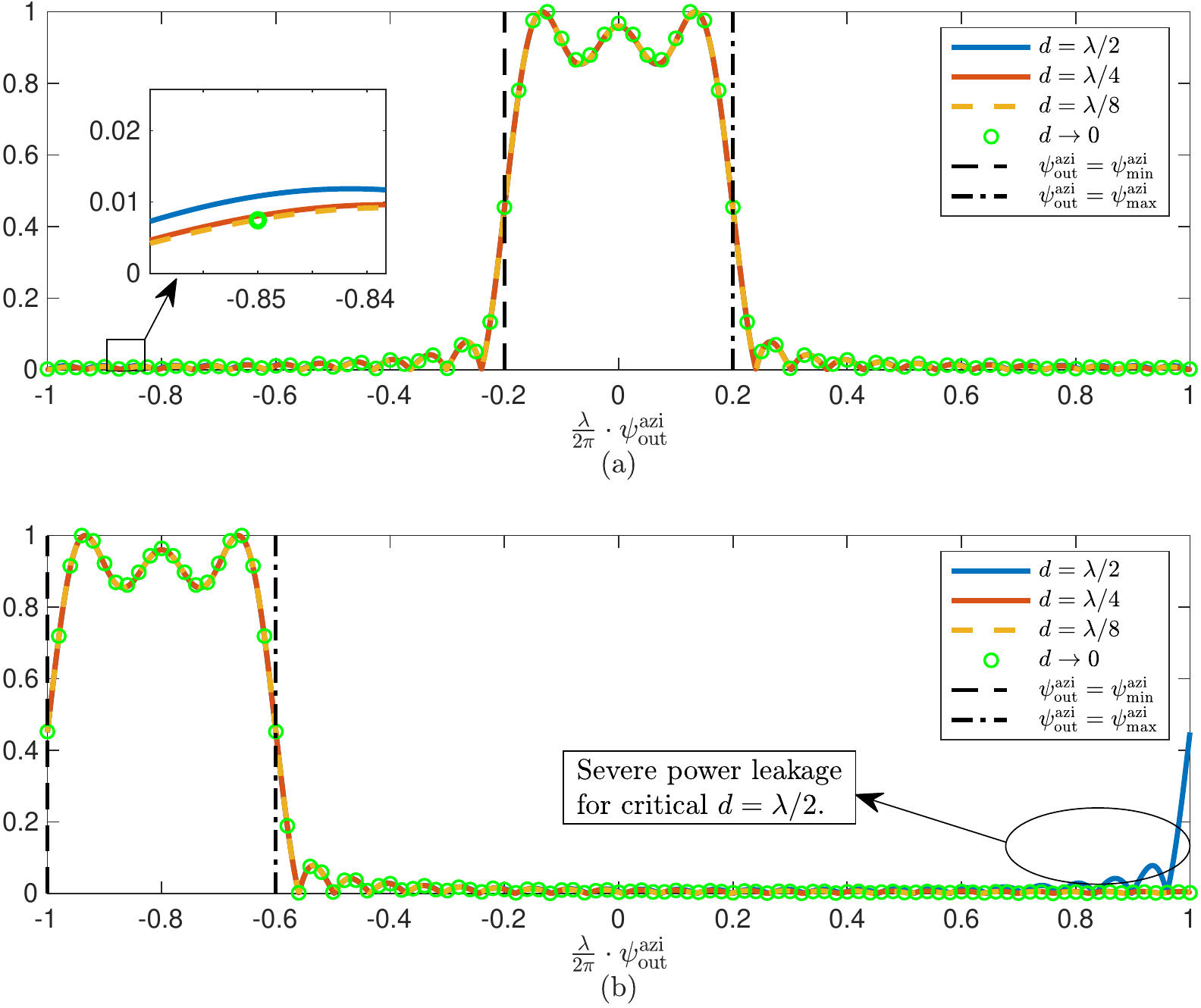}
\vspace{-3mm}
\caption{The beam patterns of SBF with normalized amplitudes when $A_x = 16\lambda$. (a) $\frac{\lambda}{2\pi}{\psi}^{\rm{azi}}_{\rm{min}}=-0.2$, $\frac{\lambda}{2\pi}{\psi}^{\rm{azi}}_{\rm{max}}=0.2$; (b) $\frac{\lambda}{2\pi}{\psi}^{\rm{azi}}_{\rm{min}}=-1$, $\frac{\lambda}{2\pi}{\psi}^{\rm{azi}}_{\rm{max}}=-0.6$. Power leakage is observed in (b) for $d=\lambda/2$.}
\label{fig:SBF}
%\vspace{-10mm}
\end{figure}
\begin{equation}
\label{equ:PhiCMS}
\tilde \Phi \left( {x,y} \right) \propto \lim_{d\to 0} \Phi \left( {m,n} \right)\left| {_{dm = x,dn = y}} \right. \text{, } x\in \left[ 0, A_x\right] \text{, } y\in \left[ 0, A_y\right] \text{.}
\end{equation}
The corresponding beam pattern of a CMS, $\tilde g\left( {{\bm \psi _{{\rm{out}}}},{\bm \psi _{{\rm{in}}}}} \right)$ can be calculated as
\begin{equation}
\label{equ:gCMS}
\tilde g\left( {{\bm \psi _{{\rm{out}}}},{\bm \psi _{{\rm{in}}}}} \right) \propto \lim_{d\to 0} g\left( {{\bm \psi _{{\rm{out}}}},{\bm \psi _{{\rm{in}}}}} \right) \text{.}
\end{equation}
In this sub-section, the parameters related to CMSs are top-marked with $\mathop {\left(  \cdot  \right)}\limits^\sim $ in order to distinguish them from DPA-based RISs. We note, in particular, that the constraints $N_x d=A_x$ and $N_y d = A_y$ need to be inherently enforced in (\ref{equ:PhiCMS}) and (\ref{equ:gCMS}). Based on (\ref{equ:PhiCMS}) and (\ref{equ:gCMS}), the next two corollaries report the beamforming design of NBS and SBF for application to CMSs.

\begin{corollary}
Consider a CMS based on the NBS-based beamforming design. The reflection coefficients and the beam pattern, which are denoted by $\tilde \Phi_{\rm NBS} \left( {x,y};{\bm \psi}_{\rm opt} \right)$ and $\tilde g_{\rm{NBS}}\left( {{{\bm \psi} _{{\rm{out}}}},{{\bm \psi} _{{\rm{in}}}}};{\bm \psi}_{\rm opt} \right)$, respectively, can be formulated as
\begin{equation}
\label{equ:PhiCMSNBS}
\tilde \Phi_{\rm NBS} \left( {x,y};{\bm \psi}_{\rm opt} \right) \propto {e^{j\left( {x{k_{{\rm{opt}}}} + y{l_{{\rm{opt}}}}} \right)}}
\text{,}
\end{equation}
\begin{align}
\label{equ:gCMSNBS}
&\tilde g_{\rm{NBS}}\left( {{{\bm \psi} _{{\rm{out}}}},{{\bm \psi} _{{\rm{in}}}}};{\bm \psi}_{\rm opt} \right) \propto {e^{-j\frac{{A_x}}{2}{\Delta _x}}}{e^{-j\frac{A_y}{2}{\Delta _y}}}  \nonumber \\
& \quad \times {\rm sinc}\left[ \frac{A_x}{2} \left( {k_{\rm{opt}} -\Delta_x} \right) \right] {\rm sinc}\left[ \frac{A_y}{2} \left( {l_{\rm{opt}} - \Delta_y} \right) \right] \text{.}
\end{align}
\end{corollary}
%\vspace{-3mm}
\begin{proof}
See {\bf Appendix A}.
%\vspace{-3mm}
\end{proof}

\begin{corollary}
Consider a CMS based on the SBF-based beamforming design. The reflection coefficients, $\tilde \Phi_{\rm SBF} \left( {x,y};{\bm \psi}_{\rm opt} \right)$, can be formulated as
\begin{align}
\label{equ:PhiCMSSBF}
&\tilde \Phi_{\rm SBF} \left( {x,y};{\bm \psi}_{\min}, {\bm \psi}_{\max} \right) \propto \nonumber \\
& {e^{ - j\left( {{\bar x}\psi _{{\rm{in}}}^{{\rm{azi}}} + {\bar y}\psi _{{\rm{in}}}^{{\rm{ele}}}} \right)}}\frac{{{e^{j{\bar x}\psi _{{\rm{max}}}^{{\rm{azi}}}}} - {e^{j{\bar x}\psi _{\min }^{{\rm{azi}}}}}}}{{{\bar x}}}\frac{{{e^{j{\bar y}\psi _{{\rm{max}}}^{{\rm{azi}}}}} - {e^{j{\bar y}\psi _{\min }^{{\rm{azi}}}}}}}{{{\bar y}}}
\text{,}
\end{align}
where $\bar x = x - A_x/2$ and $\bar y = y - A_y/2$. The corresponding beam pattern, $\tilde g_{\rm{SBF}}\left( {{{\bm \psi} _{{\rm{out}}}},{{\bm \psi} _{{\rm{in}}}}};{\bm \psi}_{\min}, {\bm \psi}_{\max} \right)$, can be formulated as follows

%\vspace{-8mm}
\begin{align}
\label{equ:gCMSSBF}
&\tilde g_{\rm{SBF}}\left( {{{\bm \psi} _{{\rm{out}}}},{{\bm \psi} _{{\rm{in}}}}};{\bm \psi}_{\min}, {\bm \psi}_{\max} \right) \propto {e^{-j\frac{{A_x}}{2}{\Delta _x}}}{e^{-j\frac{{A_y}}{2}{\Delta _y}}} \nonumber \\
& \times \int_{l_{\min}}^{l_{\max}} {\int_{k_{\min}}^{k_{\max}} {\rm sinc}\left[ \frac{A_x}{2} \left( {k -\Delta_x} \right) \right] {\rm sinc}\left[ \frac{A_y}{2} \left( {l - \Delta_y} \right) \right] } {\rm{d}}k{\rm{d}}l \text{.}
\end{align}
\end{corollary}

\begin{proof}
The proof of {\bf Corollary 2} is similar to that of {\bf Corollary 1} in {\bf Appendix A}. 
%Thus, it is omitted for brevity.
\end{proof}

In Figs. \ref{fig:NBS} and \ref{fig:SBF}, we report the beam patterns of a CMS and compare them with those of a critically-spaced RIS and an ultra-dense RIS.
As far as CMSs are concerned, we observe that the Dirichlet kernel functions that characterize the beam patterns of critically-spaced and ultra-dense RISs are replaced by the ``sinc'' functions, which have relatively small sidelobes and no periodicity compared with Dirichlet kernel functions. 
This is in agreement with the far-field results in \cite{Y}.
\begin{figure*}[t]
%\vspace{-3mm}
%\captionsetup{font={footnotesize}, name = {Fig.}, labelsep = period}
%\captionsetup[subfigure]{singlelinecheck = on, justification = raggedright, font={footnotesize}}
\centering
\includegraphics[width=5in]{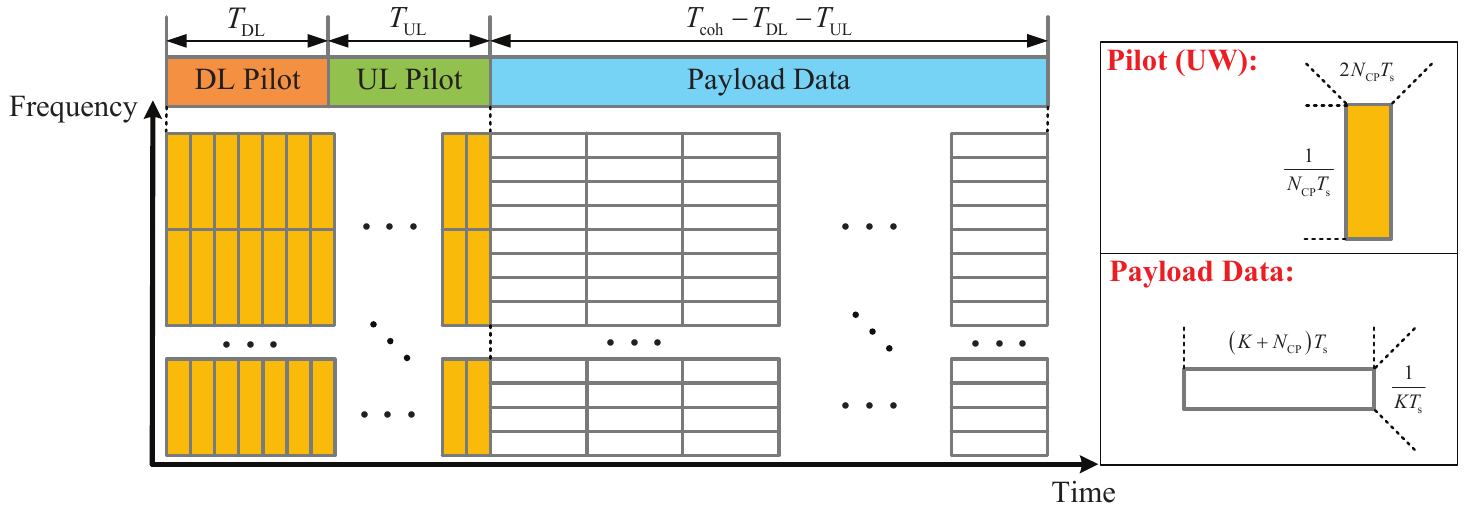}
\caption{The transmission frame structure of the system. $T_{\rm DL}$($T_{\rm UL}$): Pilot overhead of downlink (uplink) CE stage. $T_{\rm coh}$: Channel coherence time, during which the CSI is assumed to be invariant.}
\label{fig:frame}
\end{figure*}

Although CMSs are not realizable in practice, since it is not possible to build surfaces with an infinite number of reflecting elements, Figs. \ref{fig:NBS} and \ref{fig:SBF} show that the beam pattern of a CMS can be well approximated by an ultra-dense RIS, which is a practical extension of conventional critically-spaced RISs, especially in the THz band. Based on the examples illustrated in Figs. \ref{fig:NBS} and \ref{fig:SBF}, we observe that a spacing $d \le \lambda/4$ makes the beam patterns of ultra-dense RISs and CMSs almost indistinguishable from each other (see also Fig. 6 in \cite{Y}). In this context, a holographic RIS can be defined as a CMS in theory and as an ultra-dense RIS in practice.

As illustrated in Fig. \ref{fig:NBS}, the minimum width of the mainlobe of the NBS beam pattern is mainly determined by the physical size, i.e., $A_x$ and $A_y$, of the surface. We refer to this minimum width as {\it spatial resolution}. This parameter characterizes the ability of RISs of distinguishing different UEs that are closely located. Based on Fig. \ref{fig:NBS}, we observe that the spatial resolution is the same for the three implementations, i.e., critically-spaced RISs, ultra-dense RISs, and CMSs, as long as their physical size is kept unchanged \cite{Holo2}.

%\vspace{-3mm}
\section{Closed-Loop Channel Estimation Scheme}
In this section, we investigate the CE problem for the holographic RIS described in the previous section and for application to THz massive MIMO systems, as illustrated in Fig. 1.
The transmission frame structure is illustrated in Fig. \ref{fig:frame}. During the whole CE stage, short-length OFDM symbols, which are also referred to as unique words (UWs) \cite{UW}, are transmitted as pilot signals.
Each UW consists of $N_{\rm CP}$ subcarriers in the frequency domain and has length (duration) equal to $2{N_{{\rm{CP}}}}{T_{\rm{s}}}$ (half of which constitutes the mainbody of $N_{\rm CP}$-point OFDM symbol while the rest constitutes the $N_{\rm CP}$-point CP) in the time domain, as elaborated in Fig. \ref{fig:frame}. 
For each RFC at the BS, the beamforming is assumed to be the NBS in order to obtain the high beamforming gain, and it is obtained by assuming that the LoS direction of the BS-RIS channel is known. In particular, the beamforming design of each RFC at the BS is{\footnote{The proposed beamforming design can be applied to the active beamforming for MIMO systems as well. In this case, we can set $\bm \psi_{\rm in} = [0,0]^T$, and $\Phi(m,n)$ can be considered as the phase shift induced by the phase shifter.}}
%\vspace{-2mm}
\begin{equation}
\label{equ:R1}
\Phi_r^{{\rm{B}}}\left( {m,n} \right) = \Phi_{{\rm{NBS}}}^{\rm{B}}\left( {m,n;{\bm \psi^{{\rm{B}}}}} \right) \text{, } \forall r = 1,..., N_{\rm RF} \text{.}
%\vspace{-2mm}
\end{equation}
The specific expression of $\Phi_{{\rm{NBS}}}^{\rm{B}}$ is obtained from (\ref{equ:PhiStr}). It is worth noting that the constraint of constant modulus in the hybrid analog-digital architecture is implicitly satisfied when the NBS beamforming is considered. The signals of all the RFCs are added together and a single data stream is obtained. Thus, by dropping the index $r$ in (\ref{equ:DLmodel1}), we re-write it as
%\vspace{-2mm}
\begin{equation}
\label{equ:htau}
h^{\rm DL}(\tau) = \alpha^{\rm DL} G^{\rm B} \left[ {{h^{{\rm{DL,LoS}}}}\left( \tau  \right) + {h^{{\rm{DL,NLoS}}}}\left( \tau  \right)} \right] \text{,}
%\vspace{-2mm}
\end{equation}
where $G^{\rm B} = g_{\rm NBS}^{{\rm{B}}}\left({{\bm \psi^{{\rm{B}}}}};{{\bm \psi^{{\rm{B}}}}} \right)$ is the beam pattern corresponding to $\Phi_{{\rm{NBS}}}^{\rm{B}}\left( {m,n;{\bm \psi^{{\rm{B}}}}} \right)$. 

The proposed CS scheme consists of two phases, which are applied to the downlink and uplink transmissions. The two phases are described in the following two sub-sections, respectively. 
\vspace{-3mm}
\subsection{Downlink CE Stage}
During the downlink CE stage, the UE coarsely estimates $\bm \mu^{\rm LoS}$ and $\bm \nu^{\rm LoS}$ by constraining their values within smaller ranges. {First, the whole range of AoD at the RIS is divided into $G_x$ groups along the azimuth direction and into $G_y$ groups along the elevation direction.} The range of azimuth-AoD in the $g_x$-th azimuth group, $g_x = 1,...,G_x$, is $[\psi _{{\rm{min,}}{g_x}}^{{\rm{azi}}} , \psi _{{\rm{max,}}{g_x}}^{{\rm{azi}}} ]$ with
\begin{align}
\label{equ:COazi}
\begin{split}
\psi_{{\rm{min,}}{g_x}}^{{\rm{azi}}} = \frac{{2\pi }}{\lambda }\left[ { - 1 + \frac{{2\left( {{g_x} - 1} \right)}}{{{G_x}}}} \right] \text{, } \\ 
\psi _{{\rm{max,}}{g_x}}^{{\rm{azi}}} = \frac{{2\pi }}{\lambda }\left[ { - 1 + \frac{{2{g_x}}}{{{G_x}}} - \frac{\lambda }{{{A_x}}}} \right] \text{,}
\end{split}
\end{align}
where $\frac{\lambda}{A_x}$ is the azimuth resolution of the RIS. Without loss of generality, a gap equal to $\frac{\lambda}{A_x}$ between two adjacent groups is assumed. Similarly, the range of elevation-AoD in the $g_y$-th elevation group, $g_y = 1,...,G_y$, is $[\psi _{{\rm{min,}}{g_y}}^{{\rm{ele}}},\psi _{{\rm{max,}}{g_y}}^{{\rm{ele}}}]$, where
\begin{equation}
\label{equ:COele}
\begin{split}
\psi _{{\rm{min,}}{g_y}}^{{\rm{ele}}} = \frac{{2\pi }}{\lambda }\left[ { - 1 + \frac{{2\left( {{g_y} - 1} \right)}}{{{G_y}}}} \right] \text{, } \\
\psi _{{\rm{max,}}{g_y}}^{{\rm{ele}}} = \frac{{2\pi }}{\lambda }\left[ { - 1 + \frac{{2{g_y}}}{{{G_y}}} - \frac{\lambda }{{{A_y}}}} \right] \text{.}
\end{split}
\end{equation}
Then, we aim to find the group that contains the LoS AoD of the RIS-UE channel, ${\bm \mu}_{\rm LoS}$. In particular, the downlink CE stage can be formulated as follows
\begin{align}
\text{Find }(\hat g_x,\hat g_y) & = \Big\{ ({g_x},{g_y})\left| \psi _{\min ,{g_x}}^{{\rm{azi}}} \le \mu _{{\rm{LoS}}}^{{\rm{azi}}} \le \psi _{\max ,{g_x}}^{{\rm{azi}}} \right. \nonumber \\
& \qquad {\text{ and }} \psi _{\min ,{g_y}}^{{\rm{ele}}} \le \mu _{{\rm{LoS}}}^{{\rm{ele}}} \le \psi _{\max ,{g_y}}^{{\rm{ele}}} \Big\} \text{.}
\end{align}

To this end, for the $(g_x,g_y)$-th group, the RIS uses the SBF beamforming with the cut-off angles given in (\ref{equ:COazi}) and (\ref{equ:COele}), that is
\vspace{-2mm}
\begin{equation}
\label{equ:temp1}
\tilde \Phi(x,y) = \tilde \Phi_{\rm SBF}(x,y;\bm \psi_{{\rm min},g_x,g_y},\bm \psi_{{\rm max},g_x,g_y}) \text{,}
\vspace{-2mm}
\end{equation}
where $\bm \psi_{{\rm min},g_x,g_y} = \left[ \psi _{{\rm{min,}}{g_x}}^{{\rm{azi}}}, \psi _{{\rm{min,}}{g_y}}^{{\rm{ele}}} \right]^T$, $\bm \psi_{{\rm min},g_x,g_y} = \left[ \psi _{{\rm{max,}}{g_x}}^{{\rm{azi}}}, \psi _{{\rm{max,}}{g_y}}^{{\rm{ele}}} \right]^T$, and the specific expression of \eqref{equ:temp1} is obtained from (\ref{equ:PhiSBF}) (for the DPA-based RIS) or from (\ref{equ:PhiCMSSBF}) (for the CMS).

At the UE, given that the dimension of the antenna array at the UE is relatively small, the NBS beamforming towards $M^{\rm U}$ desired directions is employed to coarsely estimate $\bm \nu_{\rm LoS}$. Specifically, for the $(n_x,n_y)$-th desired direction, $1 \le n_x \le M^{\rm U}_x$, $1 \le n_y \le M^{\rm U}_y$, the UE uses the beamforming design
%\vspace{-2mm}
\begin{equation}
\label{equ:temp2}
\Phi^{{\rm{U}}}\left( {m,n} \right) = \Phi_{{\rm{NBS}}}^{\rm{U}}\left( {m,n;{\bm \psi_{{\rm opt},n_x,n_y}}} \right) \text{,}
\end{equation}
where ${\bm \psi_{{\rm opt},n_x,n_y}} = \frac{{2\pi }}{\lambda }{\left[ { - 1 + \frac{{2\left( {{n_x} - 1} \right)}}{{{M^{\rm U}_x}}}, - 1 + \frac{{2\left( {{n_y} - 1} \right)}}{{{M^{\rm U}_y}}}} \right]^T}$.

In an OFDM-based system, the effective baseband channel in the delay domain (\ref{equ:htau}) can be transformed to the frequency domain as follows

\vspace{-6mm}
\begin{align}
\label{equ:temp3}
h^{{\rm{DL,Fd}}}_k & = \frac{1}{{\sqrt{N_{\rm CP} }}}\sum\limits_{d = 0}^{N_{\rm CP} - 1} {{h}^{\rm DL}\left( {d{T_s}} \right){e^{ - j\frac{{2\pi d}}{N_{\rm CP}}k}}} \nonumber \\
& = \gamma^{\rm DL, LoS}_k \tilde g\left( {{\bm \mu^{\rm LoS}},{\bm \psi^{{\rm{R}}}}} \right) g^{{\rm{U}}}\left( {{\bm \nu^{\rm LoS}}}\right) \nonumber \\
& \quad + \frac{{1}}{{\sqrt {L{K_{\rm{f}}}} }}\sum\limits_{l = 1}^L {\gamma^{\rm DL}_{l,k} \tilde g\left( {{\bm \mu_{l}},{\bm \psi^{{\rm{R}}}}} \right)} g^{{\rm{U}}}\left( {{\bm \nu_{l}}}\right) \text{,}
\end{align}
where $1 \le k \le N_{\rm CP}$, $h_k^{\rm DL,Fd}$ is the frequency-domain (Fd) channel in the $k$-th subcarrier, ${\gamma^{\rm DL,LoS}_k} = \frac{{\alpha^{\rm DL} \beta^{\rm DL,LoS} G^{\rm B}}}{{\sqrt {{N_{{\rm{CP}}}}} }}\sum\limits_{d = 0}^{{N_{{\rm{CP}}}} - 1} {p\left( {d{T_{\rm{s}}} - {\tau^{\rm DL,LoS}}} \right)} {e^{ - j\frac{{2\pi d}}{{{N_{{\rm{CP}}}}}}k}}$, and ${\gamma^{\rm DL}_{l,k}} = \frac{{\alpha^{\rm DL} \beta_{l}^{\rm DL} {G^{{\rm{B}}}}}}{{\sqrt {{N_{{\rm{CP}}}}} }}\sum\limits_{d = 0}^{{N_{{\rm{CP}}}} - 1}$ ${p\left( {d{T_{\rm{s}}} - {\tau^{\rm DL}_l}} \right)} {e^{ - j\frac{{2\pi d}}{{{N_{{\rm{CP}}}}}}k}}$.
Substituting the beamforming designs (\ref{equ:temp1}) and (\ref{equ:temp2}) into (\ref{equ:temp3}), the received pilot signal at the UE in the $k$-th subcarrier can be expressed as
\begin{align}
\label{equ:DLRxpilot}
& \qquad {y_{k,{g_x},{g_y},{n_x},{n_y}}} \nonumber \\
& = \sqrt {\frac{{P_{{\rm{Tx}}}^{{\rm{DL}}}}}{{{N_{{\rm{CP}}}}}}} \Bigg( \gamma_k^{{\rm{DL,LoS}}}\tilde G_{{g_x},{g_y}}^{{\rm{LoS}}}G_{{n_x},{n_y}}^{{\rm{U,LoS}}} \nonumber \\
& + \frac{1}{{\sqrt {L{K_{\rm{f}}}} }}\sum\limits_{l = 1}^L {{\gamma^{\rm DL}_{l,k}}{{\tilde G}_{l,{g_x},{g_y}}}G_{l,{n_x},{n_y}}^{{\rm{U}}}}  \Bigg) + {n_{k,{g_x},{g_y},{n_x},{n_y}}}  \text{,}
\end{align}
where $P_{\rm Tx}^{\rm DL}$ is the total downlink transmit power, ${n_{k,{g_x},{g_y},{n_x},{n_y}}}$ is the additive white Gaussian noise (AWGN) with distribution ${\mathcal{CN}}(0,\sigma^2_{\rm n})$, $G_{{n_x},{n_y}}^{{\rm{U,LoS}}} = g_{{\rm{NBS}}}^{{\rm{U}}}\left( {{\bm \nu^{{\rm{LoS}}}};{\bm \psi _{{\rm{opt}},{n_x},{n_y}}}} \right)$, $G_{l,{n_x},{n_y}}^{{\rm{U}}} = g_{{\rm{NBS}}}^{{\rm{U}}}\left( {{\bm \nu_{l}};{\bm \psi _{{\rm{opt}},{n_x},{n_y}}}} \right)$,
$\tilde G_{{g_x},{g_y}}^{{\rm{LoS}}} = {{\tilde g}_{{\rm{SBF}}}}\left( {{\bm \mu^{{\rm{LoS}}}},{\bm \psi^{{\rm{R}}}};{\bm \psi _{\min ,{g_x},{g_y}}},{\bm \psi _{\max ,{g_x},{g_y}}}} \right)$, 
and
$\tilde G_{l,{g_x},{g_y}} = {{\tilde g}_{{\rm{SBF}}}} \left( {\bm \mu_{l}}, {\bm \psi^{{\rm{R}}}}; {\bm \psi _{\min ,{g_x},{g_y}}},{\bm \psi _{\max ,{g_x},{g_y}}} \right)$.
In particular, $\tilde g_{\rm SBF}$ and $g_{\rm SBF}^{\rm U}$ are the beam patterns corresponding to $\tilde \Phi_{\rm SBF}$ in (\ref{equ:temp1}) and $\Phi_{\rm NBS}^{\rm U}$ in (\ref{equ:temp2}), respectively.

After collecting ${y_{k,{g_x},{g_y},{n_x},{n_y}}}$ for $1 \le g_x \le G_x$, $1 \le g_y \le G_y$, $1 \le n_x \le M^{\rm U}_x$ and $1 \le n_y \le M^{\rm U}_y$ in successive UWs, the UE  conducts a search to decide which groups $\bm \mu^{\rm LoS}$ and $\bm \nu^{\rm LoS}$ belong to, i.e.,
\begin{equation}
\label{equ:FindMax}
\left({\hat g_x,\hat g_y,\hat n_x,\hat n_y} \right) = \mathop {\arg \max }\limits_{\left( {{g_x},{g_y},{n_x},{n_y}} \right)} \sum\limits_{k = 1}\limits^{N_{\rm CP}} {\left| {{y_{k,{g_x},{g_y},{n_x},{n_y}}}} \right|} \text{.}
\end{equation}
Considering the energy focusing property of NBS (as shown in Fig. \ref{fig:NBS3d}) and the bandpass property of SBF (as shown in Fig. \ref{fig:SBF3d}), we can expect that $\bm \nu^{\rm LoS} \approx \bm \psi_{{\rm opt},\hat n_x, \hat n_y}$, $\psi^{\rm azi}_{\min,\hat g_x, \hat g_y} \le \mu^{\rm azi}_{\rm LoS} \le \psi^{\rm azi}_{\max,\hat g_x, \hat g_y}$, and $\psi^{\rm ele}_{\min,\hat g_x, \hat g_y} \le \mu^{\rm ele}_{\rm LoS} \le \psi^{\rm ele}_{\max,\hat g_x, \hat g_y}$. In other words, after the downlink CE stage, the range of possible values for $\bm \mu_{\rm LoS}$ is narrowed down to $\psi^{\rm azi}_{\max,\hat g_x} - \psi^{\rm azi}_{\min,\hat g_x} = \frac{2\pi}{\lambda}\left(\frac{2}{G_x}-\frac{\lambda}{A_x}\right)$ along the azimuth direction,
and to $\psi^{\rm ele}_{\max,\hat g_y} - \psi^{\rm ele}_{\min,\hat g_y} = \frac{2\pi}{\lambda}\left(\frac{2}{G_y}-\frac{\lambda}{A_y}\right)$ along the elevation direction.
This significantly reduces the search space of the uplink finer-grained CE discussed in the next sub-section. The proposed downlink CE scheme is summarized in {\bf Algorithm 1}.
\begin{algorithm}[t]
\caption{The downlink CE scheme}
\label{alg:alg1}
\begin{algorithmic}[1]
\renewcommand{\algorithmicrequire}{\textbf{Input:}}
\renewcommand{\algorithmicensure}{\textbf{Output:}}
\State {Determine $\bm \psi^{\rm B}$, $\bm \psi^{\rm R}$, $G_x$ and $G_y$.}
\State Set the beamforming design for each RFC at the BS as in (\ref{equ:R1});
%\Ensure {The combiner ${{\bf{W}}_i^{\rm{d}}}$ at the BS.}
\For{the $(g_x,g_y)$-th group at the RIS, $1 \le g_x \le G_x$, $1 \le g_y \le G_y$,}
    \State Set the beamforming design at the RIS as in (\ref{equ:temp1});
    \For{ the $(n_x,n_y)$-th group at the UE, $1 \le n_x \le M^{\rm U}_x$, $ 1\le n_y \le M^{\rm U}_y$,}
    		\State Set the beamforming design at each UE as in (\ref{equ:temp2});   		
    		\State The BS broadcasts one UW as pilot signals;
		\State Each UE receives the pilot signals ${y_{k,{g_x},{g_y},{n_x},{n_y}}}$ as in (\ref{equ:DLRxpilot});
    \EndFor
\EndFor
\State Each UE obtains its index of optimal groups $(\hat g_x,\hat g_y,\hat n_x,\hat n_y)$ via (\ref{equ:FindMax});
\end{algorithmic}
\end{algorithm}

\vspace{-5mm}
\subsection{Uplink CE stage}
After the downlink CE stage, each UE obtains the indices of the optimal groups $(\hat g_x,\hat g_y,\hat n_x,\hat n_y)$, and this information is fed back to the BS via the control links for UE scheduling. Specifically, the BS schedules the UEs having the same $(\hat g_x,\hat g_y)$ into the same {\it uplink group}. Each scheduled group performs the finer-grained uplink CE and the subsequent payload data transmission. In this sub-section, therefore, we consider a generic group as follows
\vspace{-2mm}
\begin{align}
{\cal{G}}_{\hat g_x,\hat g_y} & = \big\{\text{All the UEs that have the same } \nonumber \\
& \qquad (\hat g_x,\hat g_y) \text{ obtained from the donwlink CE}\big\} \text{,}
\vspace{-2mm}
\end{align}
and we assume {${\rm card}({\cal{G}}_{\hat g_x,\hat g_y}) = N_{\rm UE}$}. During the uplink CE and payload data transmission stages, the beamforming design at the $u$-th UE ($1 \le u \le N_{\rm UE}$) in ${\cal{G}}_{\hat g_x,\hat g_y}$ is chosen as follows
\vspace{-2mm}
\begin{equation}
\label{equ:temp4}
\Phi^{{\rm{U}}}_u\left( {m,n} \right) = \Phi_{{\rm{NBS}}}^{\rm{U}}\left( {m,n;{\bm \psi_{{\rm opt},\hat n^u_x,\hat n^u_y}}} \right) \text{,}
\vspace{-2mm}
\end{equation}
where $\{\hat n^u_x,\hat n^u_y\}$ is the index of the optimal group obtained by the $u$-th UE during the downlink CE stage. 

As mentioned in the previous sub-section, after the downlink CE stage, the LoS angle of the $u$-th UE in ${\cal{G}}_{\hat g_x,\hat g_y}$, i.e., ${\bm \mu}^{\rm LoS}_{u}$, is coarsely estimated and its value is confined to the range between $\bm \psi_{{\rm min},\hat g_x,\hat g_y} = \left[ \psi _{{\rm{min,}}{\hat g_x}}^{{\rm{azi}}}, \psi _{{\rm{min,}}{\hat g_y}}^{{\rm{ele}}} \right]^T$ and $\bm \psi_{{\rm max},\hat g_x,\hat g_y} = \left[ \psi _{{\rm{max,}}{\hat g_x}}^{{\rm{azi}}}, \psi _{{\rm{max,}}{\hat g_y}}^{{\rm{ele}}} \right]^T$. Hence, during the uplink CE stage, we only need to search within this range in order to determine the LoS angle with a finer resolution. Given that the spatial resolution of the considered RIS is $\lambda/A_x$ along the azimuth direction and $\lambda/A_y$ along the elevation direction, we define the search space for the uplink CE stage as
\begin{equation}
\label{equ:temp6}
{\zeta}_{b_x}^{\rm azi} = \psi _{{\rm{min,}}{\hat g_x}}^{{\rm{azi}}} + \frac{2\pi}{\lambda} \frac{(b_x-1)\lambda}{A_x} \text{, } 1 \le b_x \le B_x \text{,}
\end{equation}
\begin{equation}
\label{equ:temp7}
{\zeta}_{b_y}^{\rm ele} = \psi _{{\rm{min,}}{\hat g_y}}^{{\rm{ele}}} + \frac{2\pi}{\lambda} \frac{(b_y-1)\lambda}{A_y} \text{, } 1 \le b_y \le B_y \text{,}
\end{equation}
where $B_x \buildrel \Delta \over = \frac{\lambda}{2\pi} \frac{\psi _{{\rm{max,}}{g_x}}^{{\rm{azi}}}-\psi _{{\rm{min,}}{g_x}}^{{\rm{azi}}}}{\lambda/A_x}+1 = \frac{2A_x}{\lambda G_x}$, $B_y \buildrel \Delta \over = \frac{\lambda}{2\pi} \frac{\psi _{{\rm{max,}}{g_y}}^{{\rm{ele}}}-\psi _{{\rm{min,}}{g_y}}^{{\rm{ele}}}}{\lambda/A_y}+1 = \frac{2A_y}{\lambda G_y}$ ($B_x$,$B_y \in \mathbb{Z}$ without loss of generality). 
In particular, the target of the uplink finer-grained CE is to solve the following optimization problem
%\vspace{-3mm}
\begin{equation}
({\hat b_x},{\hat b_y}) = \mathop {\arg \max }\limits_{({b_x},{b_y})} {\tilde g_{{\rm{NBS}}}}\left( {\bm \mu^{{\rm{LoS}}}_u,{\bm \psi^{{\rm{R}}}};{\bm \zeta _{{b_x},{b_y}}}} \right) \text{,}
%\vspace{-2mm}
\end{equation}
where ${\bm \zeta_{b_x,b_y}} = \left[\zeta_{b_x}^{\rm azi},\zeta_{b_y}^{\rm ele}\right]^T$.
Instead of executing an exhaustive beam scanning over the $B_xB_y$ directions in \eqref{equ:temp6} and \eqref{equ:temp7}, during the $i$-th time slot of the uplink CE stage, the RIS employs the {\it overlapped NBS beamforming} towards all $B_xB_y$ directions with different random phases. This can be formulated as follows
\vspace{-2mm}
\begin{equation}
\label{equ:temp5}
\tilde \Phi(x,y) = \frac{1}{\sqrt{B_xB_y}}\sum\limits_{{b_x} = 1}^{{B_x}} {\sum\limits_{{b_y} = 1}^{{B_y}} {{e^{j{\theta _{i,{b_x},{b_y}}}}}\tilde \Phi _{{\rm{NBS}}}^{}\left( {x,y;{\bm \zeta _{{b_x},{b_y}}}} \right)} } \text{,}
\vspace{-2mm}
\end{equation}
where $\theta _{i,{b_x},{b_y}} \sim {\cal U}[0,2\pi)$.

By using the beamforming designs in \eqref{equ:temp4} and \eqref{equ:temp5}, and by capitalizing on the channel reciprocity between the downlink and uplink transmissions, the uplink effective baseband channel related to the $u$-th UE in the $i$-th time slot can be formulated as
\begin{equation}
\label{equ:ULmodel1}
{h^{\rm UL}_{i,u}}\left( \tau  \right) =  {h_{i,u}^{{\rm{UL,LoS}}}\left( \tau  \right) + h_{i,u}^{{\rm{UL,NLoS}}}\left( \tau  \right)} \text{,}
\end{equation}
which consists of the LoS part and the NLoS part. 
${h_{u,i}^{{\rm{UL,LoS}}}\left( \tau  \right)}$ can be expressed as \eqref{equ:ULmodel2}, as shown at the top of the next page,
\begin{figure}[b]
%\captionsetup{font={footnotesize}, name = {Fig.}, labelsep = period}
\centering
\includegraphics[width=3in]{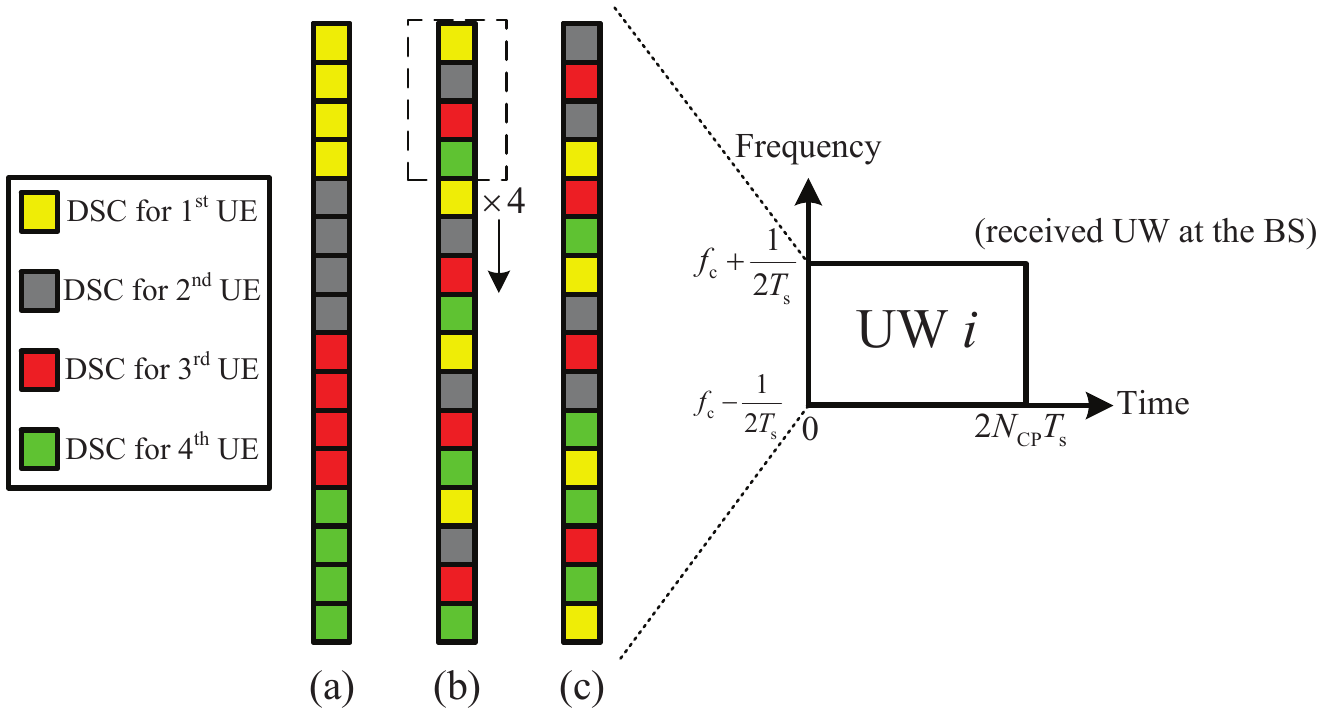}
%\vspace{-5mm}
\caption{Examples of the received UW at the BS with DSCs, where ${{N_{{\rm{CP}}}} = 16}$ and ${{N_{{\rm{UE}}}} = 4}$. Three different DSC allocation schemes are illustrated. (a) Block allocation; (b) Uniform allocation; (c) Random allocation.}
\label{fig:UW}
%\vspace{-8mm}
\end{figure}
%\vspace{-2mm}
\begin{figure*}
\begin{align}
\label{equ:ULmodel2}
h_{i,u}^{{\rm{UL,LoS}}}\left( \tau  \right) & = \alpha^{\rm UL} \beta^{\rm UL,LoS}  {G^{{\rm{B}}}} \sum\limits_{{b_x} = 1}^{{B_x}} {\sum\limits_{{b_y} = 1}^{{B_y}} {\frac{{{e^{j{\theta _{i,{b_x},{b_y}}}}}}}{{\sqrt {{B_x}{B_y}} }}{{\tilde g}_{{\rm{NBS}}}}\left( {\bm \mu^{{\rm{LoS}}}_u,{\bm \psi^{{\rm{R}}}};{\bm \zeta _{{b_x},{b_y}}}} \right)} } g_{{\rm{NBS}}}^{{\rm{U}}}\left( {\bm \nu^{{\rm{LoS}}}_u;{\bm \psi _{{\rm{opt}},\hat n_x^u,\hat n_y^u}}} \right)p\left( {\tau  - {{\tau}_u^{{\rm{UL,LoS}}}}} \right) \nonumber \\
& = {\bf{w}}_i^T{\bf{h}}_u^{{\rm{UL,LoS}}} p\left( {\tau  - {\tau_u^{{\rm{UL,LoS}}}}} \right) \text{,}
\vspace{-2mm}
\end{align}
\begin{align}
\label{equ:hLoS}
{\bf{h}}_u^{{\rm{UL,LoS}}} = & \alpha^{\rm UL} \beta^{\rm UL,LoS} {G^{{\rm{B}}}} g_{{\rm{NBS}}}^{{\rm{U}}}\left( {\bm \nu^{\rm LoS}_u;{\bm \psi _{{\rm{opt}},\hat n_x^u,\hat n_y^u}}} \right) \times \nonumber \\
& \left[ {{{\tilde g}_{{\rm{NBS}}}}\left( {\bm \mu^{\rm LoS}_u,{\bm \psi^{{\rm{R}}}};{\bm \zeta _{1,1}}} \right)},...,{{{\tilde g}_{{\rm{NBS}}}}\left( {\bm \mu^{\rm LoS}_u,{\bm \psi^{{\rm{R}}}};{\bm \zeta _{b_x,{b_y}}}} \right)},...,{{{\tilde g}_{{\rm{NBS}}}}\left( {\bm \mu^{\rm LoS}_u,{\bm \psi^{{\rm{R}}}};{\bm \zeta _{{B_x},{B_y}}}} \right)} \right]^T \in \mathbb{C}^{B_xB_y \times 1} \text{.}
\end{align}
\hrule
\end{figure*}
where ${{\bf{w}}_i} = \frac{1}{{\sqrt {{B_x}{B_y}} }}{\left[ {{e^{j{\theta _{i,1,1}}}},...,{e^{j{\theta _{i,1,{B_y}}}}},{e^{j{\theta _{i,2,1}}}},...,{e^{j{\theta _{i,{B_x},{B_y}}}}}} \right]^T} \in \mathbb{C}^{B_xB_y \times 1}$, and ${\bf{h}}_u^{{\rm{UL,LoS}}}$, as shown in \eqref{equ:hLoS} at the top of the next page, is the effective angular-domain LoS channel with finer angular resolution. The elements in ${\bf h}_u^{\rm UL,LoS}$ show the relation between the LoS angle ${\bm \mu}^{\rm LoS}_u$ and all the pre-defined codewords ${\bm \zeta}_{b_x,b_y}$.

\vspace{1mm}
Further, ${h_{i,u}^{{\rm{UL,NLoS}}}\left( \tau  \right)}$ in \eqref{equ:ULmodel1} can be written as follows
\vspace{-2mm}
\begin{equation}
\label{equ:ULmodel3}
h _{i,u}^{{\rm{UL,NLoS}}}\left( \tau  \right)  = {\bf{w}}_i^T\sum\limits_{l = 1}^L {{\bf{h}}_{u,l}^{{\rm{UL,NLoS}}}}  p\left( {\tau  - {\tau_{u,l}^{\rm UL}}} \right) \text{,}
\end{equation}
where ${{\bf{h}}_{u,l}^{{\rm{UL,NLoS}}}}$ is the effective angular-domain channel that corresponds to the $l$-th NLoS path. The formulation of ${{\bf{h}}_{u,l}^{{\rm{UL,NLoS}}}}$ is similar to \eqref{equ:hLoS} and thus it is omitted for brevity.

In order to simultaneously perform the uplink CE for all the UEs in ${\cal{G}}_{\hat g_x,\hat g_y}$ and to avoid the interference of different UEs' pilot signals at the BS, we consider the set of dedicated subcarriers (DSCs) for the $u$-th UE as follows
\begin{equation}
{{\cal K}_u} = \left\{ {{k_{u,n}}\left| {n = 1,....,{N_{{\rm{used}}}}} \right.} \right\} \text{,}
\end{equation}
where ${N_{{\rm{used}}}} \buildrel \Delta \over = {N_{{\rm{CP}}}}/{N_{{\rm{UE}}}}$ (${N_{{\rm{used}}}} \in \mathbb{Z}$ without loss of generality), ${{\cal K}_u} \subseteq \left\{ {1,...,{N_{{\rm{CP}}}}} \right\}$, {${\rm card}({{{\cal K}_u}} ) = {N_{{\rm{used}}}}$}, and ${{\cal K}_{{u}}} \cap {{\cal K}_{{u'}}} = \emptyset$ for $\forall {u} \ne {u'}$. During the uplink CE stage, the $u$-th UE transmits its own pilot signals by using only the $N_{\rm used}$ subcarriers indexed by ${{\cal K}_u}$ out of the $N_{\rm CP}$ available subcarriers. On the other hand, no signals are transmitted in other subcarriers. This makes easier to separate the pilot signals of different UEs at the BSs, since different DSCs are used. Three examples of the possible structure of an UW with DSCs are illustrated in Fig. \ref{fig:UW}.

Taking into account the DSC allocation above, the uplink channel in the delay domain \eqref{equ:ULmodel1} can be transformed to the frequency domain, which is similar to \eqref{equ:temp3}. In particular, the frequency-domain representation of the channel related to the $u$-th UE corresponding to its $n$-th DSC in the $i$-th time slot can be written as
\begin{align}
h_{i,u,n}^{{\rm{UL,Fd}}} & = \frac{1}{{\sqrt {{N_{{\rm{CP}}}}} }}\sum\limits_{d = 0}^{{N_{{\rm{CP}}}} - 1} {{{h}^{\rm UL}_{i,u}}\left( {d{T_{\rm{s}}}} \right)} {e^{ - j\frac{{2\pi }}{{{N_{{\rm{CP}}}}}}{k_{u,n}}}} = {\bf{w}}_i^T{\bf{H}}_u^{{\rm{AdDd}}}{{\bf{f}}_{{k_{u,n}}}} \text{,}
\end{align}
where ${{\bf{f}}_{{k_{u,n}}}}$ is the ${{k_{u,n}}}$-th column vector of the $N_{\rm CP} \times N_{\rm CP}$ DFT matrix ${{\bf{F}}_{{N_{{\rm{CP}}}}}}$, and
\begin{equation}
\label{equ:HAdDd}
{\bf{H}}_u^{{\rm{AdDd}}} = \left[ {{{\bf{h}}_1},...,{{\bf{h}}_{{N_{{\rm{CP}}}}}}} \right] \in \mathbb{C}^{B_xB_y \times N_{\rm CP}}
\end{equation}
with ${{\bf{h}}_t} = {\bf{h}}_u^{{\rm{UL,LoS}}} p\left[ {\left( {t - 1} \right) {T_{\rm{s}}} - \tau^{{\rm{UL,LoS}}}_u} \right] + \sum\limits_{l = 1}^L {{\bf{h}}_{u,l}^{{\rm{UL,NLoS}}} p\left[ {\left( {t - 1} \right) {T_{\rm{s}}} - \tau_{u,l}^{\rm UL}} \right]}$, $t = 1,..,N_{\rm CP}$, being the effective angular-domain and delay-domain (AdDd) uplink channel to be estimated.

We assume that each UE transmits the pilot signal $\sqrt{\frac{P^{\rm UL}_{\rm Tx}}{N_{\rm used}}}$ over its DSCs during the uplink CE stage, where $P^{\rm UL}_{\rm Tx}$ is the total transmit power of the UE. The received pilot signal in the $i$-th time slot at the BS can be expressed as
\begin{equation}
y_{i,u,n} = \sqrt {\frac{{P_{{\rm{Tx}}}^{{\rm{UL}}}}}{{{N_{{\rm{used}}}}}}} h_{i,u,n}^{{\rm{UL,Fd}}} + n_{i,u,n} = \sqrt {\frac{{P_{{\rm{Tx}}}^{{\rm{UL}}}}}{{{N_{{\rm{used}}}}}}}  {\bf{w}}_i^T{\bf{H}}_u^{{\rm{AdDd}}}{{\bf{f}}_{{k_{u,n}}}} + n_{i,u,n} \text{,}
\vspace{-1.8mm}
\end{equation}
where $n_{i,u,n} \sim {\cal{CN}}(0,\sigma^2_{\rm n})$ is the AWGN.
By collecting the received pilot signals of all the DSCs for the $u$-th UE $\left\{ {{{\bf{y}}_{i,u,n}}} \right\}_{n = 1}^{{N_{{\rm{used}}}}}$, we have
\begin{equation}
\label{equ:inA2}
{\bf y}^T_{i,u} = \sqrt {\frac{{P_{{\rm{Tx}}}^{{\rm{UL}}}}}{{{N_{{\rm{used}}}}}}} {\bf{w}}_i^T{\bf{H}}_u^{{\rm{AdDd}}}{{\bf{F}}_u} + {\bf n}^T_{i,u} \text{,}
\end{equation}
where ${\bf y}_{i,u} = \left[ y_{i,u,1},...,y_{i,u,N_{\rm used}} \right]^T\in \mathbb{C}^{N_{\rm used}\times 1}$, ${\bf n}_{i,u} = \left[ n_{i,u,1},...,n_{i,u,N_{\rm used}} \right]^T \in \mathbb{C}^{N_{\rm used}\times 1}$, and ${\bf F}_{u} = {\left[ {{{\bf{F}}_{{N_{{\rm{CP}}}}}}} \right]_{{{\cal K}_u}}} \in \mathbb{C}^{N_{\rm CP} \times N_{\rm used}}$ is the partial DFT matrix.
For ${N_{\rm{P}}}$ successive time slots, we aggregate the channel observations of the $u$-th UE $\left\{ {{{\bf{y}}^T_{i,u}}} \right\}_{i = 1}^{{N_{\rm{P}}}}$ into the matrix ${{\bf{Y}}_u} \in \mathbb{C}^{N_{\rm P}\times N_{\rm used}}$, which can be formulated as follows
\begin{equation}
\label{equ:CSpro}
{{\bf{Y}}_u} = {\left[ {{{\bf{y}}_{1,u}},...,{{\bf{y}}_{{N_{\rm{P}}},u}}} \right]^T} = {\bf{W}}{\bf{H}}_u^{{\rm{AdDd}}}{{\bf{F}}_u} + {\bf N}_{u} \text{,}
\end{equation}
where ${\bf{W}} =\sqrt{\frac{ P^{\rm UL}_{\rm Tx}}{N_{\rm used}}} \left[ {\bf{w}}_1,...,{\bf{w}}_{N_{\rm P}} \right]^T \in \mathbb{C}^{N_{\rm P}\times B_xB_y}$ and ${\bf N}_{u} = \left[ {\bf{n}}_{1,u},...,{\bf{n}}_{N_{\rm P},u} \right]^T\in \mathbb{C}^{N_{\rm P}\times N_{\rm used}}$. The objective of the uplink CE stage is therefore, to estimate ${\bf{H}}_u^{{\rm{AdDd}}}$ by exploiting the knowledge of ${\bf W}$, ${\bf F}_u$, and the noisy matrix ${\bf Y}_u$.
Usually, we have ${N_{\rm{P}}} < B_xB_y$ due to the limited channel coherence time, and ${N_{\rm{used}}} < {N_{\rm{CP}}}$ due to the allocation of different DSCs among multiple UEs.
These constraints make (\ref{equ:CSpro}) an under-determined system, which is not possible to solve by using traditional estimation techniques such as the LS estimator \cite{FFFCE,FSFCE}. 
Fortunately, due to the strong LoS link between the UE and the RIS, the channel matrix ${\bf{H}}^{\rm AdDd}_u$ is expected to be sparse, and, in particular, only the elements whose indices fulfill the conditions ${\bm \mu}^{\rm LoS}_u \approx {\bm \zeta_{b_x,b_y}}$, $b_x \in \{1,...,B_x\}$, $b_y \in \{1,...,B_y\}$, and $\tau^{{\rm{UL,LoS}}}_u \approx (t-1)T_{\rm s}$, $t \in \{1,...,N_{\rm CP}\}$ have a non-negligible absolute value. On the other hand, the other entries have a much smaller absolute value, which can be safely ignored. This property of ${\bf{H}}^{\rm AdDd}_u$ is referred to as the dual sparsity in both the angular domain and delay domain, which can be exploited to solve the CE problem.
%In particular, we use CS methods to estimate ${\bf{H}}^{\rm AdDd}_u$ directly from the under-determined measurements in \eqref{equ:CSpro}.
%by solving the following optimization problem
%\begin{equation}
%\label{equ:OptPro1}
%\mathop {\min }\limits_{{\bf{H}}^{\rm AdDd}_u}  {{{\left\| {\rm vec}({\bf{H}}^{\rm AdDd}_u) \right\|}_0}}
%\end{equation}
%\begin{equation}
%\label{equ:OptPro2}
%\text{s.t. } {\left\| {{{\rm vec}({\bf Y}_u)} - ({\bf{F}}_u^T \otimes {\bf{W}}) {\rm{vec(}}{{\bf{H}}^{\rm AdDd}_u}{\rm{)}} } \right\|_2^2}  < \varepsilon \text{,}
%\end{equation}
%where $\varepsilon$ is a threshold for stop criterion. 

Based on these considerations, {\bf Algorithm 2} reports the details of the proposed uplink CE scheme, where the CS-based orthogonal matching pursuit (OMP) algorithm is adopted to recover sparse channels based on \eqref{equ:CSpro}. Once the estimated channel ${\bf{\hat H}}^{\rm AdDd}_u$ is obtained, interpolation-based methods can be applied to reconstruct the channels in all $K$ subcarriers \cite{FSFCE}.

{{\it Remark}: The proposed closed-loop CE framework can be extended to the case where the LoS paths between the BS and the UEs exist.
To this end, the RIS can be configured in order to reduce the scattering from it (i.e., absorption state \cite{Holo4}). Under this configuration, the UEs do not receive the signal reflected by the RIS. Therefore, the direct channel between the BS and the UEs can be estimated by using state-of-the-art CE algorithms (see, e.g., \cite{korean,GaoTSP,GaoCL,CE1,ALiao}). It is worth mentioning that the proposed closed-loop CE scheme can be applied to conventional MIMO (in the absence of RISs) communications in order to estimate the direct channels between the BS and the UEs.
Once the direct link is estimated, it can be removed from the signal received in the presence of the RIS, and the proposed CE algorithm can be applied to estimate the channel between the RIS and the UEs.
}
\begin{algorithm}[t]
\caption{The uplink CE scheme}
\label{alg:alg2}
\begin{algorithmic}[1]
\renewcommand{\algorithmicrequire}{\textbf{Input:}}
\renewcommand{\algorithmicensure}{\textbf{Output:}}
\State {Determine ${\cal G}_{\hat g_x,\hat g_y}$, and the DSC allocation scheme ${\cal K}_u$, $\forall u \in \{1,...,N_{\rm UE}\}$ for $N_{\rm UE}$ UEs;}
\State Set the beamforming design at each UE as in (\ref{equ:temp4});
\For{the $i$-th time slot, $1 \le i \le N_{\rm P}$,}
	\State Generate ${\bf w}_i$ and accordingly set the beamforming design at the RIS as in (\ref{equ:temp5});		
    	\State All $N_{\rm UE}$ UEs transmit the UWs with pilot signals by using the assigned DSCs;
	\State The BS receives the pilot signals ${\bf y}_{i,u}$ from each UE in (\ref{equ:inA2});
\EndFor
\State Compute ${\bf W}$, ${\bf Y}_u$, and ${\bf F}_u$ in \eqref{equ:CSpro}, $\forall u \in \{1,...,N_{\rm UE}\}$;
\State {\it \% CS-based CE algorithm below. Take a certain $u$-th UE as an example.}
\State {\bf Initialization}: $\text{iter} = 0$, $\mathcal{I} = $ empty set, ${{\bf{r}}} = {\rm vec}({\bf{Y}}_u)$, ${\bf{\hat H}}$ is an all-zero matrix of size $B_xB_y \times N_{\rm CP}$, $N_{\max}$ is the maximum number of iterations;
\While{$\text{iter} < N_{\max}$,}
    \State ${i^*} = \mathop {\arg \max } \limits_i  \left| {{{\left[ { {{{\left( {{\bf{F}}_u^T \otimes {\bf{W}}} \right)}^H}{{\bf{r}}}} } \right]}_i}} \right|$;
    \State ${{\mathcal{I}}}={{\mathcal{I}}} \cup \{{{i}^{*}}\}$;
    \State ${\bf{\hat h}}_{\rm temp} = \left[ {{\bf{F}}^T_u \otimes {\bf{W}}} \right]_{\cal I}^\dag {\rm vec}({\bf{Y}}_u)$;
    \State ${{\bf{r}}} = {\rm vec}({\bf{Y}}_u) - \left[ {{\bf{F}}^T_u \otimes {\bf{W}}}\right]_{\cal I}{{\bf{\hat h}}_{\rm temp}}$;
    \State $\text{iter} = \text{iter} + 1$;
\EndWhile
\State ${\bf{\hat h}} = {\rm vec}({\bf{\hat H}})$;
\State $[{\bf{\hat h}}]_{\cal I} = {\bf{\hat h}}_{\rm temp}$;
\State ${\bf{\hat H}}_u^{\rm AdDd}$ = ${\rm vec}^{-1}({\bf{\hat h}})$;
\Ensure {The estimated channels ${\bf{\hat H}}_u^{\rm AdDd}$, $\forall u \in \{1,...,N_{\rm UE}\}$;}
\end{algorithmic}
\end{algorithm}

{\vspace{-3mm}
\subsection{Pilot Overhead and Computational Complexity Analysis}
In this subsection, we analyze pilot overhead and the computational complexity of the proposed closed-loop CE scheme. As far as the pilot overhead is concerned, we evince from {\bf Algorithm 1} and {\bf Algorithm 2} that $G_xG_yM^{\rm U}_xM^{\rm U}_y$ and $N_{\rm P}$ UWs are required for the downlink and uplink CE stages, respectively. Therefore, the total required pilot overhead of the proposed CE scheme is $T_{\rm DL} + T_{\rm UL} = 2N_{\rm CP}T_{\rm s} (G_xG_yM^{\rm U}_xM^{\rm U}_y+N_{\rm P})$, where $2N_{\rm CP}T_{\rm s}$ is the length (duration) of one UW according to Fig. \ref{fig:frame}. 

The computational complexity of the proposed scheme consists of the following two parts: 

{\it 1) Downlink computational complexity}. This is determined by the search at each UE in order to determine the angular groups that the LoS angles belong to, as detailed in \eqref{equ:FindMax}. In this case, each UE needs to find the index of the maximum among $G_x G_y M^{\rm U}_x M^{\rm U}_y$ signals, and, therefore, the computational complexity at each UE is $\textsf{O}(G_x G_y M^{\rm U}_x M^{\rm U}_y)$, where $\textsf{O}(N)$ stands for ``of the order of $N$''. It is noteworthy that this complexity is affordable even for energy-constrained UEs, since $M^{\rm U}_x$, $M^{\rm U}_y$, $G_x$ and $G_y$ are much smaller than the numbers of elements available at the RIS and BS.

{\it 2) Uplink computational complexity}. This is mainly determined by the uplink CE algorithm. 
Since the CE problem has been formulated in \eqref{equ:CSpro} as a sparse signal recovery problem via under-determined measurements, various off-the-shelf algorithms, such as greedy algorithms (e.g., \cite{My1,THzIRS}), Bayesian algorithms (e.g., \cite{XYuan}), and deep learning methods (e.g., \cite{SLiu}), can be used for CE. The corresponding computational complexity may, therefore, vary significantly. As far as the OMP algorithm adopted in {\bf Algorithm 2} is concerned, we consider the total number of complex-valued multiplications to evaluate the computational complexity, as listed in \cite[Table I]{My2}. The specific results are presented in the next section.
}

\vspace{-3mm}
\section{Simulation Results}
In this section, we present numerical results to evaluate the different types of RISs, and the performance of the proposed CE scheme.
\vspace{-4mm}
\subsection{Experimental Setting}
We consider the system model as shown in Fig. \ref{fig:simulation}. The BS and the RIS serve the active UEs distributed within a sector of radius $R$ and central angle ${120^ \circ }$. The BS and the RIS are elevated to the height of $h_1$ and the UEs have the height of $h_2$. We assume that the normal directions of the arrays at the BS and RIS point towards each other, which yields ${\bm \psi}^{\rm B} = {\bm \psi}^{\rm R} = \left[0,0\right]^T$. The simulation setup is detailed as follows unless stated otherwise: $R = 20$ m, $h_1 = 10$ m, $h_2 = 1.5$ m, $M^{\rm B}_x = M^{\rm B}_y = 64$, $N_{\rm RF} = 4$, $M^{\rm U}_x = M^{\rm U}_y = 8$, $A_x = A_y = 0.2$ m, $f_{\rm c} = 0.15$ THz ($\lambda \approx 2$ mm), $T_{\rm s} = 2\times 10^{-9}$ sec (bandwidth ${\rm BW} = 1/T_{\rm s} = 500$ MHz), $N_{\rm CP} = 64$, $K = 256$. A raised cosine filter $p(\tau)$ with roll-off factor $0.8$ is employed. The noise power spectrum density at the receiver is $\sigma _{{\rm{NSD}}}^2 = -174$ dBm/Hz. Thus, the power of the AWGN $\sigma_{\rm n}^2$ is $\sigma_{\rm n}^2 = \sigma _{{\rm{NSD}}}^2 \times {\rm BW} \approx -87$ dBm. The number of iterations $N_{\max}$ in {\bf Algorithm 2} is set to $20$.
\begin{figure}[b]
%\vspace{-10mm}
\centering
%\captionsetup{font={footnotesize}, name = {Fig.}, labelsep = period}
\subfigure[]{\includegraphics[width=1.5in]{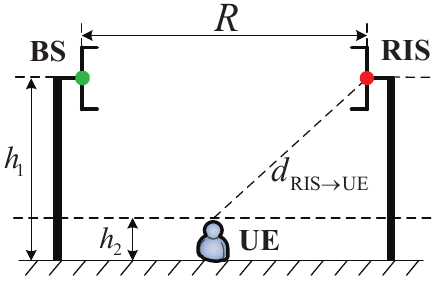}}
\hspace{5mm}
\subfigure[]{\includegraphics[width=1.5in]{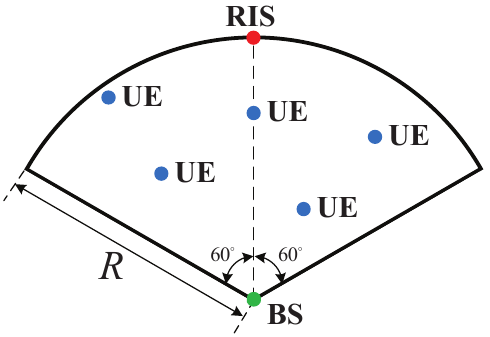}}
%\vspace{-5mm}
\caption{RIS-aided THz massive MIMO systems. (a) side view; (b) top view.}
\label{fig:simulation}
%\vspace{-8mm}
\end{figure}

As for the downlink channel model in \eqref{equ:DLmodel1}, we set $L = 1$, $K_{\rm f} = 30$ dB. The angle ${\bm \mu}^{\rm LoS}$ is calculated according to the position of the UE, and the angles ${\bm \nu}^{\rm LoS}$, ${\bm \mu}_{l}$ and ${\bm \nu}_{l}$ are randomly generated. The delay offsets $\tau^{\rm UL,LoS}$ and $\tau^{\rm DL}_{l}$ ($\tau^{\rm UL,LoS} < \tau^{\rm DL}_{l}$) follow a uniform distribution $\mathcal{U}(0,(N_{\rm CP}-1)T_{\rm s})$. The channel coefficients $\alpha^{\rm DL}$ and $\beta^{\rm DL,LoS}$ ($\beta^{\rm DL}_l$ can be similarly modeled) are modeled as follows \cite[Eq. (20)]{SJin}
\begin{align}
\alpha^{\rm DL} & = e^{j \theta_{\alpha}} \sqrt{\frac{G_{\rm Tx} S_{\rm eff}}{4 \pi R^2 A_{\rm abs}(f_{\rm c},R)}}  \text{,} \\
\beta^{\rm DL,LoS} & = e^{j \theta_{\beta}} \sqrt{\frac{G_{\rm RIS}G_{\rm Rx}}{A_{\rm abs}(f_{\rm c},d_{\rm RIS \to UE})}} \frac{\lambda }{{4\pi d_{\rm RIS \to UE}}} \text{,}
\end{align}
where $\theta_{\alpha},\theta_{\beta} \sim \mathcal{U}[0,2\pi)$ are the phase shifts introduced by the channels, $S_{{\rm{eff}}}$ is the effective reflection area of the RIS, $d_{\rm RIS \to UE}$ is the distance between the UE and the RIS, $G_{\rm Tx}$, $G_{\rm RIS}$, and $G_{\rm Rx}$ are the array gains of the BS, the RIS, and the UE, respectively{\footnote{The array gain can be calculated as $G = \frac{4\pi}{{\int_{\theta _{{\rm{out}}}^{{\rm{azi}}} = {\rm{0}}}^{{\rm{2}}\pi } {\int_{\theta _{{\rm{out}}}^{{\rm{ele}}} = {\rm{0}}}^{\pi {\rm{/2}}} {\left| {g\left( {{{\bm \psi} _{{\rm{out}}}}} \right)} \right|} } ^{\rm{2}}}\sin \theta _{{\rm{out}}}^{{\rm{ele}}}{\rm{d}}\theta _{{\rm{out}}}^{{\rm{ele}}}{\rm{d}}\theta _{{\rm{out}}}^{{\rm{azi}}}}$ \cite{SJin}, where ${g\left( {{{\bm \psi} _{{\rm{out}}}}} \right)}$ is the beam pattern of the array that is discussed in previous text.}}, and $A_{\rm abs}(f_{\rm c},d_{\rm Tx \to Rx})$ is the attenuation caused by molecular absorption \cite{THz1,THz2}. $A_{\rm abs}(f_{\rm c},d_{\rm Tx \to Rx})$ is related to the carrier-frequency $f_{\rm c}$ and the transmission distance $d_{\rm Tx \to Rx}$, and its specific values are obtained based on the recommendations of the International Telecommunications Union (ITU) \cite{ITU}.
%By taking into account that the physical size of the RIS ($0.5m \times 0.5m$) is much smaller than the transmission distance ($R = 20m$),
The effective reflection area $S_{{\rm{eff}}}$ of the RIS can be modeled as its whole physical area (aperture) \cite{SJin}
\begin{equation}
\label{equ:Seff}
{S_{{\rm{eff}}}} = \left\{ {\begin{array}{*{20}{l}}
{{A_x}{A_y},}&\text{for CMS,}\\
{\frac{{{A_x}{A_y}}}{{{d^2}}}{S_{{\rm{ele}}}},}&\text{for DPA-based RIS with spacing $d$,}
\end{array}} \right.
\end{equation}
where $S_{{\rm{ele}}} \le d^2$ is the physical size of a single reflection element in a DPA-based RIS. We assume $S_{{\rm{ele}}} \approx 200$ $\mu$m $\times$ $200$ $\mu$m as demonstrated in \cite[Fig. 1]{THzIRS}.
%Based on Fig. \ref{fig:model}, we assume $d \geq \sqrt{S_{\rm ele}}$.
The parameters of the uplink channel in \eqref{equ:ULmodel1}-\eqref{equ:ULmodel3} can be similarly modeled and thus are omitted for brevity.

\vspace{-5mm}
\subsection{Numerical Results}
Fig. \ref{fig:CC} shows the required pilot overhead and computational complexity versus the number of groups $\{G_x,G_y\}$ when $N_{\rm P}=40$. It is observed that, by setting different groups $\{G_x,G_y\}$, the proposed closed-loop scheme provides a trade-off between the pilot overhead and the computational complexity. Two cases deserve further attention in Fig. \ref{fig:CC}: (i) the case of $G_x=G_y=1$ refers to estimating the complete channels in the uplink (in an open-loop manner). In this case, the uplink channel matrix to be estimated in \eqref{equ:HAdDd} has a total size of $40,000 \times 64$, which causes an unaffordable computational complexity and storage burden \cite{ALiao}; and (ii) the case of $G_x=2A_x/\lambda = 200$ and $G_y = 2A_y/\lambda = 200$ refers to acquiring the complete CSI only in the downlink by exhaustive beam scanning. This case would suffer from a long $T_{\rm DL}$ (about $0.6$ sec as shown in Fig. \ref{fig:CC}), which may degrade the net spectral efficiency. 

Fig. \ref{fig:DL} shows the accuracy of the downlink CE stage by investigating the probability of grouping failure versus the total downlink transmission power at the BS. The indices obtained by {\bf Algorithm 1} are compared with the oracle LoS angles ${\bm \mu}^{\rm LoS}$ and ${\bm \nu}^{\rm LoS}$ to decide whether the downlink CE succeeds or not. Three different types of RISs, i.e., critically-spaced RIS ($d=\lambda/2$), ultra-dense RIS ($d < \lambda/2$) and CMS ($d \to 0$) are considered. 
We observe that the CMS provides the best performance among the three types of RISs. {The performance of ultra-dense RISs is significantly better than traditional critically-spaced RISs, since ultra-dense RISs have larger effective reflection area and smaller sidelobes (i.e., higher array gain in the mainlobe) which bring a better receive signal-to-noise ratio (SNR).
Fig. \ref{fig:DL} also reveals that the performance of a practical ultra-dense RIS can well approach that of an ideal (but unrealistic) CMS. A practical $d = \lambda/8$ renders only a minor performance gap compared with the CMS, so it would be sufficient to treat the ultra-dense RIS as the holographic RIS.}
Moreover, it is observed that the downlink CE performance improves as the number of groups $\{G_x,G_y\}$ increases. This is because with a larger $\{G_x,G_y\}$ the SBF beam pattern has a narrower passband, which enhances the amplitude of the beam pattern (i.e., array gain) in the passband and yields a better receive SNR.
\begin{figure}[t]
%\captionsetup{font={footnotesize, color={red}}, name = {Fig.}, labelsep = period}
\centering
\includegraphics[height=2.4in,width=3in]{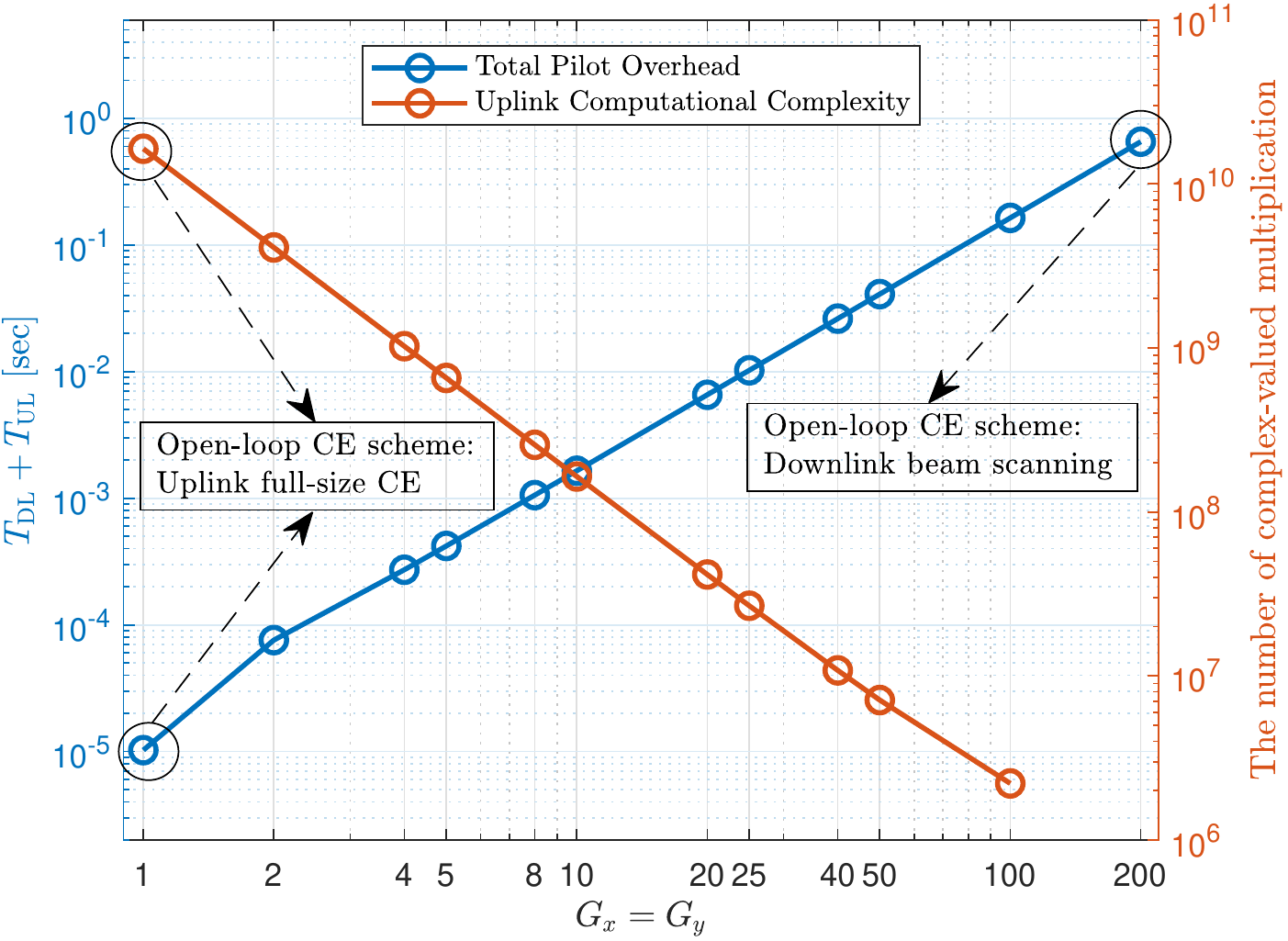}
%\vspace{-7mm}
\caption{The trade-off between the total pilot overhead and computational complexity of the proposed closed-loop CE scheme, where $N_{\rm P} = 40$ is considered.}
\label{fig:CC}
\end{figure}
\begin{figure}[t]
%\captionsetup{font={footnotesize}, name = {Fig.}, labelsep = period}
\centering
\includegraphics[width=3in]{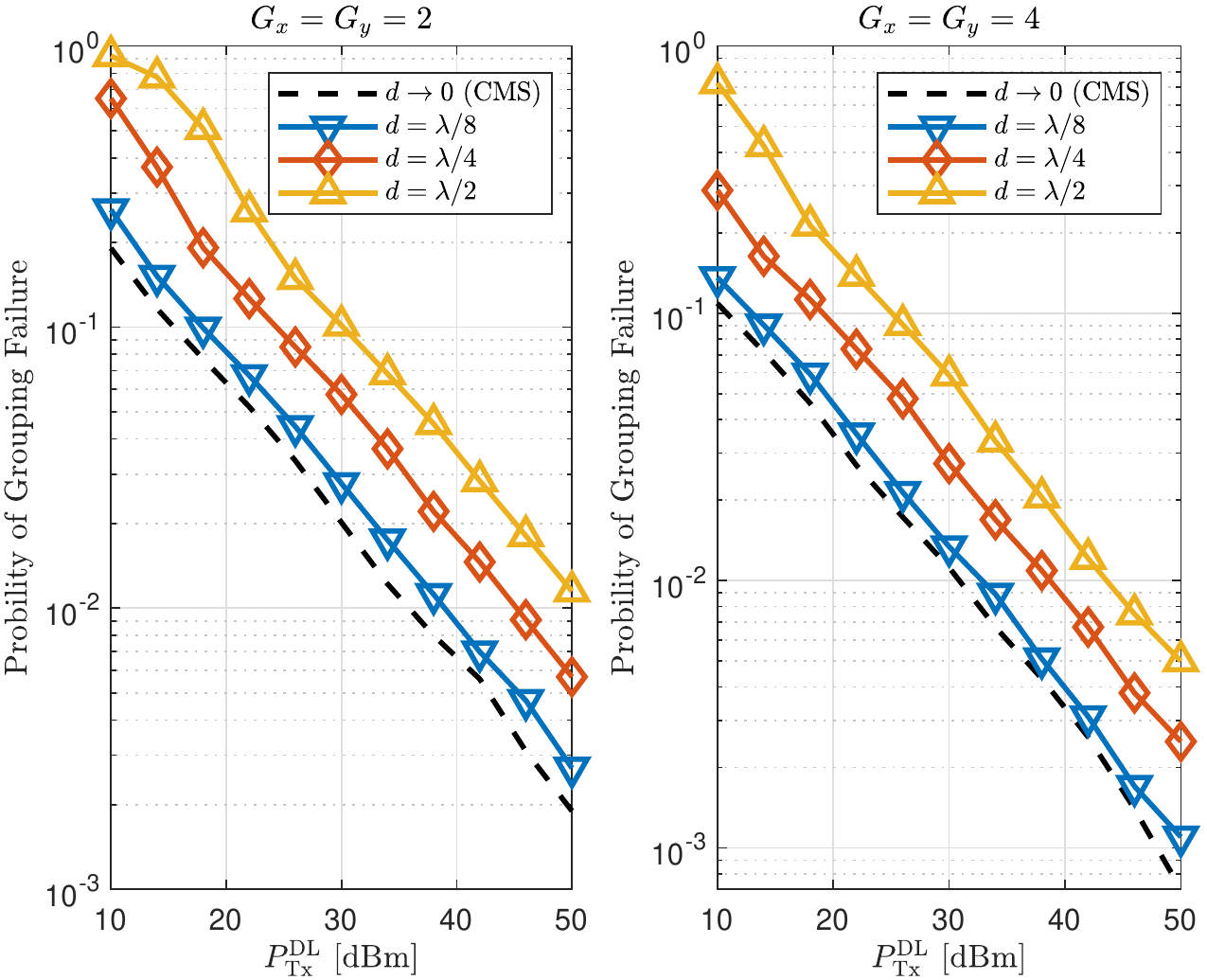}
%\vspace{-7mm}
\caption{The downlink CE performance versus the BS transmit power. Different types of RISs and different values of $\{G_x,G_y\}$ are compared.}
\label{fig:DL}
%\vspace{-12mm}
\end{figure}

Next, we investigate the performance of the uplink finer-grained CE stage that is obtained by using {\bf Algorithm 2}, and compare it with some existing schemes. We set $G_x = G_y = 10$, i.e., $B_x=B_y=20$. In Figs. \ref{fig:diffP}-\ref{fig:diffPtx}, the normalized mean square error (NMSE) is adopted as the performance metric of interest, which is given by $\mathbb{E}\left\{ {\frac{{\left\| {{\bf{\hat H}}_u^{{\rm{AdDd}}} - {\bf{H}}_u^{{\rm{AdDd}}}} \right\|_F^2}}{{\left\| {{\bf{H}}_u^{{\rm{AdDd}}}} \right\|_F^2}}} \right\}$.
Fig. \ref{fig:diffP} depicts the uplink NMSE performance versus the uplink pilot overhead $N_{\rm P}$ and the different DSC allocation schemes as illustrated in Fig. \ref{fig:UW} for $N_{\rm UE}=4$. It can be observed that the random DSC allocation scheme achieves better CE performance, while the uniform allocation scheme fails to work properly. If the random DSC allocation scheme is considered, in addition, a sufficiently high CE accuracy can be ensured even with a low compression ratio $(N_{\rm P}N_{\rm used})/(B_xB_yN_{\rm CP})$ in the range $\{0.0063,0.0187,0.0313,0.0437,0.0563,0.0688\}$ in Fig. \ref{fig:diffP}.
Due to its superior performance, the random DSC allocation strategy is adopted to obtain the rest of the results.
\begin{figure}[t]
%\captionsetup{font={footnotesize}, name = {Fig.}, labelsep = period}
\centering
\includegraphics[width=3in]{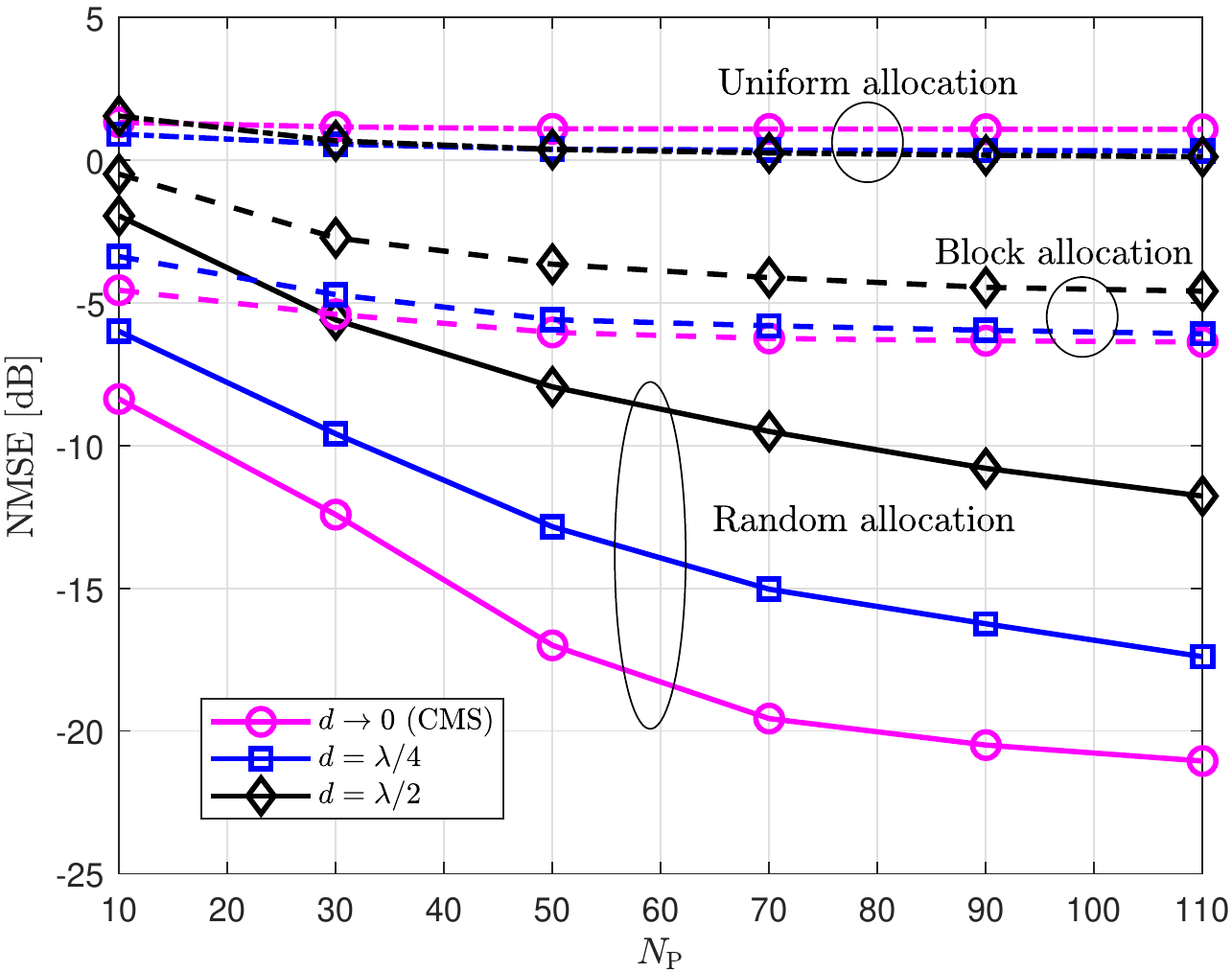}
%\vspace{-7mm}
\caption{NMSE performance versus both the uplink pilot overhead $N_{\rm P}$ and DSC allocation schemes when $N_{\rm UE} = 4$ and $P_{\rm Tx}^{\rm UL}$ = $23$ dBm.}
\label{fig:diffP}
\end{figure}
\begin{figure}[t]
\centering
\includegraphics[width=3in]{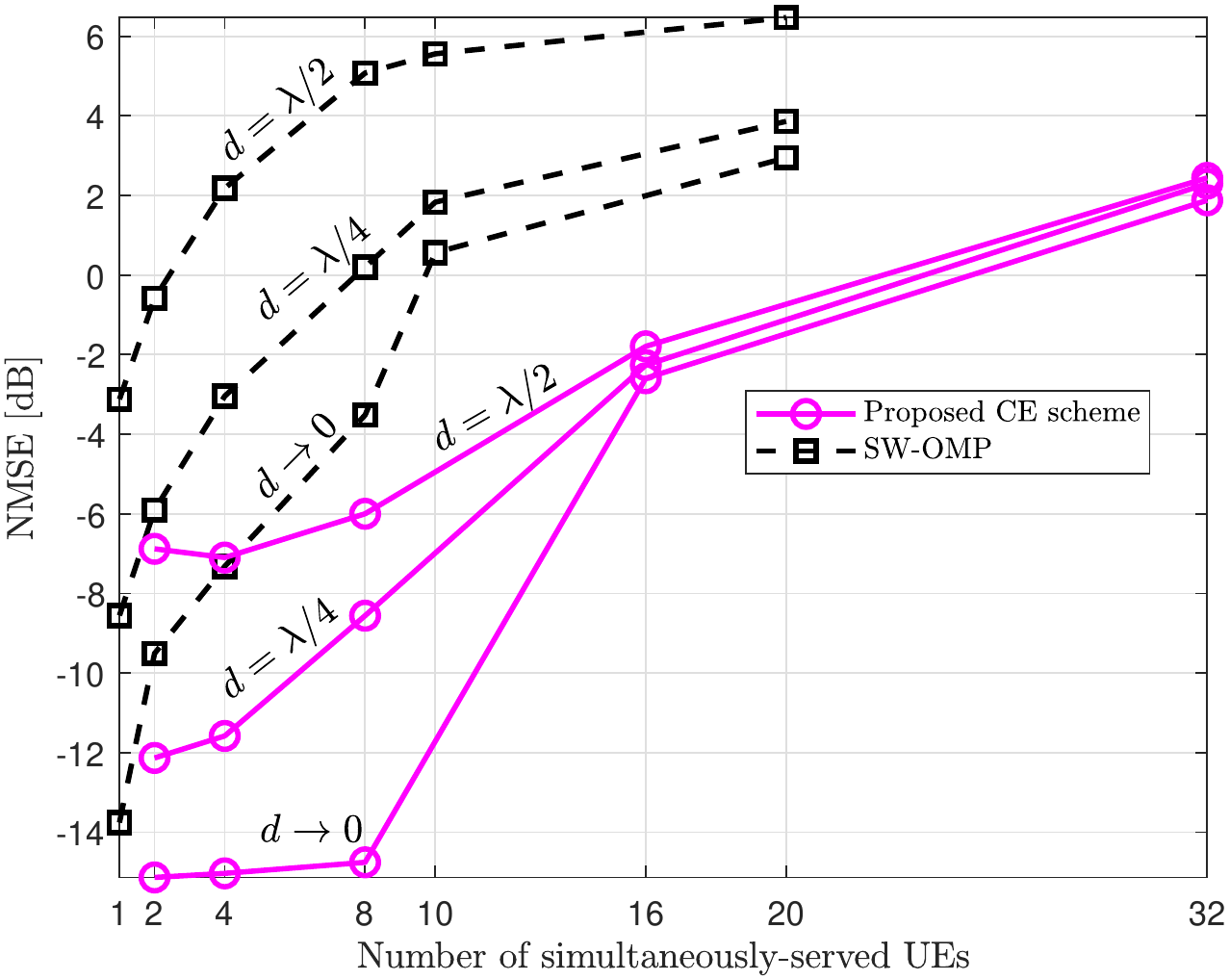}
%\vspace{-7mm}
\caption{NMSE performace veruse the number of UEs $N_{\rm UE}$ when $N_{\rm P} = 40$ and $P_{\rm Tx}^{\rm UL}$ = $23$ dBm. The SW-OMP algorithm \cite{CE1} is adopted as the benchmark.}
\label{fig:diffK}
%\vspace{-12mm}
\end{figure}

In Fig. \ref{fig:diffK}, we plot the NMSE performance of the proposed CE scheme as a function of the number of simultaneously-served UEs $N_{\rm UE}$ for $N_{\rm P}=40$.
{As the benchmark, we adopt the simultaneous weighted OMP (SW-OMP) \cite{CE1} algorithm which estimates the channels in the frequency domain (rather than the delay domain) via well-determined measurements. When SW-OMP is considered, $N_{\rm P}$ UWs are equally divided into $N_{\rm UE}$ parts, each of which is dedicated for one UE. Therefore, only $N_{\rm P}/N_{\rm UE}$ UWs are available for each UE to conduct CE.} It can be observed that the proposed CE scheme outperforms the considered benchmark even if a larger number of UEs are served simultaneously. This is obtained because the proposed CE scheme exploits the dual sparsity of THz MIMO channels in both the angular domain and delay domain, while the CE scheme based on SW-OMP only utilizes the sparsity in the angular domain.
Fig. \ref{fig:diffPtx} compares the NMSE performance of different CE schemes against the uplink transmission power $P_{\rm Tx}^{\rm UL}$ for $N_{\rm UE}=4$. As a benchmark, we consider the LS estimator \cite{FFFCE,FSFCE} with well-determined measurements in both the angular domain and delay domain, which requires a large number $N_{\rm P} = B_xB_yN_{\rm UE} = 1600$ of UWs as pilot signals. By leveraging the dual sparsity of THz MIMO channels in both the angular domain and delay domain, the proposed CS-based CE scheme outperforms the LS estimator even if a much smaller number of pilot signals ($N_{\rm P} = \{40,80\}$) is used. As a second benchmark scheme, we analyze the open-loop CE scheme that implements the uplink CE stage without using the downlink grouping. In such a case, the UEs do not have any prior information of the coarsely-estimated LoS angles and thus the NBS beamforming in \eqref{equ:temp4} is unavailable. Instead, random pilot signals \cite{GaoCL,GaoTSP} are employed at the UEs to realize an omni-directional beam pattern. Compared with the NBS beamforming towards the coarsely-estimated LoS direction, the omni-directional beam pattern disperses the transmit energy towards several directions. This results in the poor CE performance, as demonstrated in Fig. \ref{fig:diffPtx}.

\begin{figure}[t]
%\vspace{-5mm}
%\captionsetup{font={footnotesize}, name = {Fig.}, labelsep = period}
\centering
\includegraphics[width=3in]{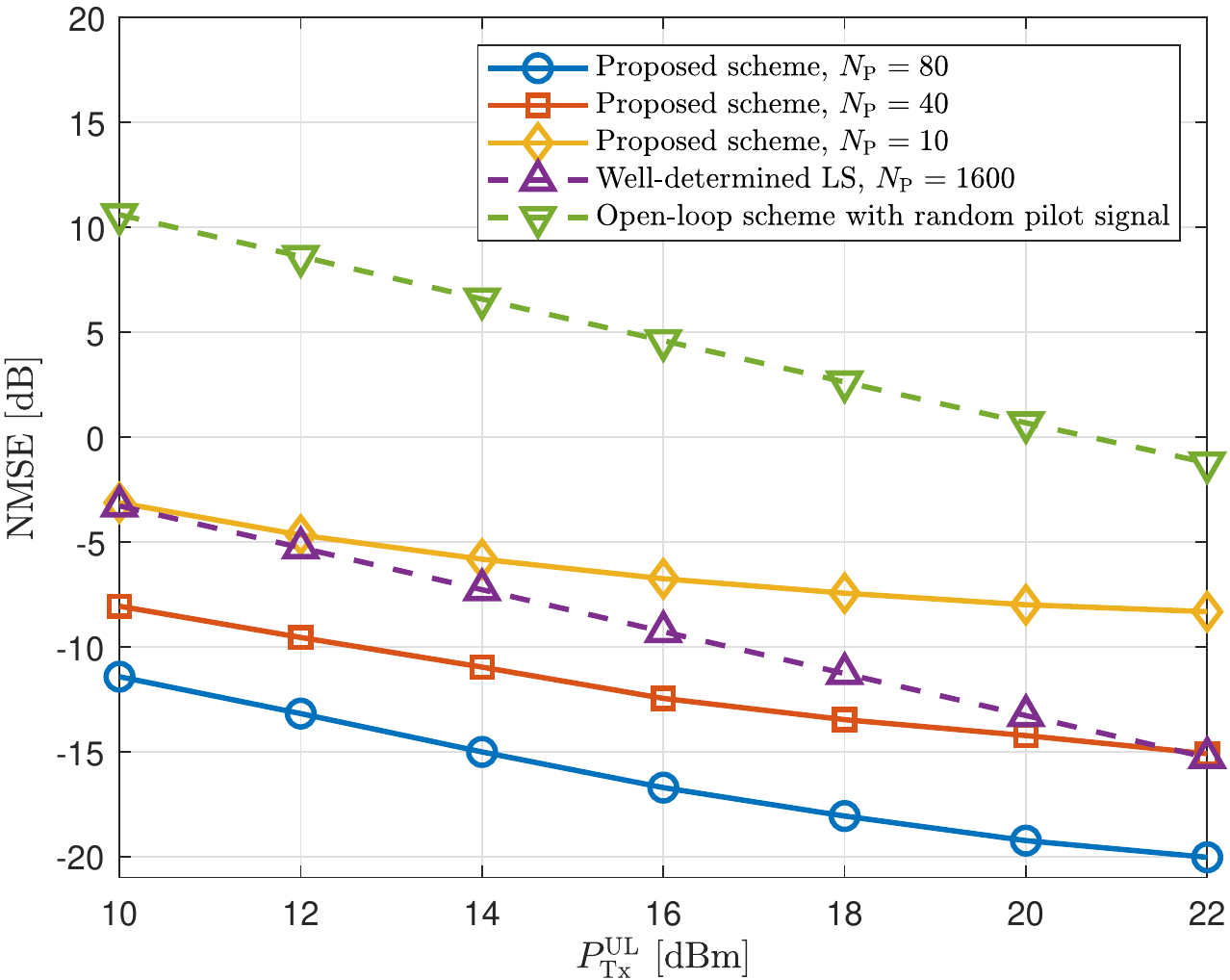}
%\vspace{-7mm}
\caption{NMSE performance of different CE schemes against the uplink transmission power $P_{\rm Tx}^{\rm UL}$ when $N_{\rm UE}=4$. The CMS is considered.}
\label{fig:diffPtx}
\end{figure}
\begin{figure}[t]
\centering
\includegraphics[width=3in]{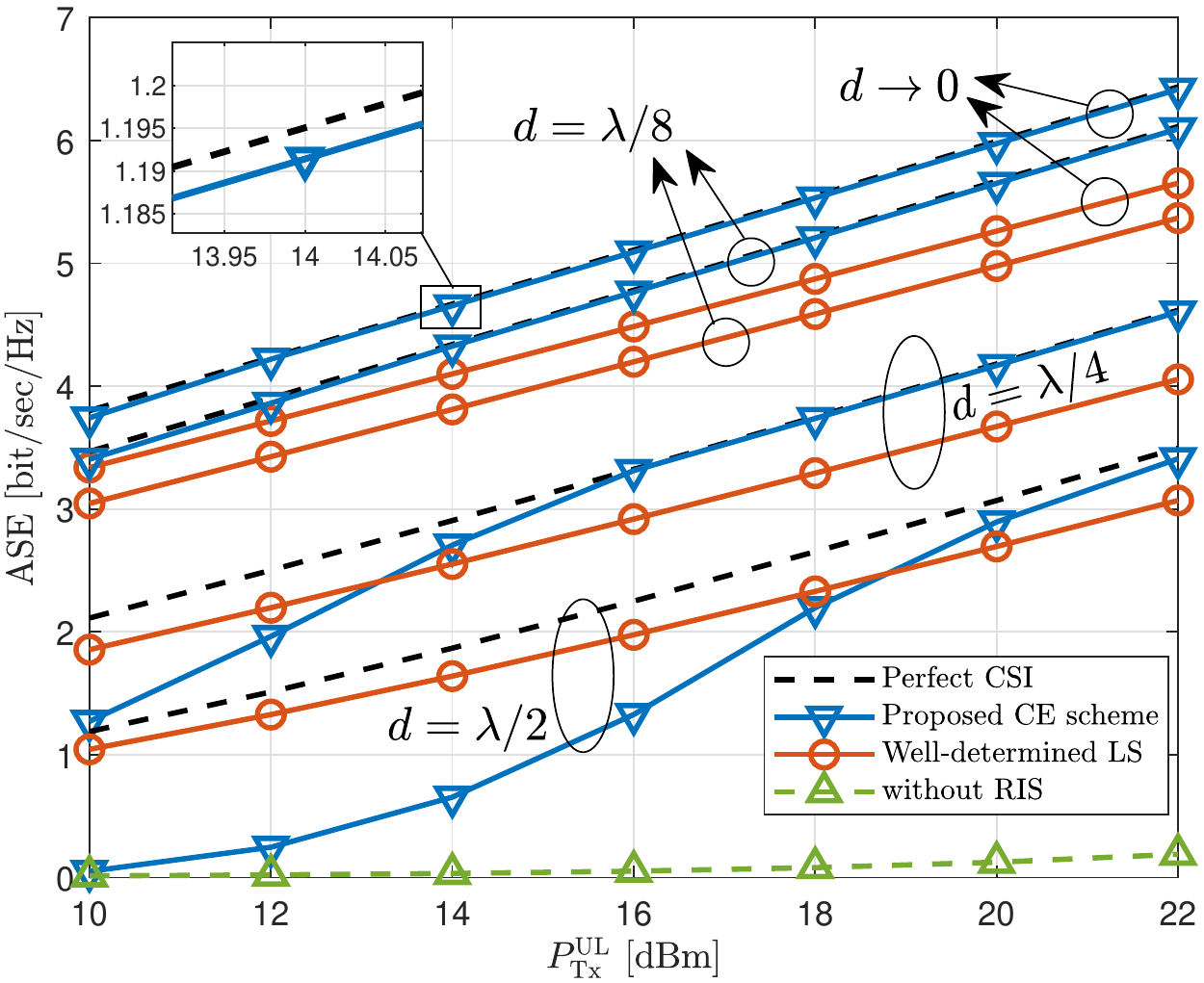}
%\vspace{-7mm}
\caption{Comparison of ASE performance when $N_{\rm UE}=4$ and $T_{\rm coh}=5$ ms. $N_{\rm P}=40$ is considered for the proposed CS-based CE scheme.}
\label{fig:ASE}
%\vspace{-12mm}
\end{figure}

{In Fig. \ref{fig:ASE}, we compare the average spectral efficiency (ASE) performance that is achieved by different schemes in order to further evaluate the accuracy of the CE schemes.} The ASE of the $u$-th UE is defined as
%\vspace{-2mm}
\begin{align}
& \text{ASE} = \left( {1 - \frac{{{T_{{\rm{DL}}}} + {T_{{\rm{UL}}}}}}{{{T_{{\rm{coh}}}}}}} \right) \nonumber \\
& \times \mathbb{E} \left\{\frac{1}{K} {\sum\limits_{k = 1}^K {{{\log }_2}\left[ {1 + P_{\rm Tx}^{\rm UL}{\left| {h^{{\rm{Fd}}}_{u,k}} \right|}^2/({K{\sigma _{\rm{n}}^2}})} \right]} } \right\} \ \text{[bit/sec/Hz],}
\vspace{-2mm}
\end{align}
where ${h_{u,k}^{{\rm{Fd}}}}$ is the frequency-domain effective baseband channel of the $k$-th subcarrier by using the NBS beamforming towards the estimated LoS direction of the RIS.
The NBS beamforming towards the oracle LoS direction is illustrated as an upper-bound (i.e., perfect CSI without uplink pilot overhead). We observe that the ASE obtained by the proposed CE scheme has good tightness with the upper-bound, while the ASE obtained by the well-determined LS estimator is worse due to the long time duration that is required for pilot transmission. {It is also observed that the ASE improves as the element spacing $d$ decreases, and it can approach the performance of the ideal CMS with practical $d$ (e.g., $d = \lambda/8$), which further verify that the ultra-dense RIS is a good realization of the holographic RIS.}
Further, the ASE that is obtained in the absence of an RIS is illustrated. In this case, we assume that the UE communicates with the BS via the NLoS link with perfect CSI and no pilot overhead.
We observe that the virtual LoS link provided by an RIS in THz massive MIMO systems significantly increases the ASE.

{In addition, we investigate the impact of quantization error, i.e., only a finite number of phase shifts can be realized in practice, on the performance of the proposed schemes, similar to \cite{THzIRS,D}. {Note that some state-of-the-art architectures of reflecting element, such as semiconductor diodes in \cite{Model}, are not suited to much higher frequencies (e.g., THz) \cite{APL}. Therefore, the analysis of the impact of the interplay between the phase and the amplitude of the reflecting elements of the THz holographic RIS is postponed to a future search contribution.}
We assume that the phase of each reflecting element of the RIS is quantized by using $B$ quantization bits. This implies that the phases of the reflection coefficients are drawn from the finite set ${\cal{B}} = \left\{ {{{2}}\pi \frac{{1 - {2^{B - 1}}}}{{{{{2}}^B}}},{{2}}\pi \frac{{2 - {2^{B - 1}}}}{{{{{2}}^B}}},...,{{2}}\pi \frac{{{2^{B - 1}}}}{{{{{2}}^B}}}} \right\}$. As far as the proposed beamforming designs (NBS and SBF) are concerned, this correspond to quantizing the reflection coefficients $\Phi(m,n)$ in \eqref{equ:PhiStr} or \eqref{equ:PhiSBF} as ${\Phi ^{{\rm{Q}}}}(m,n) = \left| {\Phi (m,n)} \right|{e^{j{\phi _{m,n}}}}$, where $\phi_{m,n}\in {\cal B}$ is the phase element in ${\cal B}$ that is closest to the phase of $\Phi(m,n)$. Then, the corresponding beam patterns are computed via \eqref{equ:g1} by replacing $\Phi(m,n)$ with $\Phi^{\rm Q}(m,n)$.
In Fig. \ref{fig:SBFQuan}, we illustrate some examples of the SBF beam patterns as a function of the number of quantization bits. If $B = 1$, the quantization error is too large and the accuracy of the beam patterns degrades significantly. If $B = \{2,3\}$, on the other hand, the obtained beam patterns after quantization exhibit the desired band-pass properties, with an acceptable degradation compared with the ideal case without quantization error ($B = \infty$).
In Fig. \ref{fig:NBSQuan}, we evaluate the impact of quantization error on the ASE.
The numerical results confirm that the ASE performance degradation is not significant if $B = \{2,3\}$. In addition, this holds true for any values $d \le \lambda/2$.
}
\begin{figure}[t]
%\captionsetup{font={footnotesize}, name = {Fig.}, labelsep = period}
\centering
\includegraphics[width=3in]{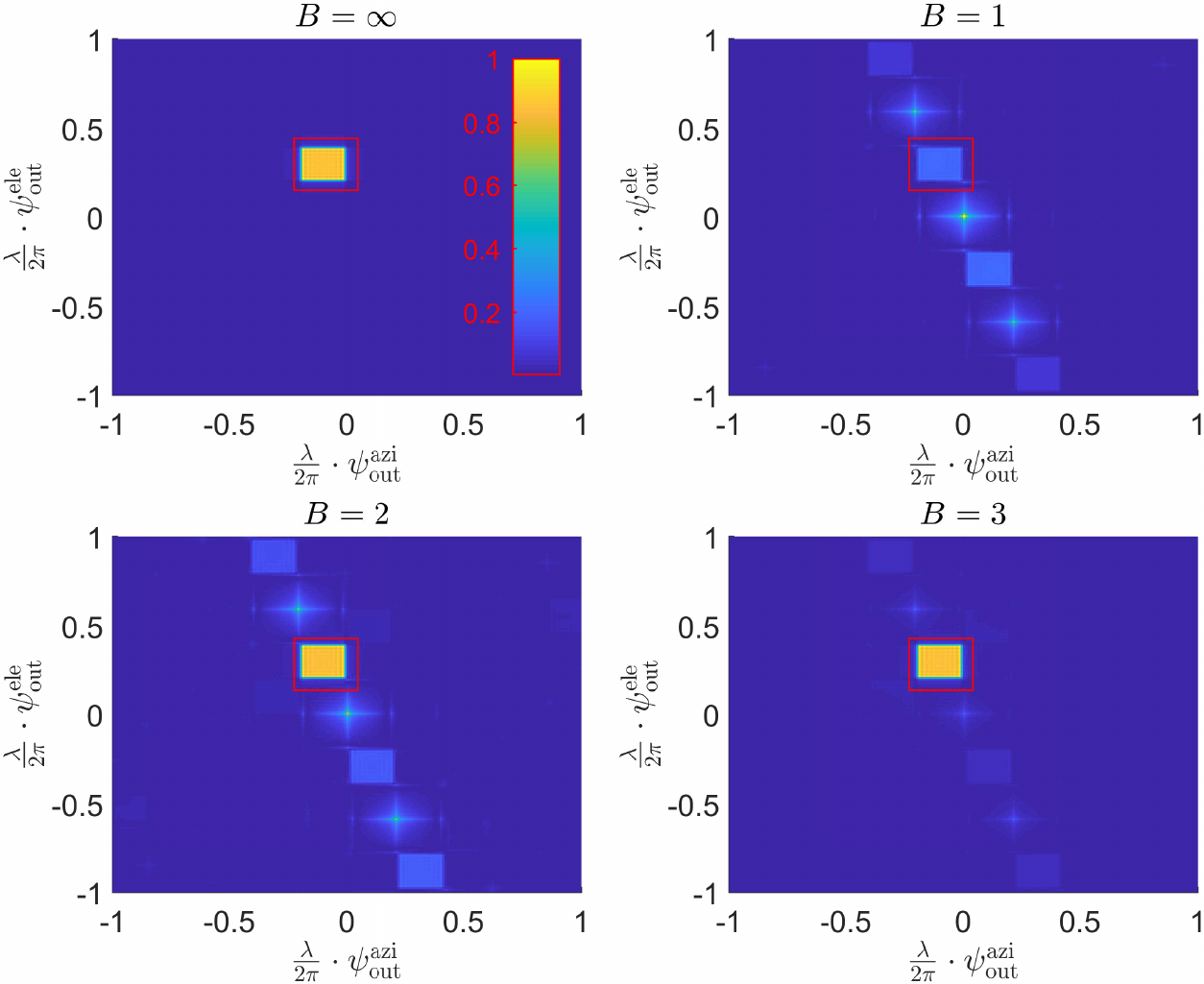}%
%\vspace{-7mm}
\caption{SBF beam patterns used in the proposed CE scheme with different quantization bits. The ultra-dense RIS with elements spacing $d = \lambda/4$ is considered.}
\label{fig:SBFQuan}
\end{figure}
\begin{figure}[t]
\centering
\includegraphics[width=3in]{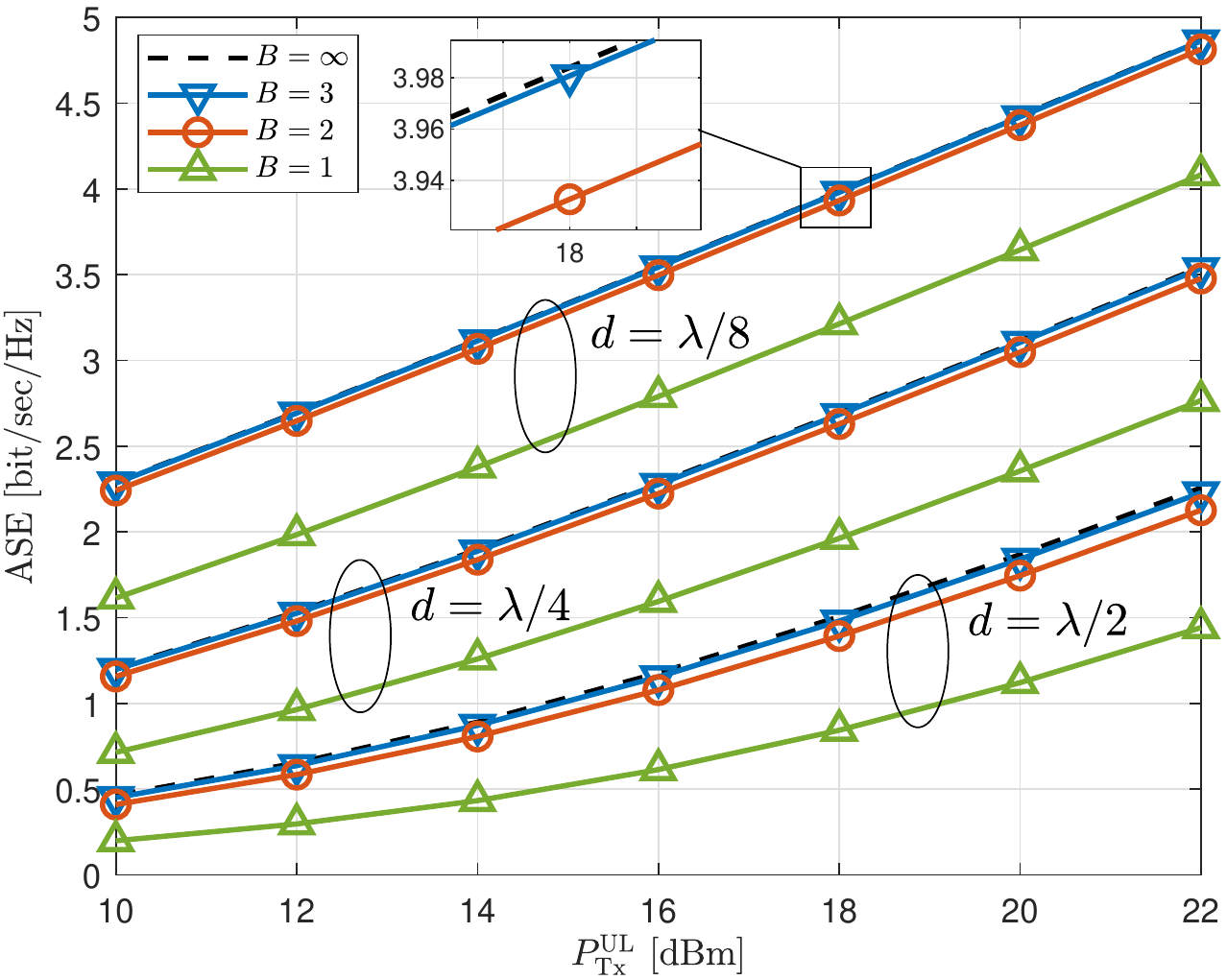}
%\vspace{-7mm}
\caption{ASE performance with different quantization bits under the same scenario as Fig. \ref{fig:ASE}. Perfect CSI is considered.}
\label{fig:NBSQuan}
%\vspace{-12mm}
\end{figure}
%{{\it Remark}: It is worth noting that the cost-effective deployment of holographic RISs whose reflecting elements have a deep sub-wavelength size and inter-distance requires a large number of electronic circuits for controlling each reflecting element individually (assuming that the size of the RIS is kept fixed). Therefore, the development of economies of scale plays an important role in the possible future deployment of nearly-passive RISs.}

\vspace{-3mm}
\section{Conclusions}
%\vspace{-2mm}
Motivated by the concept of holographic communications, we studied the physical layer transmission of holographic RISs in the THz band, where a large number of sub-wavelength reconfigurable elements are densely integrated into a compact space to approach a spatially continuous aperture. We derived the beam pattern of a holographic RIS and proposed two beamforming designs that are formulated {in closed-form expressions}. Based on the proposed beamforming designs, we proposed a closed-loop broadband CE scheme for the RIS-aided THz massive MIMO systems. The proposed CE scheme encompasses a downlink grouping stage and an uplink finer-grained CE stage. In order to reduce the pilot overhead, a CS-based CE algorithm was proposed that exploits the dual sparsity of THz MIMO channels in both the angular domain and delay domain.
Simulation results showed that holographic RISs are able to outperform traditional designs based on non-holographic RISs as well as communication schemes that do not use RISs.
{Possible future research directions based on the results obtained in this paper include the practical amplitude and phase shift model of reflecting elements, the analysis of hardware impairments, the development of robust signal processing methods, the analysis of near-field communications, and the proof-of-concept validations and field experiments.
}

\vspace{-3mm}
\begin{appendices}
\section{Proof of {\bf Corollary 1}}
%\vspace{-3mm}
The result in (\ref{equ:PhiCMSNBS}) can be obtained by letting $dm = x$ and $dn = y$ in (\ref{equ:PhiStr}). As far as (\ref{equ:gCMSNBS}) is concerned, in particular, we have
\vspace{-4mm}
\begin{align}
& \quad \ \tilde g_{\rm{NBS}}\left( {{{\bm \psi} _{{\rm{out}}}},{{\bm \psi} _{{\rm{in}}}}};{\bm \psi}_{\rm opt} \right) \nonumber \\
& \propto \lim_{d \to 0} g_{\rm{NBS}}\left( {{{\bm \psi} _{{\rm{out}}}},{{\bm \psi} _{{\rm{in}}}}};{\bm \psi}_{\rm opt} \right)  \nonumber \\
& = \lim_{d \to 0} {e^{j\frac{{d - {A_x}}}{2}{\Delta _x}}}{e^{j\frac{{d - {A_y}}}{2}{\Delta _y}}} {\Xi _{{N_x}}}\left[ {d} \left( {k_{\rm{opt}} -\Delta_x} \right) \right] {\Xi _{{N_y}}}\left[ {d} \left( {l_{\rm{opt}} - \Delta_y} \right) \right] \nonumber \\
& = {e^{-j\frac{{A_x}}{2}{\Delta _x}}}{e^{-j\frac{A_y}{2}{\Delta _y}}} \lim_{d \to 0} {\Xi _{{N_x}}}\left[ {d} \left( {k_{\rm{opt}} -\Delta_x} \right) \right] {\Xi _{{N_y}}}\left[ {d} \left( {l_{\rm{opt}} - \Delta_y} \right) \right]
 \text{.}
 \vspace{-2mm}
\end{align}
Given the symmetry, we only need to prove $\mathop {\lim }\limits_{d \to 0} {\Xi _{{N_x}}}\left[ {d\left( {k_{\rm{opt}} - {{{\Delta _x}}}} \right)} \right] = {\rm{sinc}}\left[ {\frac{A_x}{2} \left( {k_{\rm{opt}} - {{{\Delta _x}}}} \right)} \right]$ in order to obtain (\ref{equ:gCMSNBS}).
When $k_{\rm{opt}} \ne \Delta_x$, based on the definition of the function $\Xi_{N_x}$, we have
\begin{align}
\label{equ:A1}
& \quad \ \mathop {\lim }\limits_{d \to 0} {\Xi _{{N_x}}}\left[ {d\left( {k_{\rm{opt}} - {{{\Delta _x}}}} \right)} \right] \nonumber \\
& = \mathop {\lim }\limits_{d \to 0} \frac{{\sin \left[ {\frac{{{N_x}d}}{2}\left( {{k_{{\rm{opt}}}} - {\Delta _x}} \right)} \right]}}{{{N_x}\sin \left[ {\frac{d}{2}\left( {{k_{{\rm{opt}}}} - {\Delta _x}} \right)} \right]}} \nonumber \\
& \mathop  = \limits^{{\rm{(a)}}} \mathop {\lim }\limits_{d \to 0} \frac{{d\sin \left[ {\frac{{{A_x}}}{2}\left( {{k_{{\rm{opt}}}} - {\Delta _x}} \right)} \right]}}{{{A_x}\sin \left[ {\frac{d}{2}\left( {{k_{{\rm{opt}}}} - {\Delta _x}} \right)} \right]}}\nonumber \\
& \mathop  = \limits^{{\rm{(b)}}} \mathop {\lim }\limits_{d \to 0} \frac{{\sin \left[ {\frac{{{A_x}}}{2}\left( {{k_{{\rm{opt}}}} - {\Delta _x}} \right)} \right]}}{{\frac{{{A_x}}}{2}\left( {{k_{{\rm{opt}}}} - {\Delta _x}} \right)\cos \left[ {\frac{d}{2}\left( {{k_{{\rm{opt}}}} - {\Delta _x}} \right)} \right]}} \nonumber \\
& = \frac{{\sin \left[ {\frac{{{A_x}}}{2}\left( {{k_{{\rm{opt}}}} - {\Delta _x}} \right)} \right]}}{{\frac{{{A_x}}}{2}\left( {{k_{{\rm{opt}}}} - {\Delta _x}} \right)}}\nonumber \\
& = {\rm{sinc}}\left[ {\frac{{{A_x}}}{2}\left( {{k_{{\rm{opt}}}} - {\Delta _x}} \right)} \right] \text{,}
\end{align}
%\vspace{-1mm}
where the equality (a) follows from $N_x d = A_x$, and the equality (b) is obtained by applying the {\it De l'H{\^o}pital} rule with respect to $d$. In addition, when $k_{\rm{opt}} = \Delta_x$, $\mathop {\lim }\limits_{d \to 0} {\Xi _{{N_x}}}{\left( d \cdot 0 \right)} = 1 = {\rm{sinc}}\left( \frac{A_x}{2} \cdot 0 \right)$. Thus, we obtain (\ref{equ:A1}). This completes the proof of {\bf Corollary 1}.
\end{appendices}


\vspace{-3mm}
\begin{thebibliography}{99}
%\vspace{-2mm}
\bibitem{THz1}H. Elayan, O. Amin, B. Shihada, R. M. Shubair, and M.-S. Alouini,
``Terahertz band: The last piece of RF spectrum puzzle for communication systems,''
{\it IEEE Open J. Commun. Soc.}, vol. 1, pp. 1-32, 2020.

\bibitem{THz2}T. S. Rappaport {\it et al.},
``Wireless communications and applications above 100 GHz: Opportunities and challenges for 6G and beyond,''
{\it IEEE Access}, vol. 7, pp. 78729-78757, 2019.

%\bibitem{THz3}C. Han, A. O. Bicen, and I. F. Akyildiz,
%``Multi-ray channel modeling and wideband characterization for wireless communications in the terahertz band,''
%{\it IEEE Trans. Wireless Commun.}, vol. 14, no. 5, pp. 2402-2412, May 2015.

\bibitem{DifSca}C. Jansen {\it et al.},
``Diffuse scattering from rough surfaces in THz communication channels,''
{\it IEEE Trans. Terahertz Sci. Technol.}, vol. 1, no. 2, pp. 462-472, Nov. 2011.

\bibitem{MM}C. Lin and G. Y. Li,
``Indoor terahertz communications: How many antenna arrays are needed?,'' 
{\it IEEE Trans. Wireless Commun.}, vol. 14, no. 6, pp. 3097-3107, June 2015.

\bibitem{UMM2}H. Sarieddeen, M.-S. Alouini, and T. Y. Al-Naffouri,
``Terahertz-band ultra-massive spatial modulation MIMO,''
{\it IEEE J. Sel. Areas Commun.}, vol. 37, no. 9, pp. 2040-2052, Sep. 2019.

{\bibitem{THzModel1}L. Yan, C. Han and J. Yuan,
``A dynamic array-of-subarrays architecture and hybrid precoding algorithms for terahertz wireless communications,''
{\it IEEE J. Sel. Areas Commun.}, vol. 38, no. 9, pp. 2041-2056, Sep. 2020.}

{\bibitem{THzModel2}H. Yuan, N. Yang, K. Yang, C. Han and J. An,
``Hybrid beamforming for terahertz multi-carrier systems over frequency selective fading,''
{\it IEEE Trans. Commun.}, vol. 68, no. 10, pp. 6186-6199, Oct. 2020.}

\bibitem{MDR}M. Di Renzo {\it et al.}, 
``Smart radio environments empowered by reconfigurable intelligent surfaces: How it works, state of research, and road ahead,''
{\it IEEE J. Sel. Areas Commun.}, vol. 38, no. 11, pp. 2450-2525, Nov. 2020.

\bibitem{Europe}M. Di Renzo {\it et al.}, 
``Smart radio environments empowered by reconfigurable AI meta-surfaces: An idea whose time has come,''
{\it EURASIP J. Wireless Commun. Netw.}, vol. 2019, p. 129, May 2019.

\bibitem{Fellow3}E. Basar, M. Di Renzo, J. De Rosny, M. Debbah, M.-S. Alouini, and R. Zhang, 
``Wireless communications through reconfigurable intelligent surfaces,''
{\it IEEE Access}, vol. 7, pp. 116753-116773, Sep. 2019.

\bibitem{RZhang1}Q. Wu and R. Zhang, 
``Towards smart and reconfigurable environment: Intelligent reflecting surface aided wireless network,''
{\it IEEE Commun. Mag.}, vol. 58, no. 1, pp. 106-112, Jan. 2020.

{\bibitem{UAV}Z. Wei, Y. Cai, Z. Sun, D. W. Kwan Ng, J. Yuan, M. Zhou, and L. Sun,
``Sum-rate maximization for IRS-assisted UAV OFDMA communication systems,''
to appear in {\it IEEE Trans. Wireless Commun.}}

{\bibitem{UAV2}S. Li, B. Duo, X. Yuan, Y. Liang and M. Di Renzo,
``Reconfigurable intelligent surface assisted UAV communication: Joint trajectory design and passive beamforming,'' 
{\it IEEE Wireless Commun. Lett.}, vol. 9, no. 5, pp. 716-720, May 2020.}

{\bibitem{SWIPT}D. Xu, X. Yu, V. Jamali, D. W. Kwan Ng, and R. Schober,
``Resource allocation for large IRS-assisted SWIPT systems with non-linear energy harvesting model,'' Oct. 2020.
[Online]. Available: arXiv:2010.00846v1.}

{\bibitem{CR}D. Xu, X. Yu, Y. Sun, D. W. K. Ng and R. Schober,
``Resource allocation for IRS-assisted full-duplex cognitive radio systems,'' 
{\it IEEE Trans. Commun.}, vol. 68, no. 12, pp. 7376-7394, Dec. 2020.}

\bibitem{B}M. Di Renzo {\it et al.}, 
``Reconfigurable intelligent surfaces vs. relaying: Differences, similarities, and performance comparison,'' 
{\it IEEE Open J. Commun. Soc.}, vol. 1, pp. 798-807, 2020.

\bibitem{C}G. Zhou, C. Pan, H. Ren, K. Wang, M. Di Renzo, and A. Nallanathan,
``Robust beamforming design for intelligent reflecting surface aided MISO communication systems,''
{\it IEEE Wireless Commun. Lett.}, vol. 9, no. 10, pp. 1658-1662, Oct. 2020.

{\bibitem{D}X. Qian {\it et al.}, 
%M. Di Renzo, J. Liu, A. Kammoun and M.-S. Alouini,
``Beamforming through reconfigurable intelligent surfaces in single-user MIMO systems: SNR distribution and scaling laws in the presence of channel fading and phase noise,''
{\it IEEE Wireless Commun. Lett.}, vol. 10, no. 1, pp. 77-81, Jan. 2021.}

\bibitem{Y}F. H. Danufane {\it et al.},
``On the path-loss of reconfigurable intelligent surfaces: An approach based on Green's Theorem applied to vector fields,'' [Online]. Available: arXiv:2007.13158v2.

\bibitem{RZhang2}Q. Wu and R. Zhang, 
``Intelligent reflecting surface enhanced wireless network via joint active and passive beamforming,''
{\it IEEE Trans. Wireless Commun.}, vol. 18, no. 11, pp. 5394-5409, Nov. 2019.

\bibitem{MSlim}Q.-U.-A. Nadeem {\it et al.}, 
``Asymptotic max-min SINR analysis of reconfigurable intelligent surface assisted MISO systems,''
{\it IEEE Trans. Wireless Commun.}, vol. 19, no. 12, pp. 7748-7764, Dec. 2020.

%\bibitem{KYing}K. Ying, Z. Gao, S. Lyu, Y. Wu, H. Wang, and M.-S. Alouini,
%``GMD-based hybrid beamforming for large reconfigurable intelligent surface assisted millimeter-wave massive
%MIMO,''
%{\it IEEE Access}, vol. 8, pp. 19530-19539, 2020.

\bibitem{CHuang}C. Huang, R. Mo, and C. Yuen,
``Reconfigurable intelligent surface assisted multiuser MISO systems exploiting deep reinforcement learning,''
{\it IEEE J. Sel. Areas Commun.}, vol. 38, no. 8, pp. 1839-1850, Aug. 2020.

{\bibitem{Reviewer}M. Najafi, V. Jamali, R. Schober and H. Vincent Poor,
``Physics-based modeling and scalable optimization of large intelligent reflecting surfaces,''
to appear in {\it IEEE Trans. Commun.}}

{\bibitem{Model}S. Abeywickrama, R. Zhang, Q. Wu and C. Yuen,
``Intelligent reflecting surface: Practical phase shift model and beamforming optimization," 
{\it IEEE Trans. Commun.}, vol. 68, no. 9, pp. 5849-5863, Sep. 2020.}

{\bibitem{Model2}G. Gradoni and M. Di Renzo,
``End-to-end mutual coupling aware communication model for reconfigurable intelligent surfaces: An electromagnetic-compliant approach based on mutual impedances,''
to appear in {\it IEEE Wireless Commun. Lett.}}

\bibitem{FFFCE}D. Mishra and H. Johansson, 
``Channel estimation and low-complexity beamforming design for passive intelligent surface assisted MISO wireless energy transfer,'' 
in {\it Proc. 2019 IEEE Int. Conf. Acoustics, Speech and Signal Processing (ICASSP)}, Brighton, United Kingdom, May 2019.

\bibitem{FSFCE}B. Zheng and R. Zhang,
``Intelligent reflecting surface-enhanced OFDM: Channel estimation and reflection optimization,''
{\it IEEE Wireless Commun. Lett.}, vol. 9, no. 4, pp. 518-522, April 2020.

\bibitem{XYuan}Z.-Q. He and X. Yuan, 
``Cascaded channel estimation for large intelligent metasurface assisted massive MIMO,''
{\it IEEE Wireless Commun. Lett.}, vol. 9, no. 2, pp. 210-214, Feb. 2020.

\bibitem{My1}Z. Wan, Z. Gao, and M.-S. Alouini,
``Broadband channel estimation for intelligent reflecting
surface aided mmWave massive MIMO systems,''
in {\it Proc. 2020 IEEE Int. Conf. Commun. (ICC)}, Dublin, Ireland, 2020, pp. 1-6.

\bibitem{AA}A. Taha, M. Alrabeiah, and A. Alkhateeb,
``Deep learning for large intelligent surfaces in millimeter wave and massive MIMO systems,'' 
in {\it Proc. 2019 IEEE Global Communications Conference (GLOBECOM)}, Waikoloa, HI, USA, 2019.

\bibitem{SLiu}S. Liu {\it et al.},
``Deep denoising neural network assisted compressive channel estimation for mmWave intelligent reflecting surfaces,'' 
{\it IEEE Trans. Veh. Technol.}, vol. 69, no. 8, pp. 9223-9228, Aug. 2020.

{\bibitem{THzIRS}X. Ma {\it et al.},
``Joint channel estimation and data rate maximization for intelligent reflecting surface assisted terahertz MIMO communication systems,''
{\it IEEE Access}, vol. 8, pp. 99565-99581, 2020.}

\bibitem{THzIRS2}J. Qiao and M.-S. Alouini,
``Secure transmission for intelligent reflecting surface-assisted mmWave and terahertz systems,'',
{\it IEEE Wireless Commun. Lett.}, vol. 9, no. 10, pp. 1743-1747, Oct. 2020.

\bibitem{Holo1}A. Pizzo, T. L. Marzetta, and L. Sanguinetti,
``Spatially-stationary model for holographic MIMO small-scale fading,''
{\it IEEE J. Sel. Areas Commun.}, vol. 38, no. 9, pp. 1964-1979, Sep. 2020.

\bibitem{Holo2}N. Rajatheva {\it et al.},
(2020). White paper on broadband connectivity in 6G [White paper]. (6G Research Visions, No. 10). University of Oulu. http://urn.fi/urn:isbn:9789526226798

\bibitem{Holo3}S. Hu, F. Rusek, and O. Edfors,
``Beyond massive MIMO: The potential of data transmission with large intelligent surfaces,''
{\it IEEE Trans. Signal Process.}, vol. 66, no. 10, pp. 2746-2758, May 2017.

\bibitem{Holo4}C. Huang {\it et al.},
``Holographic MIMO surfaces for 6G wireless networks: Opportunities, challenges, and trends,''
{\it IEEE Wireless Commun.}, vol. 27, no. 5, pp. 118-125, Oct. 2020.

\bibitem{korean}J. Lee, G. Gil, and Y. H. Lee,
``Channel estimation via orthogonal matching pursuit for hybrid MIMO systems in millimeter wave communications,''
{\it IEEE Trans. Commun.}, vol. 64, no. 6, pp. 2370-2386, Jun. 2016.

\bibitem{My2}Z. Wan, Z. Gao {\it et. al.},
``Compressive sensing based channel estimation for millimeter-wave full-dimensional MIMO with lens-array,''
{\it IEEE Trans. Veh. Technol.}, vol. 69, no. 2, pp. 2337-2342, Feb. 2020.
%\bibitem{BShim}J. W. Choi, B. Shim, Y. Ding, B. Rao, and D. I. Kim,
%``Compressed sensing for wireless communications: Useful tips and tricks,''
%{\it IEEE Commun. Surv. Tutor.}, vol. 19, no. 3, pp. 1527-1550, 3rd Quart., 2017.

\bibitem{UW}M. Huemer, A. Onic, and C. Hofbauer, 
``Classical and Bayesian linear data estimators for unique word OFDM,''
{\it IEEE Trans. Signal Process.}, vol. 59, no. 12, pp. 6073-6085, Dec. 2011.

{\bibitem{SJin}W. Tang {\it et al.},
''Wireless communications with reconfigurable intelligent surface: Path loss modeling and experimental measurement,''
{\it IEEE Trans. Wireless Commun.}, vol. 20, no. 1, pp. 421-439, Jan. 2021.
}

\bibitem{GaoTSP}Z. Gao, L. Dai, Z. Wang, and S. Chen,
``Spatially common sparsity based adaptive channel estimation and feedback for FDD massive MIMO,'' 
{\it IEEE Trans. Signal Process.}, vol. 63, no. 23, pp. 6169-6183, Dec. 2015.

\bibitem{GaoCL}Z. Gao {\it et al.},
``Channel estimation for millimeter-wave massive MIMO with hybrid precoding over frequency-selective fading channels,'' 
{\it IEEE Commun. Lett.}, vol. 20, no. 6, pp. 1259-1262, Apr. 2016.

\bibitem{CE1}J. P. Gonz\'alez-Coma {\it et al.},
%, J. Rodr\'iguez-Fern\'andez, N. Gonz\'alez-Prelcic, L. Castedo and R. W. Heath, 
``Channel estimation and hybrid precoding for frequency selective multiuser mmWave MIMO systems,''
{\it IEEE J. Sel. Topics Signal Process.}, vol. 12, no. 2, pp. 353-367, May 2018.

%\bibitem{CE2}J. Rodr\'iguez-Fern\'andez {\it et al.}, 
%%N. Gonz\'alez-Prelcic, K. Venugopal, and R. W. Heath,
%``Frequency-domain compressive channel estimation for frequency-selective hybrid millimeter wave MIMO systems,'' 
%{\it IEEE Trans. Wireless Commun.}, vol. 17, no. 5, pp. 2946-2960, May 2018.

\bibitem{ALiao}A. Liao, Z. Gao, H. Wang, S. Chen, M.-S. Alouini, and H. Yin,
``Closed-loop sparse channel estimation for wideband millimeter-wave full-dimensional MIMO systems,'' 
{\it IEEE Trans. Commun.}, vol. 67, no. 12, pp. 8329-8345, Dec. 2019.

{\bibitem{APL}E. Carrasco, M. Tamagnone, J. Perruisseau-Carrier,
``Tunable graphene reflective cells for THz reflectarrays and generalized law of reflection,''
{\it Appl. Phys. Lett.}, vol. 102, 2013, Art. no. 104103.}

\bibitem{ITU}ITU-R Recommendation P.676-7: Attenuation by atmospheric gases, ITU-R Std., Feb. 2007.
\end{thebibliography}
\end{document}